\documentclass[11pt]{article}

\usepackage{setspace}
\usepackage[hypertexnames=false]{hyperref}

\usepackage{xcolor}

\definecolor{navy}{RGB}{0,0,128}
\definecolor{changecol}{RGB}{0,128,0}

\usepackage{titling}
\usepackage{authblk}
\usepackage[margin=0.25in]{caption}

\usepackage{titlesec}

\usepackage[margin=0.8in]{geometry}
\usepackage{amsmath}
\usepackage{amsfonts}
\usepackage{amssymb}
\usepackage{amsthm}
\usepackage[mathscr]{euscript}
\usepackage{bbm}
\usepackage{graphicx}
\graphicspath{ {./figures/} }

\usepackage{tikz}
\usetikzlibrary{shapes}

\usepackage{enumerate}
\usepackage[square,numbers]{natbib}
\usepackage{url}

\newtheorem{theorem}{Theorem}
\newtheorem{lemma}{Lemma}

\theoremstyle{definition}
\newtheorem{remark}{Remark}
\newtheorem{assumptionCausal}{Condition}

\newtheorem{assumptionEncourage}{Condition}

\newcommand{\real}{\mathbb{R}}

\newcommand{\Q}{\mu}
\DeclareMathOperator{\expect}{E}
\DeclareMathOperator{\cexpect}{\mathbb{E}}

\DeclareMathOperator{\expit}{expit}

\newcommand{\edecision}{{\rho}}

\newcommand{\bigO}{{\mathrm{O}}}
\newcommand{\smallo}{{\mathrm{o}}}
\newcommand{\const}{{\mathscr{C}}}

\newcommand{\FR}{{\mathrm{FR}}}
\newcommand{\TP}{{\mathrm{TP}}}
\newcommand{\RD}{{\mathrm{RD}}}

\usepackage{multirow}

\allowdisplaybreaks %

\begin{document}
\title{Individualized treatment rules\\ under stochastic treatment cost constraints}
\date{}
\author[1]{Hongxiang Qiu}
\author[2]{Marco Carone}
\author[3]{Alex Luedtke}

\affil[1]{Department of Statistics, the Wharton School, University of Pennsylvania}
\affil[2]{Department of Biostatistics, University of Washington}
\affil[3]{Department of Statistics, University of Washington}

\maketitle

\begin{abstract}
    Estimation and evaluation of individualized treatment rules have been studied extensively, but real-world treatment resource constraints have received limited attention in existing methods. We investigate a setting in which treatment is intervened upon based on covariates to optimize the mean counterfactual outcome under treatment cost constraints when the treatment cost is random.
    In a particularly interesting special case, an instrumental variable corresponding to encouragement to treatment is intervened upon with constraints on the proportion receiving treatment.
    For such settings, we first develop a method to estimate optimal individualized treatment rules. We further construct an asymptotically efficient plug-in estimator of the corresponding average treatment effect relative to a given reference rule.
\end{abstract}

\section{Introduction}

The effect of a treatment often varies across subgroups of the population \citep{Rothwell2005,Varadhan2013}. When such differences are clinically meaningful, it may be beneficial to assign treatments strategically depending on subgroup membership. 
Such treatment assignment mechanisms are called individualized treatment rules (ITRs). 
A treatment rule is commonly evaluated on the basis of the mean counterfactual outcome value it generates --- what is often referred to as the treatment rule's value --- and an ITR with an optimal value is called an optimal ITR. There is an extensive literature on estimation of optimal ITRs and their corresponding values using data from randomized trials or observational studies \citep{Chakraborty2013,Luedtke2016optimalrule,Murphy2003,Robins2004,Zhao2012}.

Most existing approaches for estimating ITRs do not %
incorporate real-world resource constraints. Without such constraints, an optimal ITR would assign the treatment to members of a subgroup provided there is any benefit for such individuals, even when this benefit is minute. In contrast, under treatment resource limits, it may be more advantageous to reserve treatment for subgroups with the greatest benefit from treatment. This issue has received attention in recent work. \citeauthor{Luedtke2016} developed methods for estimation and evaluation of optimal ITRs with a constraint on the proportion receiving treatment \citep{Luedtke2016}. \citeauthor{Qiu2021} instead considered related problems in settings in which instrumental variables (IVs) are available \citep{Qiu2021}. In one of the settings they considered, the same resource constraint is imposed as in \citet{Luedtke2016} but a binary IV is used to identify optimal ITRs even in settings in which there may be unmeasured confounders. In another setting considered in \citet{Qiu2021}, the authors considered interventions on a causal IV or encouragement status, and developed methods to estimate individualized encouragement (rather than treatment) rules with a constraint on the proportion receiving both encouragement and treatment \citep{Qiu2021correction}. They also developed nonparametrically efficient estimators of the average causal effect of optimal rules relative to a prespecified reference rule. \citeauthor{Sun2021} considered a setting in which the cost of treatment is random and dependent upon baseline covariates. They developed methods to estimate optimal ITRs under a constraint on the expected additional treatment cost as compared to control, though inference on the impact of implementing the optimal ITR in the population was not studied \citep{Sun2021}. 
\citet{Sun2021EWM} considered a related problem involving the development of optimal ITRs under resource constraints, and established the asymptotic properties of the estimated optimal ITR. Their method appears viable when the class of ITRs is restricted by the user \textit{a priori}.

In this paper, we study estimation and inference for an optimal rule under two different cost constraints. The first is the same as appearing in \citet{Sun2021}. In contrast to earlier work on this setting, we do not constrain the class of ITRs considered and provide a means to obtain inference about the optimal ITR. The second constraint we consider places a cap on the total cost under the rule rather than on the incremental cost relative to control. To our knowledge, the latter problem has not previously been considered in the literature. 
Both of these estimation problems mirror the intervention-on-encouragement setting considered in \citet{Qiu2021} but involve different constraints and a more general cost function.

Similarly as in \citet{Qiu2021}, the estimators that we develop are asymptotically efficient within a nonparametric model and enable the construction of asymptotically valid confidence intervals for the impact of implementing the optimal rule. We develop our estimators using similar tools --- such as semiparametric efficiency theory \citep{Pfanzagl1990,vanderVaart1998} and targeted minimum loss-based estimation (TMLE) \citep{VanderLaan2006,VanderLaan2018} --- as were used to tackle the related problem studied in \citet{Qiu2021}. Consequently, our proposed estimators are similar to that in \citet{Qiu2021}. Therefore, we will streamline the presentation by highlighting the key similarities and focusing on the differences between these related problems and estimation schemes.

The rest of this paper is organized as follows. In Section~\ref{section: basic assumption}, we describe the problem setup, introduce notation, and present the causal estimands along with basic causal conditions. In Section~\ref{section: identification}, we present additional causal conditions and the corresponding nonparametric identification results. In Section~\ref{section: estimators}, we present our proposed estimators and their theoretical properties. In Section~\ref{section: simulation}, we present a simulation illustrating the performance of our proposed estimators. We make concluding remarks in Section~\ref{section: discussion}. 
Proofs, technical conditions, and additional simulation results can be found in the Supplementary Material.

\section{Setup and objectives} \label{section: basic assumption}

To facilitate comparisons with \citeauthor{Qiu2021} \citep{Qiu2021}, we adopt similar notation as in that work. %
Suppose that we observe independent and identically distributed data units $O_1,O_2,\ldots,O_n \sim P_0$, where $P_0$ is an unknown sampling distribution. A prototypical data unit $O$ consists of the quadruplet $(W,T,C,Y)$, where $W \in \mathscr{W}\subseteq \mathbb{R}^p$ is the vector of baseline covariates, $T \in \{0,1\}$ is the treatment status, $C \in [0,\infty)$ is the random treatment cost, and $Y \in \real$ is the outcome of interest. As a convention, we assume that larger values of $Y$ are preferable. We use $V=V(W) \in \mathcal{V}$ to denote a fixed transformation of $W$ upon which we allow treatment decisions to depend. For example, $V$ may be a subset of covariates in $W$ or a summary of $W$ (e.g., BMI as a summary of height and weight). In practice, $V$ may be chosen based on prior knowledge on potential modifiers of the treatment effect as well as the cost of measuring various covariates. We distinguish between $V(W)$ and $W$ because of their different roles. On the one hand, we will assume that the full covariate $W$ contains all confounders and thus is used to identify causal effects, while $V(W)$ might not be sufficient for this purpose. On the other hand, some covariates in $W$ may be expensive or difficult to measure in future applications, and thus implementing an optimal ITR based on a subset $V(W)$ of covariate $W$ may be desirable. In the rest of this paper, we will use the shorthand notation $V$, $V_i$ and $v$ to refer to $V(W)$, $V(W_i)$ and $V(w)$, respectively. We define an individualized (stochastic) treatment rule (ITR) to be a function $\edecision:\mathcal{V}\rightarrow[0,1]$ that prescribes treatment with probability $\edecision(v)$ according to an exogenous source of randomness for an individual with covariate value $v$. Any stochastic ITR that only takes values in $\{0,1\}$ is referred to as a deterministic ITR. %

In this work, we adopt the potential outcomes framework \citep{Neyman1923,Rubin1974}. For each individual, we use $C(t)$ and $Y(t)$ to denote the potential treatment cost and potential outcome, respectively, corresponding to scenarios in which the individual has treatment status $t$. We use $\cexpect$ to denote an expectation over the counterfactual observations and the exogenous random mechanism defining a rule, and $\expect_0$ to denote an expectation over observables alone under sampling from $P_0$. We make the usual Stable Unit Treatment Value Assumption (SUTVA) assumption.
\begin{assumptionCausal}[Stable Unit Treatment Value Assumption] \label{IV assumption: sutva}
    The counterfactual data unit of one individual is unaffected by the treatment assigned to other individuals, and there is only a single version of the treatment, so that $T=t$ implies that $C=C(t)$ and $Y=Y(t)$.
\end{assumptionCausal}

\begin{remark}
    The ITRs we consider are not truly individualized, because they are based on the value of covariate $V$ rather than each individual's unique potential treatment effects $Y(1)-Y(0)$ and $C(1)-C(0)$. Nevertheless, depending on the resolution of $V$, these ITRs can be considerably more individualized than assigning everyone to either treatment or control. In this paper, we adopt the conventional nomenclature and refer to the treatment rules we study as ITRs \citep[see, e.g.,][]{Butler2018,Chen2018,Imai2021,Laber2015,Lei2012,Petersen2007,Qian2011,Song2015,VanderLaan2007ITR,Zhao2012,Zhao2015DR,Zhou2017}.
\end{remark}

We define $C(\edecision)$ and $Y(\edecision)$ to be the
counterfactual treatment cost
and outcome, respectively, for an ITR $\edecision$ under an exogenous random mechanism. We note that if $\edecision(v) \in (0,1)$ for an individual with covariate $v$, an exogenous random mechanism is used to randomly assign treatment with probability $\edecision(v)$ and thus $C(\edecision)$ and $Y(\edecision)$ are random for this given individual. If $\edecision$ were implemented in the population, then the population mean outcome would be $\cexpect\left[Y(\edecision)\right]$, where we use $\cexpect$ to denote expectation under the true data-generating mechanism involving potential outcomes $(C(t),Y(t))$ and exogenous randomness in $\rho$. We consider a generic treatment resource constraint requiring that a convex combination of the population average treatment cost and the population average additional treatment cost compared to control be no greater than a specified constant $\kappa \in (0,\infty]$. Consequently, an optimal ITR $\edecision_0$ under this constraint is a solution in $\edecision: \mathcal{V} \rightarrow [0,1]$ to
\begin{align}
    \textnormal{maximize}\hspace{0.5em} & \cexpect[Y(\edecision)] \hspace{1.5em}\text{subject to} \hspace{1.5em} \alpha \cexpect[C(\edecision)] + (1-\alpha) \cexpect[C(\edecision)-C(0)] \leq \kappa\ . \label{eq:encouragedecision}
\end{align}
Here, $\alpha \in [0,1]$ is also a constant specified by the investigator. Natural choices of $\alpha$ are $\alpha=0$, corresponding to a constraint on the population average additional treatment cost compared to control, and $\alpha=1$, corresponding to a constraint on the population average treatment cost. The first choice may be preferred when the control treatment corresponds to the current standard of care and a limited budget is available to fund the novel treatment to some patients. The second choice may be more relevant when both treatment and control incur treatment costs.

\begin{remark} \label{remark: previous works}
    Our setup is similar to that in \citet{Qiu2021} if we view $T$ and $C$ defined here as the instrumental variable/encouragement $Z$ and treatment status $A$ defined in those prior works, respectively. However, the constraint in our setup is different from the constraint  $\cexpect[\edecision(V) C(\edecision)] \leq \kappa$ considered previously. In IV settings, the constraint in \eqref{eq:encouragedecision} with $\alpha=1$ is useful when assigning treatment \emph{always} incurs a cost, regardless of whether encouragement is applied, such as in distributing a limited supply of an expensive drug within a health system based on the results of a randomized clinical trial. It is instead useful with $\alpha=0$ when no encouragement is present under the standard of care but intervention on the encouragement is of interest when additional treatment resources are available. The constraint considered in \citet{Qiu2021,Qiu2021correction} was instead useful in cases in which treatment only incurs a cost when paired with encouragement, such as  when housing vouchers are used to encourage individuals to live in a certain area. In the general setting in which $T$ is viewed as treatment status and $C$ as a random treatment cost, the constraint in \eqref{eq:encouragedecision} with $\alpha=0$ is identical to that considered in \citeauthor{Sun2021} \citep{Sun2021} --- we refer the readers to these works for a more in-depth discussion of the relation between the current problem setup and IV settings.

\end{remark}

To evaluate an optimal ITR $\edecision_0$, we follow \citet{Qiu2021} in considering three types of reference ITRs and develop methods for statistical inference on the difference in the mean counterfactual outcome between $\edecision_0$ and a reference ITR $\edecision^\mathcal{R}_0 : \mathcal{V} \rightarrow [0,1]$. %
The first type of reference ITR considered, denoted by $\edecision^\FR$ ($\FR$=fixed rule), is any fixed ITR that may be specified by the investigator before the study. When $\alpha=0$, it is usually most reasonable to consider the rule that always assigns control, namely $v \mapsto 0$, because the constraint in \eqref{eq:encouragedecision} may arise due to limited funding for implementing treatment whereas the standard of care rule is to always assign control. The second type, denoted by $\edecision^\RD_0$ ($\RD$=random), prescribes treatment completely at random to individuals regardless of their baseline covariates. The probability of prescribing treatment is chosen such that the treatment resource is saturated (i.e., all available resources are used) or all individuals receive treatment, if such a probability exists. Symbolically, this ITR is given by $\edecision^\RD_0:v\mapsto\min\left\{1,(\kappa-\alpha\cexpect\left[C(0)\right])/\cexpect\left[C(1)-C(0)\right]\right\}$ under the condition that $\cexpect\left[C(0)\right] \leq \kappa$ and $\cexpect[C(1)-C(0)]>0$. Although $\edecision^\RD_0$ has the same interpretation as the corresponding encouragement rule in \citet{Qiu2021}, its mathematical expression is different due to the different resource constraint. This rule may be of interest if it is known \textit{a priori} that treatment is harmless. The third type, denoted by $\edecision^\TP_0$ ($\TP$=true propensity), prescribes treatment according to the true propensity of the treatment implied by the study sampling mechanism $P_0$, so that $\edecision^\TP_0$ equals  $w \mapsto P_0\left(T=1\mid W=w\right)$. This ITR may be of interest in two settings. In one setting, $\edecision^\TP_0$ satisfies the treatment resource constraint. The investigator may wish to determine the extent to which the implementation of an optimal ITR would improve upon the standard of care. In the other setting, the treatment resource constraint is newly introduced and the standard of care ITR may lead to overuse of treatment resources. The investigator may then be interested in whether the implementation of an optimal constrained ITR would result,  despite the new resource constraint, in a noninferior mean outcome.

\section{Identification of causal estimands} \label{section: identification}

In this section, we present nonparametric identification results. Though these results are similar to those for individualized encouragement rules in \citet{Qiu2021}, there are two key differences. First, the form of some of the conditions in \citet{Qiu2021} need to be modified to account for the novel resource constraint considered here. Second, two additional conditions are needed to overcome challenges that arise due to this new constraint.

We first introduce notation that will be useful when presenting our identification results and our proposed estimators. For any observed-data distribution $P$, we define pointwise the conditional mean functions $\Q^C_{P}(t,w):=\expect_{P}(C\mid T=t,W=w)$ and $\Q^Y_{P}(t,w):=\expect_{P}(Y\mid T=t,W=w)$, where we use $\expect_P$ to denote an expectation over observables alone under sampling from $P$, and their corresponding contrasts due to different treatment status, $\Delta_{P}^C(w):=\Q^C_{P}(1,w)-\Q^C_{P}(0,w)$ and $\Delta_{P}^Y(w):=\Q^Y_{P}(1,w)-\Q^Y_{P}(0,w)$. We also define the average of these contrasts conditional on $V$ as $\delta^C_{P}(v):=\expect_P[\Delta^C_P(W) \mid V=v]$ and $\delta^Y_{P}(v):=\expect_P[\Delta^Y_P(W) \mid V=v]$, and the propensity to receive treatment $\Q^T_P(w):= P\left(T=1\mid W=w\right)$.
Additionally, we define $\nu_P(t,v):=\expect_{P}\left[\expect_P\left(C\mid T=t,W\right) \mid V=v\right]$, $\phi_P := \expect_P[\Q^C_P(0,W)]$. These quantities play an important role in tackling the problem at hand. Throughout the paper, for ease of notation, if $f_P$ is a quantity or operation indexed by distribution $P$, we may denote $f_{P_0}$ by $f_0$. As an example, we may use $\Delta^Y_0$ to denote $\Delta^Y_{P_0}$.

We introduce additional causal conditions we will require, positivity and unconfoundedness. In one form or another, these conditions commonly appear in the causal inference literature \citep{VanderLaan2018}, including in the IV literature \citep{Abadie2003,Imbens1994,TchetgenTchetgen2013,Wang2018}.

\begin{assumptionCausal}[Strong positivity] \label{IV assumption: strong IV positivity}
    There exists a constant $\epsilon_T>0$ such that  $\epsilon_T<\Q^T_{0}(w) <1-\epsilon_T$ holds for $P_0$-almost every $w$.
\end{assumptionCausal}

\begin{assumptionCausal}[Unconfoundedness of treatment] \label{IV assumption: LATE confounder}
    For each $t\in\{0,1\}$, $T$ and $(C(t),Y(t))$ are conditionally independent given $W=w$ for $P_0$-almost every $w$.
\end{assumptionCausal}

Equipped with these conditions, we are able to state a theorem on the nonparametric identification of the mean counterfactual outcomes and average treatment effect (ATE) --- these results can be viewed as a corollary of the well-known G-formula \citep{Robins1986}.
\begin{theorem}[Identification of ATE and expected treatment resource expenditure] \label{theorem: identify ATE}
    Provided Conditions~\ref{IV assumption: sutva}--\ref{IV assumption: LATE confounder} are satisfied, it holds that $\cexpect[Y(t) \mid W=w]=\Q^Y_0(t,w)$, $\cexpect[Y(1) - Y(0) \mid W=w] = \Delta^Y_0(w)$ and $\cexpect[Y(1)-Y(0) \mid V=v]=\delta^Y_{0}(v)$ for $P_0$-almost every $w$ and $v$, and so,  $\cexpect[Y(\edecision) - Y(\edecision^\mathcal{R}_0)] = \expect_0[\{\edecision(V)-\edecision^\mathcal{R}_0(W)\} \Delta^Y_0(W)]$. In addition, it holds that $\cexpect[C(t) \mid W=w]=\Q^C_0(t,w)$ for $P_0$-almost every $w$, and so, $\cexpect[C(\edecision)] = \expect_0[\edecision(V) \Q^C_0(1,W) + (1-\edecision(V)) \Q^C_0(0,W)]$.
\end{theorem}

In view of Theorem~\ref{theorem: identify ATE}, the objective function in \eqref{eq:encouragedecision} can be identified as
\begin{align*}
    \cexpect\left[Y(\edecision)\right]= \expect_0 \left[ \edecision(V) \Q^Y_0(1,W) + (1-\edecision(V)) \Q^Y_0(0,W) \right] = \expect_0\left[\edecision(V)\Delta_0^Y(W)\right] + \expect_{0}[\Q^Y_0(0,W)]\ , %
\end{align*}
and, similarly, the expected cost is identified as $\cexpect\left[C(\edecision)\right]=\expect_0\left[\edecision(V)\Delta_0^C(W)\right] + \expect_{0}[\Q^C_0(0,W)]$. It follows that the optimization problem \eqref{eq:encouragedecision} is equivalent to
\begin{align}
    \textnormal{maximize}\hspace{0.5em} & \expect_0\left[\edecision(V) \delta_{0}^Y(V)\right] \hspace{1.5em}\text{subject to} \hspace{1.5em}  \expect_0\left[\edecision(V) \delta_{0}^C(V)\right] + \alpha \phi_0 \leq \kappa\ . \label{eq:ATEknapsack}
\end{align}
This differs from Equation~3 defining optimal individualized encouragement rules in \citet{Qiu2021}. We now present two additional conditions so that \eqref{eq:ATEknapsack} is a fractional knapsack problem \citep{Dantzig1957}, thereby allowing us to use existing results from the optimization literature. These conditions are similar to those in \citeauthor{Sun2021} \citep{Sun2021}.

\begin{assumptionCausal}[Strictly costlier treatment] \label{IV assumption: strong encouragement}
    There exists a constant $\epsilon_C>0$ such that $\Delta^C_0(w) > \epsilon_C$ holds for $P_0$-almost every $w$.
\end{assumptionCausal}

\begin{assumptionCausal}[Financial feasibility of assigning treatment] \label{IV assumption: existence of nontrivial feasible IER}
    The inequality $\alpha \phi_0 < \kappa$ holds.
\end{assumptionCausal}

Condition~\ref{IV assumption: strong encouragement} is reasonable if treatment is more expensive than control. When applied to an IV setting as outlined in Remark~\ref{remark: previous works}, this condition corresponds to the assumption that the IV is indeed an encouragement to take treatment. This condition is slightly stronger than its counterpart in \citeauthor{Sun2021} \citep{Sun2021}, which only requires that $\Delta^C_0 \geq 0$. This stronger condition is needed to ensure the asymptotic linearity of our proposed estimator in Section~\ref{section: estimators}. Under Condition~\ref{IV assumption: strong encouragement}, it is evident that Condition~\ref{IV assumption: existence of nontrivial feasible IER} is reasonable because  if $\alpha \phi_0 > \kappa$, then no ITR satisfies the treatment resource constraint in view of the fact that $\expect_0 [\edecision(V) \delta_{0}^C(V) ] \geq 0$, whereas if $\alpha \phi_0 = \kappa$, then only the trivial ITR $v \mapsto 0$ satisfies the constraint and there is no need to estimate an optimal ITR.

Under these two additional conditions, \eqref{eq:ATEknapsack} is a fractional knapsack problem \citep{Dantzig1957} in which every subgroup defined by a different value of $V$ corresponds to a different `item'. A solution in the special case in which $V(W)=W$ and $\alpha=0$ was given in Theorem~1 of \citet{Sun2021}. We now state a more general result with the following differences: (i) the treatment decision may be based on a summary $V$ rather than the entire covariate vector $W$, and (ii) $\alpha$ may take any value in $[0,1]$ rather than only zero. We also explicitly state the randomization probability at the boundary for completeness and clarity. Despite these differences, the result we obtain is similar to Theorem~1 in \citet{Sun2021}. Define pointwise $\xi_0(v) := \delta^Y_{0}(v)/\delta^C_{0}(v)$, and write $\eta_0:=\inf \{ \eta:  \expect_0 [I(\xi_0(V) > \eta) \delta^C_{0}(V) ] \leq \kappa - \alpha \phi_0 \}$ and $\tau_0 := \max\{\eta_0,0 \}$.

\begin{theorem}[Optimal ITR] \label{theorem: true optimal IER}
     Under Conditions~\ref{IV assumption: sutva}--\ref{IV assumption: existence of nontrivial feasible IER}, a solution to \eqref{eq:ATEknapsack} is explicitly given by
\begin{align*}
    \edecision_0(v) := \begin{cases}
        \ \frac{\kappa - \alpha \phi_0 - \expect_0\left[I(\xi_0(V) > \tau_0) \delta^C_{0}(V)\right]}{\expect_0\left[I(\xi_0(V) = \tau_0) \delta^C_{0}(V)\right]} &:\ \text{ if } \tau_0>0,\ \xi_0(v)=\tau_0\textnormal{ and }\expect_0\left[I(\xi_0(V) = \tau_0) \delta^C_{0}(V)\right] > 0 \\
        \ I\left(\xi_0(v)>\tau_0\right) &:\ \text{ otherwise\ .}
    \end{cases}
\end{align*}
Here, the first case is the boundary case with the randomization probability that saturates the treatment resource.
\end{theorem}

We also note that the reference ITRs introduced in Section~\ref{section: basic assumption} are also identified under the above conditions. In particular, it can be shown that $\edecision^\RD_0(v):= \min\{1,(\kappa-\alpha \phi_0)/\expect_0[\Delta^C_0(W)]\}$ and $\edecision^\TP_0=\Q^T_0$.

\section{Estimating and evaluating optimal individualized treatment rules} \label{section: estimators}

In this section, we present an estimator of an optimal ITR $\edecision_0$ and an inferential procedure for its ATE relative to a reference ITR $\edecision^\mathcal{R}_0$, where $\mathcal{R}$ is any of $\FR$, $\RD$ or $\TP$. The proposed procedure is an adaptation of the method first proposed in \citet{Qiu2021,Qiu2021correction}.

We begin by introducing some notations that are useful for defining the estimands. We define the parameter $\Psi_{\edecision}(P):=\expect_P\left[\edecision(V)\Delta^Y_P(W)\right]$ or $\Psi_{\edecision}(P):=\expect_P\left[\edecision(W)\Delta^Y_P(W)\right]$ for each ITR $\edecision$ and distribution $P\in\mathscr{M}$, depending on whether the domain of $\edecision$ is $\mathcal{V}$ or $\mathcal{W}$. Here, we consider the model $\mathscr{M}$ to be locally nonparametric at $P_0$ \citep{Pfanzagl1990}. For $P\in\mathscr{M}$, the ATE of an optimal ITR $\edecision_P$ relative to a reference ITR $\edecision^\mathcal{R}_P$ equals $\Psi_{\mathcal{R}}(P):=\Psi_{\edecision_P}(P) - \Psi_{\edecision^\mathcal{R}_P}(P)$. We are interested in making inference about $\psi_{0}:=\Psi_{\mathcal{R}}(P_0)$, where we have suppressed dependence on $\mathcal{R}$ from our shorthand notation.

\subsection{Pathwise differentiability of the ATE}

We first present a result regarding the pathwise differentiability of the ATE. Pathwise differentiability of the parameter of interest serves as the foundation for constructing asymptotically efficient estimators of this parameter, based on which an inferential procedure may be developed. Additional technical conditions are required and are provided in Section~\ref{section: additional technical conditions} in the Supplementary Material.
For a distribution $P \in \mathscr{M}$, a function $\Q^C: \{0,1\} \times \mathcal{W} \rightarrow \real$, an ITR $\edecision$, and a decision threshold $\tau \in \real$, we define pointwise the following functions:
\begin{align}
\begin{split}
    D(P,\edecision,\tau,\Q^C)(o)\ &:=\ \edecision(v)\left[\frac{y-\Q^Y_P(t,w)}{t + \Q^T_P(w) - 1} +  \Delta^Y_P(w)\right] - \Psi_{e}(P) \\
    &\hspace{.25in}- \tau \left\{ \edecision(v) \left[ \frac{c - \Q^C(t,w)}{t + \Q^T_P(w) - 1} + \Delta^C(w) \right] + \alpha \left[ \frac{(1-t)(c - \Q^C(0,w))}{1-\Q^T_P(w)} + \Q^C(0,w) \right] - \kappa \right\}; \\
    G(P)(o)\ &:= D(P,\edecision_P,\tau_P,\Q^C_P) (o) \ ; \\
    D_1(P,\Q^C)(o)\ &:=\ \frac{(1-t)(c - \Q^C(0,w))}{1-\Q^T_P(w)} + \Q^C(0,w) - \expect_P \left[ \Q^C(0,W) \right]\ ; \\
    D_2(P,\Q^C)(o)\ &:=\ \frac{c - \Q^C(t,w)}{t+\Q^T_P(w)-1} + \Delta^C(w) - \expect_P \left[ \Delta^C(W) \right]\ ; \\
    G_\RD(P)(o)\ &:=\ D(P,\edecision^\RD_P,0,\Q^C_P)(o) - \frac{\alpha \Psi_{\edecision^\RD_P}(P) D_1(P,\Q^C_P)}{\kappa - \phi_P}- \frac{\Psi_{\edecision^\RD_P}(P) D_2(P,\Q^C_P)}{\expect_P\left[\Delta_P^C(W)\right]}\ ;\\
    G_\TP(P)(o)\ &:=\ \frac{\Q^T_P(w)}{t+\Q^T_P(w)-1} \left[y-\Q^Y_P(t,w)\right] + t \Delta^Y_P(w) - \Psi_{\edecision^\TP_P}(P)\ ; \\
    G_\FR(P)(o)\ &:=D(P,\edecision^\FR,0,\Q^C_P) (o) \ .
\end{split} \label{eq: gradients}
\end{align}

One key condition we rely on is the following non-exceptional law assumption.
\begin{assumptionEncourage}[Non-exceptional law] \label{e RC assumption: non-exceptional law}
    $P_0(\xi_0(V)=\tau_0)=0$.
\end{assumptionEncourage}
Under this condition, the true optimal ITR $\rho_0$ is identical to an indicator function. If all covariates are discrete, then we can plug in the empirical estimates into the identification formulae in Theorems~\ref{theorem: identify ATE}--\ref{theorem: true optimal IER} and show that the resulting estimators of the ATE are asymptotically normal by the delta method even when Condition~\ref{e RC assumption: non-exceptional law} does not hold. We do not further pursue this simple case in this paper, and thus need to rely on the non-exceptional law assumption, namely Condition~\ref{e RC assumption: non-exceptional law}, to account for continuous covariates. We list additional technical conditions in Supplement~\ref{section: additional technical conditions}.

We can now provide a formal result describing the pathwise differentiability of the ATE parameter.

\begin{theorem}[Pathwise differentiability of the ATE] \label{e RC theorem: differentiability}
    Let $\mathcal{R} \in \{\FR,\RD,\TP\}$. Provided Conditions~\ref{IV assumption: sutva}--\ref{IV assumption: existence of nontrivial feasible IER} and \ref{e RC assumption: non-exceptional law}--\ref{e RC assumption2: active constraint} are satisfied, the parameters $P \mapsto \Psi_{\edecision_P}(P)$ and $P \mapsto \Psi_{\edecision^\mathcal{R}_P}(P)$ are pathwise differentiable at $P_0$ relative to $\mathscr{M}$ with canonical gradients $G(P_0)$ and $G_\mathcal{R}(P_0)$, respectively.
\end{theorem}
We note that the pathwise differentiability of $P \mapsto \Psi_{\edecision^\mathcal{R}_P}(P)$ was established in Theorem~3 of \citet{Qiu2021} for $\mathcal{R} \in \{\FR,\TP\}$. The other results can be proven using similar techniques. We put the proof of these results in Supplement~\ref{section: proof differentiability}. In view of Theorem~\ref{e RC theorem: differentiability}, it follows that the ATE parameter $\Psi_{\mathcal{R}}$ is pathwise differentiable at $P_0$ with nonparametric canonical gradient
\begin{equation}
    D_\mathcal{R}(P_0) := G(P_0) - G_\mathcal{R}(P_0) \label{eq: gradients2}
\end{equation}
for $\mathcal{R} \in \{\FR,\RD,\TP\}$.

\begin{remark} \label{remark: additional term in gradient from knapsack}
We have noted similar additional terms related to the resource being used in the canonical gradient of the mean counterfactual outcome or ATE of optimal ITRs under resource constraints, for example, in \citet{Luedtke2016} and \citet{Qiu2021}. In our problem, this additional term is
$$- \tau_0 \left\{ \edecision_0(v) \left[ \frac{c - \Q^C_0(t,w)}{t + \Q^T_0(w) - 1} + \Delta^C_0(w) \right] + \alpha \left[ \frac{(1-t)(c - \Q^C_0(0,w))}{1-\Q^T_0(w)} + \Q^C_0(0,w) \right] - \kappa \right\}.$$
Such terms appear to come from solving a fractional knapsack problem with truncation at zero and take the form of a product of (i) the threshold in the solution, and (ii) a term that equals the influence function of the resource being used under the solution when the resource is saturated. We conjecture that such structures generally exist for fractional knapsack problems.
\end{remark}

\subsection{Proposed estimator and asymptotic linearity} \label{section: procedure}

We next present our proposed nonparametric procedure for estimating an optimal ITR $\edecision_0$ and the corresponding ATE $\psi_0$. We will generally use subscript $n$ to denote an estimator with sample size $n$, and add a hat to a nuisance function estimator that is targeted toward estimating $\phi_0$.

\begin{enumerate}
    \item Use the empirical distribution $\hat{P}_{W,n}$ of $W$ as an estimate of the true marginal distribution of $W$. Compute estimates $\Q^Y_n$, $\Q^C_n$, $\Q^T_n$, $\delta^Y_{n}$ and $\delta^C_{n}$ of $\Q^Y_0$, $\Q^C_0$, $\Q^T_0$, $\delta^Y_{0}$ and $\delta^C_{0}$, respectively, using flexible regression methods. Recall that $\Q^Y_0(t,w)=\expect_0[Y \mid T=t,W=w]$, $\Q^C_0(t,w) = \expect_0[C \mid T=t,W=w]$, $\delta^Y_{0}(v)=\expect_0[\Q^Y_0(1,W)-\Q^Y_0(0,W) \mid V=v]$, and $\delta^C_{0}(v)=\expect_0[\Q^C_0(1,W)-\Q^C_0(0,W) \mid V=v]$. Define pointwise $\Delta^C_n(w):=\Q^C_n(1,w)-\Q^C_n(0,w)$.
    
    \item Estimate an optimal ITR:
    \begin{enumerate}
        \item Estimate $\phi_0 =\expect_0[\Q^C_0(0,W)]$ with a one-step correction estimator
        $$\phi_n := \frac{1}{n} \sum_{i=1}^n \left[ \Q^C_n(0,W_i) + \frac{(1-T_i)(C_i - \Q^C_n(0,W_i))}{1-\Q^T_n(W_i)} \right].$$
        \item Let $\xi_n:=\delta^Y_{n}/\delta^C_{n}$, $\Gamma_n : \tau \mapsto \frac{1}{n}\sum_{i : \xi_n(V_i)>\tau} \Delta^C_n(W_i)$ and $\gamma_n : \tau\mapsto \frac{1}{n}\sum_{i : \xi_n(V_i)=\tau} \Delta^C_n(W_i)$. For any $k \in [0,\infty]$, define $\eta_n(k) := \inf \left\{\tau: \Gamma_n(\tau) \leq k - \alpha \phi_n \right\}$, $\tau_n(k) := \max\left\{\eta_n(k),0\right\}$ and
        $$d_{n,k}: v \mapsto \begin{cases}
            \ \frac{k - \alpha \phi_n - \Gamma_n(\eta_n(k))}{ \gamma_n(\eta_n(k))} &:\ \text{ if } \xi_n(v)=\eta_n(k) \textnormal{ and } \gamma_n(\eta_n(k)) > 0\ , \\
            \ I\{\xi_n(v)>\eta_n(k)\} &:\ \text{ otherwise.}
        \end{cases}$$
        The rule $d_{n,k}$ is the sample analog of an ITR that prescribes treatment to those with the highest values of $\xi_0(V)$, regardless of whether treatment is harmful or not, until treatment resources run out.
        \item Compute $k_n$, which is used to define an estimate of $\edecision_0$ for which the plug-in estimator is asymptotically linear under conditions, as follows:
        \begin{itemize}
            \item if $\tau_n (\kappa) > 0$ \textit{and} there is a solution in $k\in[0,\infty)$ to
            \begin{equation} \label{e RC equation: update quantile}
                \frac{1}{n}\sum_{i=1}^n  d_{n,k}(V_i) \left[ \Delta^C_n(W_i) + \frac{C_i - \Q^C_n(T_i,W_i)}{t_i + \Q^T_n(W_i) - 1} \right] + \alpha \phi_n = \kappa \ ,
            \end{equation}
            then take $k_n$ to be this solution;
            \item otherwise, set $k_n=\kappa$.
        \end{itemize}
        \item Estimate $\edecision_0$ using the sample analog of $\edecision_0$ with treatment resource constraint $k_n$, namely
        $$\edecision_n: v \mapsto \begin{cases}
            \ \frac{k_n - \alpha \phi_n - \Gamma_n(\tau_n(k_n))}{\gamma_n(\tau_n(k_n))} &:\ \text{ if } \xi_n(v)=\tau_n(k_n) \textnormal{ and } \gamma_n(\tau_n(k_n))>0 \\
            \ I\{\xi_n(v)>\tau_n(k_n)\} &:\ \text{ otherwise.}
        \end{cases}$$
    \end{enumerate}
    
    \item Obtain an estimate $\edecision^\mathcal{R}_n$ of the reference ITR $\edecision^\mathcal{R}_0$ as follows:
    \begin{itemize}
        \item For $\mathcal{R}=\FR$, take $\edecision^\mathcal{R}_n$ to be $\edecision^\FR$.
        \item For $\mathcal{R}=\RD$,
        \begin{enumerate}
            \item \label{ie:targQA} obtain a targeted estimate $\hat{\Q}^C_n(1,\cdot)$ of $\Q^C_0(1,\cdot)$: run an ordinary least-squares linear regression with outcome $C$, covariate $1/(T+\Q^T_n(W)-1)$, offset $\Q^C_n(T,W)$ and no intercept. Take $\hat{\Q}^C_n$ to be the fitted mean model;
            \item take $\edecision^\mathcal{R}_n$ to be the constant function $w \mapsto \min \left\{1,(\kappa - \phi_n)/\tfrac{1}{n}\sum_{i=1}^{n} \hat{\Delta}^C_n(W_i)\right\}$, where we define pointwise $\hat{\Delta}^C_n(w):=\hat{\Q}^C_n(1,w)-\hat{\Q}^C_n(0,w)$.
        \end{enumerate}
        \item For $\mathcal{R}=\TP$,
        take $\edecision^\mathcal{R}_n$ to be $\Q^T_n$.
    \end{itemize}

    \item Estimate ATE of $\edecision_0$ relative to the reference ITR $\edecision^\mathcal{R}_0$ with a targeted minimum-loss based estimator (TMLE) $\psi_n$:
    \begin{enumerate}
        \item \label{ie:targQY} obtain a targeted estimate $\hat{\Q}^Y_n$ of $\Q^Y_0$: run an ordinary least-squares linear regression with outcome $Y$, covariate $[\edecision_n(V)-\edecision^\mathcal{R}_n(W)]/[T + \Q^T_n(W) - 1]$, offset $\Q^Y_n(T,W)$ and no intercept. Take $\hat{\Q}^Y_n$ to be the fitted mean function.
        
        \item with $\hat{P}_n$ being any distribution with components $\hat{\Q}^Y_n$ and $\hat{P}_{W,n}$, take 
        $$\psi_n:=\Psi_{\edecision_n}(\hat{P}_n)-\Psi_{\edecision^\mathcal{R}_n}(\hat{P}_n) = \frac{1}{n} \sum_{i=1}^n [\edecision_n(V_i) - \edecision^\mathcal{R}_{n,i}] [\hat{\Q}^Y_n(1,W_i) - \hat{\Q}^Y_n(1,W_i)]\ ,$$
        where $\edecision^\mathcal{R}_{n,i}$ is defined as $\edecision^\mathcal{R}_n(W_i)$ or $\edecision^\mathcal{R}_n(V_i)$ depending on the covariate used by the reference ITR.
    \end{enumerate}
\end{enumerate}

The above procedure is similar to that proposed in \citet{Qiu2021}. One key difference is the use of the refined estimator $k_n$ of $\kappa$ obtained via the estimating equation \eqref{e RC equation: update quantile}, which is key to ensuring the asymptotic linearity of  $\psi_n$. Another difference is that the denominator of $\xi_n$ is now $\delta^C_{n}$, which is consistent with our different definition of the unit value for solving the fractional knapsack problem \eqref{eq:ATEknapsack}. Similarly to TMLE for other problems, when $C$ or $Y$ has known bounds (e.g., the closed interval $[0,1]$), to obtain a corresponding targeted estimate that respect the known bounds, we may use logistic regression rather than ordinary least-squares \citep{Gruber2010}.

The above procedure has both similarities and substantial differences compared to the estimation procedure proposed by \citet{Sun2021}. The main difference is that our procedure is targeted towards efficient estimation of and inference about the ATE of $\psi_0$ of the optimal ITR under a nonparametric model, while \citet{Sun2021} focus on estimating the optimal ITR $\edecision_0$ and does not evaluate this optimal ITR. This leads to a key difference between the two procedures when estimating the optimal ITR: we need to solve an estimating equation \eqref{e RC equation: update quantile}, which is crucial to ensuring that the estimator $\psi_n$ is asymptotically linear, while \citet{Sun2021} do not. The requirement of solving \eqref{e RC equation: update quantile} is related to the nature of the fractional knapsack problem discussed in Remark~\ref{remark: additional term in gradient from knapsack}, and we conjecture that such a calibration on the resource used is necessary for general problems of the same nature. Our procedure is also related to the method in \citet{Sun2021EWM}. \citet{Sun2021EWM} rely on the availability of asymptotically normal estimators of both the average benefit and average resource used (Assumption~2.4), a nontrivial requirement when the propensity score $\Q^T_0$ is unknown in observational studies. Our procedure essentially produces such estimators: in Step~4, an asymptotically normal estimator of the ATE is constructed, whereas an asymptotically normal estimator of the expected resource is produced in Step~2 and used to calibrate the resource expenditure of the estimated optimal ITR $\edecision_n$ in Step~2(c).

\begin{remark}
In Step~1 of the above procedure, we estimate the functions $\delta^Y_{0}$ and $\delta^C_{0}$ using a na\"ive approach based on outcome regression. It is viable to use more advanced techniques such as the doubly robust methods in \citet{VanderLaan2006}, \citet{VanderLaan2014,Luedtke2016pseudooutcome}, and \citet{Kennedy2020} or R-learning as in \citet{Nie2021}. These methods were developed for conditional average treatment effect estimation and might lead to better estimators of $\delta^Y_{0}$ and $\delta^C_{0}$. It is also possible to develop multiply robust methods to estimate $\xi_0$ using influence function techniques. Such methods to estimate $\xi_0$ are beyond the scope of our paper, whose main focus is on the inference for the ATE. Our theoretical analysis of the estimator only applies to na\"ive estimators based on outcome regression, but we expect only minor modifications to be required to study these more advanced estimators once their asymptotic behavior is characterized.
\end{remark}

\begin{remark}
In Step~2(a), it is also viable to use other efficient estimators of $\phi_0$, for example, a targeted minimum loss-based estimator (TMLE). We note that estimating $\phi_0$ is only one component of estimating the optimal ITR $\edecision_0$. Methods such as TMLE can be preferable to ensure that the estimator respects  known bounds on the estimand. However, in our case, such an improvement in estimating $\phi_0$ does not necessarily lead to an improvement in the estimation of $\edecision_0$.
\end{remark}

We now present results on the asymptotic linearity and efficiency of our proposed estimator.
We state and discuss the technical conditions required by the theorem below in Supplement~\ref{section: additional technical conditions}.

\begin{theorem}[Asymptotic linearity of ATE estimator] \label{e RC theorem: asymptotic linearity}
    Let $\mathcal{R} \in \{\FR,\RD,\TP\}$. Under Conditions~\ref{e RC assumption: non-exceptional law}--\ref{e RC assumption: Glivenko-Cantelli}, with the canonical gradient $D_\mathcal{R}(P_0)$ defined in \eqref{eq: gradients} and \eqref{eq: gradients2}, it holds that
    $$\psi_n - \psi_0 = \frac{1}{n} \sum_{i=1}^n D_\mathcal{R}(P_0)(O_i) + \smallo_p(n^{-1/2})\ .$$
    Therefore, $\sqrt{n}\left(\psi_n - \psi_0\right) \overset{d}{\longrightarrow} \textnormal{N}\left(0, \sigma_0^2 \right)$, where $\sigma_0^2 := \expect_0 \left[ D_\mathcal{R}(P_0)(O)^2 \right]$. Since $\psi_n$ is asymptotically linear with influence function equal to the canonical gradient, $\psi_n$ is also asymptotically efficient.
\end{theorem}
To conduct inference about $\psi_0$, we can directly plug the estimators of nuisance functions into $D_\mathcal{R}(P_0)$ to obtain a consistent estimator of $D_\mathcal{R}(P_0)$, and then take the sample variance to obtain a consistent estimator of the asymptotic variance $\sigma_0^2$. The proof of Theorem~\ref{e RC theorem: asymptotic linearity} can be found in Supplements~\ref{section: pseudo-gradient expansion} and \ref{proof: asymptotic linearity}.

\begin{remark}\label{remark: sample splitting}
    It may be desirable to use cross-fitting \citep{Newey2018,Zheng2011} to estimate an optimal ITR for better finite-sample performance. The asymptotic linearity is maintained by a similar argument that is used to prove Theorem~\ref{e RC theorem: asymptotic linearity}. We describe this algorithm in Section~\ref{section: sample splitting} in the Supplementary Material.
\end{remark}

\begin{remark} \label{remark: bounded outcome}
    We note that, unlike in \citet{Qiu2021} where the bound $\kappa$ lies in $(0,1]$ due to the binary nature of treatment status, the methods we propose here do not require knowledge of an upper bound on treatment costs. When such a bound is indeed known (e.g., one), our methods may still be applied as long as all special cases corresponding to $\kappa=\infty$ or $\kappa<\infty$ in Section~\ref{section: estimators} are replaced by $\kappa$ being equal to or less than the known bound, respectively.
\end{remark}

\section{Simulation} \label{section: simulation}

\subsection{Simulation setting}

In this simulation study, we investigate the performance of our proposed estimator of the ATE of an optimal ITR relative to specified reference ITRs. We focus here on the setting  $\alpha=1$. This scenario is more difficult than the case $\alpha=0$ because it requires the estimation of $\phi_0$.

We generate data from a model in which the treatment $T$ is an IV and the treatment cost $C$ and outcome $Y$ are both binary. This data-generating mechanism satisfies all causal conditions and has an unobserved confounder between treatment cost and outcome. We first generate a trivariate covariate $W=(W_1,W_2,W_3)$, where $W_1 \sim \mathrm{Unif}(-1,1)$, $W_2 \sim \mathrm{Bernoulli}(0.8)$ and $W_3 \sim \mathrm{N}(0,1)$ are mutually independent. We also simulate an unobserved treatment-outcome confounder $U \sim \mathrm{Bernoulli}(0.5)$ independently of $W$, and then simulate $T$, $C$, and $Y$ as follows:
\begin{alignat*}{2}
    & T \mid W,U\  &&\sim\  \text{Bernoulli}\left(\expit(2.5 W_1 + 0.5 W_2 W_3)\right), \\
    & C \mid T,W,U\  &&\sim\  \text{Bernoulli}\left(\expit(2T-1-W_1+0.2 W_2 + 0.7 W_3+2W_1W_2+0.5U)\right), \\
    & Y \mid T,C,W,U\   &&\sim\ \text{Bernoulli}\left(\expit(-0.3 C + C W_2 - W_1 + 0.2 W_2 - 0.9 W_3 + 0.3 C U)\right).
\end{alignat*}
We introduce $U$ in the data-generating mechanism to emphasize that we do not require assumptions on the joint distribution of treatment cost and outcome conditional on covariates. We consider all three reference ITRs $\mathcal{R} \in \{\FR,\RD,\TP\}$, where we set $\edecision^\FR: v \mapsto 0$. We set $\kappa=0.68$, which is an active constraint with $\tau_0>0$ and $\edecision^\RD_0 < 1$.

The ITRs we consider are based on all covariates --- that is, we take $V(W)=W$. We estimate the nuisance functions using the Super Learner \citep{VanderLaan2007} with library including a logistic regression, generalized additive model with logit link \citep{Hastie1990}, gradient boosting machine \citep{Friedman2001,Friedman2002,Mason1999,Mason2000}, support vector machine \citep{Bennett2000,Cortes1995} and neural network \citep{Bishop1995,Ripley2014}. Because none of the nuisance functions follow a logistic regression model, the resulting ensemble learner is not expected to achieve the parametric convergence rate. Since both $C$ and $Y$ are binary, we use logistic regression rather than ordinary least-squares to obtain their corresponding targeted estimates in Section~\ref{section: procedure}. We consider sample size $n\in\{500, 1000, 4000, 16000\}$, and run 1000 Monte Carlo repetitions for each sample size. We implement the algorithm that incorporates cross-fitting discussed in Remark~\ref{remark: sample splitting} and described in Section~\ref{section: sample splitting} in the Supplementary Material.

To evaluate the performance of our proposed estimator, we investigate the bias and root mean squared error (RMSE)
of the estimator. We also investigate the coverage probability and the width of nominal 95\% Wald CIs constructed using influence function-based standard error estimates. We further investigate the probability that our confidence lower limit falls below the true ATE, that is, the coverage probability of the 97.5\% Wald confidence lower bound.

\subsection{Simulation results}

Table~\ref{table: simulation} presents the performance of our proposed estimator in this simulation. For sample sizes 500, 1000 and 4000, the CI coverage of our proposed method is lower than the nominal coverage 95\%. When sample size is larger (16000), the CI coverage of our proposed method increases to 90--93\%. The coverage of the confidence lower bounds is much closer to nominal (97.5\%) for all sample sizes considered, though, and is always approximately nominal when the sample size is large. %
For all reference ITRs, the bias and RMSE
of our proposed estimator appear to converge to zero faster than and at the same rate as the square root of sample size, respectively. All biases are negative, which is expected in view of Remark~\ref{remark: sample splitting}.
All standard errors underestimate the variation of the estimator with the extent decreasing as sample size increases.

\begin{table}
    \caption{Performance of estimators of average treatment effects in the simulation with nuisance functions estimated via machine learning.}
    \label{table: simulation}
    \begin{center}
        \begin{tabular}{lr|r|r|r}
            Performance measure & Sample size & $\FR$ & $\RD$ & $\TP$ \\ \hline \hline
            95\% Wald CI coverage & 500 & $74\%$ & $71\%$ & $70\%$ \\
            & 1000 & $78\%$ & $74\%$ & $73\%$ \\
            & 4000 & $90\%$ & $84\%$ & $88\%$ \\
            & 16000 & $93\%$ & $90\%$ & $93\%$ \\ \hline
            97.5\% confidence lower & 500 & $94\%$ & $96\%$ & $96\%$ \\
            bound coverage & 1000 & $97\%$ & $98\%$ & $96\%$ \\
            & 4000 & $98\%$ & $98\%$ & $98\%$ \\
            & 16000 & $97\%$ & $98\%$ & $97\%$ \\ \hline
            bias & 500 & $-0.018$ & $-0.018$ & $-0.020$ \\
            & 1000 & $-0.014$ & $-0.013$ & $-0.013$ \\
            & 4000 & $-0.003$ & $-0.004$ & $-0.003$ \\
            & 16000 & $-0.000$ & $-0.001$ & $-0.000$ \\ \hline
            RMSE & 500 & $0.056$ & $0.039$ & $0.046$ \\
            & 1000 & $0.039$ & $0.025$ & $0.031$ \\
            & 4000 & $0.017$ & $0.009$ & $0.012$ \\
            & 16000 & $0.009$ & $0.004$ & $0.005$ \\
            \hline
            Ratio of mean standard error & 500 & $0.620$ & $0.620$ & $0.571$ \\
            to standard deviation & 1000 & $0.683$ & $0.673$ & $0.637$ \\
            & 4000 & $0.868$ & $0.765$ & $0.809$ \\
            & 16000 & $0.913$ & $0.870$ & $0.906$ \\
        \end{tabular}
    \end{center}
\end{table}

Figure~\ref{figure: CI width} presents the width of the  Wald CIs scaled by the square root of sample size $n$. Our theory indicates that the CI width should shrink at a root-$n$ rate, and our simulation results are consistent with this. There are some outlying cases of extremely wide or narrow CIs. This is expected for small sample sizes because the estimator of $\sigma_0^2$ in Theorem~\ref{e RC theorem: asymptotic linearity} resembles a sample mean and might not be close to $\sigma_0^2$ with high probability when sample size is small. In practice, this issue might be slightly mitigated by fine-tuning the involved machine learning algorithms.

\begin{figure}
    \begin{center}
        \includegraphics{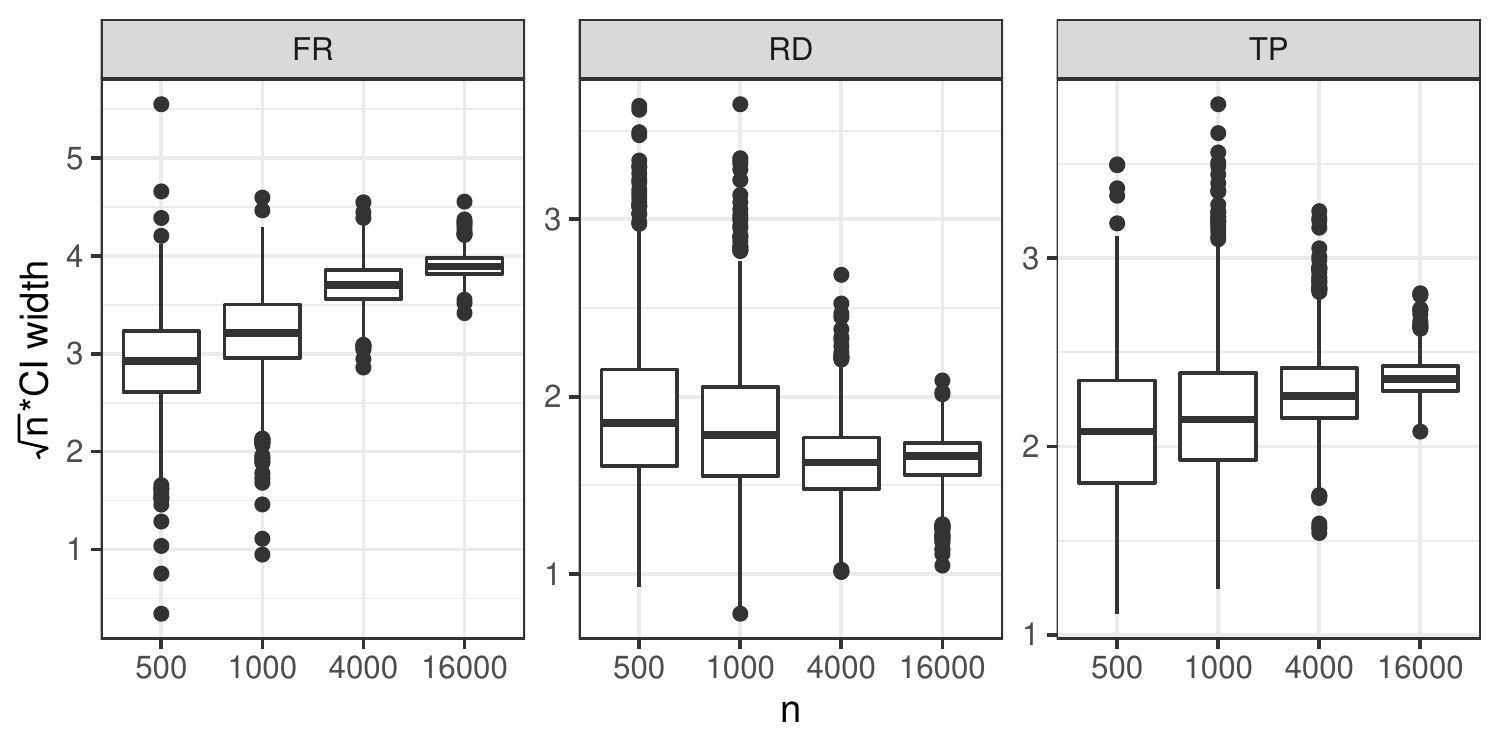}
    \end{center}
    \caption{Boxplot of $\sqrt{n} \times$ CI width for ATE relative to each reference ITR.\label{figure: CI width}}
\end{figure}

As indicated in Theorem~\ref{e RC theorem: asymptotic linearity}, theoretical guarantees for the validity of the Wald CIs rely on the nuisance function estimators converging to the truth sufficiently quickly. 
It appears that the undercoverage of our Wald CI in small samples may owe, in part, to poor estimation of these nuisance functions in small sample sizes. 
To illustrate how our procedure may perform with improved small-sample nuisance function estimators, 
we conducted another two simulations: one is identical to those reported earlier in all ways except that the nuisance function estimators $\Q^Y_n$, $\Q^C_n$ and $\Q^T_n$ are taken to be equal to the truth; the other is a simpler scenario under a lower dimension and a parametric model.
The results are presented in Section~\ref{section: simulation2} in the Supplementary Material and suggest that our proposed estimator may achieve significantly better performance with improved machine learning estimators of the nuisance functions.
This motivates seeking ways to optimize the finite-sample performance of the nuisance function estimators employed in future applications of the proposed method, possibly based on prior subject-matter expertise.
The underestimation of standard errors in this simulation also motivates future work exploring whether there are standard error estimators with better finite-sample performance, for example, estimators based on the bootstrap.

\section{Conclusion} \label{section: discussion}

There is an extensive literature on estimating optimal ITRs and evaluating their performance. Among these works, only a few incorporated treatment resource constraints. In this paper, we build upon \citet{Sun2021} and study the problem of estimating optimal ITRs under treatment cost constraints when the treatment cost is random. Using similar techniques as used in \citet{Qiu2021}, we have proposed novel methods to estimate an optimal ITR and infer about the corresponding average treatment effect relative to a prespecified reference ITR, under a locally nonparametric model. Our methods may also be applied to instrumental variable (IV) settings studied in \citet{Qiu2021} when the IV is intervened on.

\bibliographystyle{chicago}
\bibliography{references_v6}

\newpage

\setcounter{page}{1}
\renewcommand*{\theHsection}{S.\the\value{section}}
\setcounter{section}{0}
\renewcommand{\thesection}{S\arabic{section}}%
\setcounter{table}{0}
\renewcommand{\thetable}{S\arabic{table}}%
\setcounter{figure}{0}
\renewcommand{\thefigure}{S\arabic{figure}}%
\setcounter{equation}{0}
\renewcommand{\theequation}{S\arabic{equation}}%
\setcounter{lemma}{0}
\renewcommand{\thelemma}{S\arabic{lemma}}%
\setcounter{theorem}{0}
\renewcommand{\thetheorem}{S\arabic{theorem}}%
\setcounter{corollary}{0}
\renewcommand{\thecorollary}{S\arabic{corollary}}%

\begin{center}
    \LARGE Supplementary Material for ``Individualized treatment rules under stochastic treatment cost constraints''
\end{center}

This Supplementary Material is organized as follows. Section~\ref{section: additional technical conditions} contains technical conditions to ensure that the statistical parameter of interest, the average treatment effect, is pathwise differentiable and that our proposed estimator is asymptotically efficient. We discuss a particular technical condition that may be difficult to verify in Section~\ref{section: discussion of IER remainder}. In Section~\ref{section: sample splitting}, we describe a modified version of our proposed estimator with improved performance in small to moderate samples. We present proofs of theoretical results in Section~\ref{section: proof}. In Section~\ref{section: simulation2}, we present the results of a simulation under an idealized setting. These results may provide guidance on interpreting the simulation results in Section~\ref{section: simulation}.

As noted in the main text, the methods proposed in this work build upon tools used in \citet{Qiu2021}; as such, the involved technical details bear similarity.
To orient readers and facilitate comparisons, we have organized these supplementary materials for these papers similarly and shared portions of technical details when appropriate.

\section{Technical conditions for pathwise differentiability of parameter and asymptotic linearity of proposed estimator} \label{section: additional technical conditions}

In this section, we list the additional technical conditions required by Theorems~\ref{e RC theorem: differentiability} and \ref{e RC theorem: asymptotic linearity} in Section~\ref{section: estimators} that we omit in the main text. Before doing this, we define pointwise
\begin{align*}
    D_{n,\FR}(o) &:= D(\hat{P}_n,\edecision_n,\tau_0,\Q^C_n)(o) - D(\hat{P}_n,\edecision^\FR,0,\Q^C_0)(o) \ , \\
    D_{n,\RD}(o) &:= D(\hat{P}_n,\edecision_n,\tau_0,\Q^C_n)(o) - D(\hat{P}_n,\edecision^\RD_n,0,\Q^C_0)(o) \\
    &\hspace{0.5in}- \alpha \frac{\Psi_{\edecision^\RD_n}(\hat{P}_n)}{\kappa-\phi_n} D_1(\hat{P}_n,\Q^C_n)(o) - \frac{\Psi_{\edecision^\RD_n}(\hat{P}_n)}{P_n \hat{\Delta}^C_n} D_2(\hat{P}_n,\hat{\Q}^C_n)(o) \ , \\
    D_{n,\TP}(o) &:= D(\hat{P}_n,\edecision_n,\tau_0,\Q^C_n)(o) - G_\TP(\hat{P}_n)(o) \ .
\end{align*}

\begin{assumptionEncourage}[Nonzero continuous density of $\xi_0(V)$ around $\eta_0$] \label{e RC assumption: continuous density}
    If $\eta_0 > -\infty$, then the distribution of $\xi_0(V)$ has positive, finite and continuous Lebesgue density in a neighborhood of $\eta_0$.
\end{assumptionEncourage}

Since Condition~\ref{e RC assumption: continuous density} is most plausible when covariates are continuous, in this case, it is also plausible to expect the distribution of $\xi_0(V)$ to be continuous and thus Condition~\ref{e RC assumption: continuous density} holds.

\begin{assumptionEncourage}[Smooth treatment cost function or lack of constraint] \label{e RC assumption: continuous weight}
    If $\eta_0 > -\infty$, then the function $\eta \mapsto \expect_0\left[I\left(\xi_0(V)>\eta\right) \Delta^C_0(W)\right]$ is continuously differentiable with nonzero derivative in a neighborhood of $\eta_0$; if $\eta_0 = -\infty$ and $\kappa<\infty$, then $\expect_0\left[\Delta^C_0(W)\right] < \kappa - \alpha \phi_0$.
\end{assumptionEncourage}

Condition~\ref{e RC assumption: continuous weight} requires different conditions in separate cases. There are three cases in terms of the sufficiency of the budget to treat every individual: (i) there is an infinite budget and no constraint is present ($\kappa=\infty$); (ii) the budget is insufficient ($\eta_0>-\infty$); and (iii) the budget is finite but sufficient ($\eta_0=-\infty$ and $\kappa<\infty$). Condition~\ref{e RC assumption: continuous weight} makes no assumption for Case~(i). In Case~(ii), we require a function $\eta \mapsto \expect_0\left[I\left(\xi_0(V)>\eta\right) \Delta^C_0(W)\right]$ to be locally continuously differentiable. Since $\Delta_0^C > 0$ by Condition~\ref{IV assumption: strong encouragement}, this function is nonincreasing and thus only continuous differentiability is required. For each $\eta$, this function is an integral of additional cost $\Delta_0^C$ over the set $\{v: \xi_0(v)>0\}$ and has a similar nature to survival functions. When covariates are continuous, it is plausible to assume that $\Delta_0^C(W)$ is continuous and thus $\eta \mapsto \expect_0\left[I\left(\xi_0(V)>\eta\right) \Delta^C_0(W)\right]$ is continuously differentiable. In Case~(iii), we require that the budget has a surplus. When it is unknown \textit{a priori} whether the budget is sufficient to treat every individual, namely in Case~(ii) or (iii), it is highly unlikely that the budget exactly suffices with no surplus. Therefore, Condition~\ref{e RC assumption: continuous weight} is mild.

\begin{assumptionEncourage}[Bounded additional treatment cost] \label{e RC assumption: bounded encouragement effect}
    $\Delta^C_0$ is bounded.
\end{assumptionEncourage}

\begin{assumptionEncourage}[Active constraint] \label{e RC assumption2: active constraint}
    If $\mathcal{R}=\RD$, then it holds that $(\kappa-\alpha\phi_0)/\expect_0\left[\Delta^C_0(W)\right] < 1$.
\end{assumptionEncourage}

Condition~\ref{e RC assumption2: active constraint} requires that, when the rule $\edecision^\RD$ that assigns treatment completely at random while respecting the budget constraint is the reference rule of interest, it should not correspond to the trivial rule $v \mapsto 1$ that assigns treatment to every individual. The rule $\edecision^\RD$ equals $v \mapsto 1$ only when the budget is sufficient to treat every individual. Since, as a separate reference rule from given fixed rules $\edecision^\FR$, the reference rule $\edecision^\RD$ is only interesting when the budget constraint is active, Condition~\ref{e RC assumption2: active constraint} often holds automatically.

\begin{assumptionEncourage}[Sufficient rates for nuisance estimators] \label{e RC assumption: remainder}
    \begin{align*}
        \| \Q^T_n - \Q^T_0 \|_{2,P_0} \Big\{& \| \Q^Y_n - \Q^Y_0 \|_{2,P_0} + \| \hat{\Q}^Y_n - \Q^Y_0 \|_{2,P_0}  \\
        &+ \| \Q^C_n - \Q^C_0 \|_{2,P_0} + \| \hat{\Q}^C_n - \Q^C_0 \|_{2,P_0} \Big\} =\smallo_p(n^{-1/2}) \ .
    \end{align*}
\end{assumptionEncourage}

Condition~\ref{e RC assumption: remainder} holds if all above nuisance estimators converge at a rate faster than $n^{-1/4}$, which may be much slower than the parametric rate $n^{-1/2}$ and thus allows for the use of flexible nonparametric estimators. This condition also holds if $\Q^Y_n$, $\hat{\Q}^Y_n$, $\Q^C_n$ and $\hat{\Q}^C_n$ each converges slower than $n^{-1/4}$, as long as the estimated propensity score $\Q^T_n$ converges sufficiently fast to compensate.

\begin{assumptionEncourage}[Consistency of estimated influence function] \label{e RC assumption: consistency of estimated influence function}
    The following terms are all $\smallo_p(1)$:
    \begin{align*}
        & \| D_1(\hat{P}_n,\hat{\Q}^C_n)-D_1(P_0,\Q^C_0) \|_{2,P_0} \ , \quad \| D_2(\hat{P}_n,\Q^C_n)-D_2(P_0,\Q^C_0) \|_{2,P_0} \ , \quad \| D_{n,\mathcal{R}} - D_\mathcal{R}(P_0) \|_{2,P_0} \ , \\
        & \| [D(\hat{P}_n,\edecision_n,\tau_0,\Q^C_n) - D(\hat{P}_n,\edecision^\RD_n,0,\Q^C_0)] - [D(P_0,\edecision_0,\tau_0,\Q^C_0) - D(P_0,\edecision^\RD_0,0,\Q^C_0)] \|_{2,P_0} \ .
    \end{align*}
\end{assumptionEncourage}

\begin{assumptionEncourage}[Consistency of strong positivity] \label{e assumption: consistency of strong IV positivity}
    With probability tending to one over the sample used to obtain $\Q^T_n$, it holds that $\int I\{\epsilon_T<\Q^T_n(w) < 1-\epsilon_T\} dP_0(w)=1$.
\end{assumptionEncourage}

\begin{assumptionEncourage}[Consistency of strictly more costly treatment] \label{e RC assumption: consistency of strong encouragement}
    With probability tending to one over the sample used to obtain $\Delta^C_n$ and $\delta^C_{n}$, it holds that $\int I(\Delta^C_n(w) > \delta_C) dP_0(w) = 1$ and $\int I(\delta^C_{n}(v) > \delta_C) dP_0(v) = 1$.
\end{assumptionEncourage}

\begin{assumptionEncourage}[Fast rate of estimated optimal ITR] \label{e RC assumption: IER remainder}
    As sample size $n$ tends to infinity, it holds that $$\int \left\{ \edecision_n(v)-\edecision_0(v) \right\} \left\{\delta^Y_{0}(v)-\tau_0 \delta^C_{0}(v)\right\} dP_0(v)=\smallo_p(n^{-1/2})\ .$$
\end{assumptionEncourage}
Condition~\ref{e RC assumption: IER remainder} may, at first sight, appear to be difficult to verify and is discussed in detail in Section~\ref{section: discussion of IER remainder}. As shown in Theorem~\ref{theorem: remainder due to estimating optimal rule} of Section~\ref{section: discussion of IER remainder}, Condition~\ref{e RC assumption: IER remainder} may require faster rates on nuisance estimators than Condition~\ref{e RC assumption: remainder}. For example, convergence in the $L^2$-sense at a rate $\smallo_p(n^{-1/4})$ is sufficient for Condition~\ref{e RC assumption: remainder}, but a rate $\smallo_p(n^{-3/8})$ is needed in order to use Theorem~\ref{theorem: remainder due to estimating optimal rule} to show that Condition~\ref{e RC assumption: IER remainder} holds.

\begin{assumptionEncourage}[Donsker condition] \label{e RC assumption: Donsker condition}
    $\{ o\mapsto d_{n,k}(v) D_2(\hat{P}_n,\Q^C_n)(o): k \in [0,1] \}$ is a subset of a fixed $P_0$-Donsker class with probability tending to 1. Additionally, each of $D_1(\hat{P}_n,\Q^C_n)$, $D(\hat{P}_n,\edecision_n,\tau_0,\Q^C_n) - D(\hat{P}_n,\edecision^\RD_n,0,\Q^C_0)$ and $D_{n,\mathcal{R}}$ belongs to a (possibly different) fixed $P_0$-Donsker class with probability tending to 1.
\end{assumptionEncourage}

\begin{assumptionEncourage}[Glivenko-Cantelli condition] \label{e RC assumption: Glivenko-Cantelli}
    $\| \xi_n - \xi_0 \|_{1,P_0} = \smallo_p(1)$ and $\| \Delta^C_n - \Delta^C_0 \|_{1,P_0} = \smallo_p(1)$. Moreover, (i) if $\eta_0 > -\infty$, then, for any $\eta$ sufficiently close to $\eta_0$, $w \mapsto I(\xi_n(v)>\eta) \Delta^C_n(w)$ belongs to a $P_0$-Glivenko-Cantelli class with probability tending to 1; (ii) otherwise, if $\eta_0 = -\infty$, then, for any $\eta < 0$ with sufficiently large $|\eta|$, $w \mapsto I(\xi_n(v)>\eta) \Delta^C_n(w)$ belongs to a $P_0$-Glivenko-Cantelli class with probability tending to 1.
\end{assumptionEncourage}

The Donsker condition~\ref{e RC assumption: Donsker condition} and the Glivenko-Cantelli condition~\ref{e RC assumption: Glivenko-Cantelli} impose restrictions on the flexibility of the methods used to estimate nuisance functions. We refer readers to, for example, \citet{vandervaart1996}, for a more thorough introduction to such conditions.

All above conditions are similar to those in \citet{Qiu2021} except that Conditions~\ref{e RC assumption: consistency of strong encouragement} and \ref{e RC assumption: bounded encouragement effect} are additional in this paper because the assumption of more costly treatment was not needed and a boundedness condition similar to \ref{e RC assumption: bounded encouragement effect} was automatically satisfied with a binary cost.

\section{Sufficient condition for fast convergence rate of estimated optimal rule} \label{section: discussion of IER remainder}

Condition~\ref{e RC assumption: IER remainder}, which is required by Theorem~\ref{e RC theorem: asymptotic linearity}, may seem unintuitive and difficult to verify. In Theorem~\ref{theorem: remainder due to estimating optimal rule} below, we present sufficient conditions for Condition~\ref{e RC assumption: IER remainder} that are similar to those in \citeauthor{Qiu2021} \citep{Qiu2021}.

Throughout the rest of the Supplement, for two quantities $a,b \in \real$, we use $a \lesssim b$ to denote $a \leq \const b$ for some constant $\const>0$ that may depend on $P_0$.

\begin{theorem}[Sufficient condition for Condition~\ref{e RC assumption: IER remainder}] \label{theorem: remainder due to estimating optimal rule}
    Assume that $\int I(\xi_n(v) = \tau_n) dP_0(v) = \bigO_p(n^{-1/2})$. Further assume that each of $o \mapsto I(\xi_n(v) > \eta_n)$ and $o \mapsto I(\xi_n(v) > \eta_n)\delta^C_{0}(v)$ belongs to a (possibly different) fixed $P_0$-Donsker class with probability tending to 1. Suppose also that the distribution of $\xi_0(V)$ ($V\sim P_0$) has nonzero finite continuous Lebesgue density in a neighborhood of $\eta_0$ and a neighborhood of $\tau_0$. Under Condition~\ref{e RC assumption: bounded encouragement effect}, the following statements hold.
    \begin{itemize}
        \item If $\| \delta^Y_{n} - \delta^Y_{0} \|_{q,P_0} = \smallo_p(1)$ for some $q \ge 1$, then
        \begin{align*}
            | P_0 \{ (\edecision_n - \edecision_0) (\delta^Y_{0} - \tau_0 \delta^C_{0}) \} | \lesssim \| \delta^Y_{n} - \delta^Y_{0} \|_{q,P_0}^{2q/(q+1)} + \bigO_p(n^{-1}).
        \end{align*}
        \item If $\| \delta^Y_{n} - \delta^Y_{0} \|_{\infty,P_0} = \smallo_p(1)$, then
        \begin{align*}
            | P_0 \{ (\edecision_n - \edecision_0) (\delta^Y_{0} - \tau_0 \delta^C_{0}) \} | \lesssim \| \delta^Y_{n} - \delta^Y_{0} \|_{\infty,P_0}^2 + \bigO_p(n^{-1}).
        \end{align*}
    \end{itemize}
\end{theorem}

The proof of Theorem~\ref{theorem: remainder due to estimating optimal rule} is very similar to Theorem~5 in \citeauthor{Qiu2021} \citep{Qiu2021} and can be found in Section~\ref{section: IER/ITR remainder proof}.

\section{Modified procedure with cross-fitting} \label{section: sample splitting}

In this section, we describe our proposed procedure to estimate the ATE with cross-fitting, which is mentioned in Remark~\ref{remark: sample splitting}. We use $\Lambda$ to denote a user-specified fixed number of folds to split the data. Common choices of $\Lambda$ used in practice include 5, 10 and 20.

\begin{enumerate}
    \item Use the empirical distribution $\hat{P}_{W,n}$ of $W$ as an estimate of the true marginal distribution of $W$. Compute estimates $\Q^Y_n$, $\Q^C_n$ and $\Q^T_n$ of $\Q^Y_0$, $\Q^C_n$ and $\Q^T_0$, respectively using flexible regression methods.

    \item Estimate an optimal individualized treatment rule for each observation:
    \begin{enumerate}
        \item \label{it:create folds} Create folds: split the set of observation indices $\{1,2,\ldots,n\}$ into $\Lambda$ mutually exclusive and exhaustive folds of (approximately) equal size. Denote these sets by $S_\lambda$, $\lambda=1,2,\ldots,\Lambda$. Define $S_{-\lambda}:=\cup_{\lambda' \neq \lambda} S_{\lambda'}$. For each $i=1,2,\ldots,n$, let $\lambda(i)$ be the index of the fold containing $i$; in other words, $\lambda(i)$ is the unique value of $\lambda$ such that $i \in S_{\lambda}$.
        \item Estimate $\xi_0(V_i)$ using sample splitting: for each $\lambda=1,2,\ldots,\Lambda$, compute estimates $\delta^Y_{n,S_{-\lambda}}$ and $\delta^C_{n,S_{-\lambda}}$ of $\delta^Y_{0}$ and $\Delta^C_{0,b}$ using flexible regression methods based on data $\{O_i: i \in S_{-\lambda}\}$. For each $i=1,2,\ldots,n$, let $\xi_{n,i} := \delta^Y_{n,S_{-\lambda(i)}}(V_i)/\delta^C_{n,S_{-\lambda(i)}}(V_i)$ be the sample splitting estimate of $\xi_0(V_i)$.
        \item Estimate $\phi_0$ with a one-step correction estimator
        $$\phi_n := \frac{1}{n} \sum_{i=1}^n \{\Q^C_n(0,W_i) + \frac{1-T_i}{1-\Q^T_n(W_i)} [C_i - \Q^C_n(0,W_i)]\}.$$
        \item Let $\Gamma_n : \tau \mapsto \frac{1}{n}\sum_{i : \xi_{n,i}>\tau} \Delta^C_n(1,W_i)$ and $\gamma_n : \tau\mapsto \frac{1}{n}\sum_{i : \xi_{n,i}=\tau} \Delta^C_n(W_i)$. For any $k \in [0,\infty)$, define $\eta_n(k) := \inf \{\tau: \Gamma_n(\tau) \leq k - \alpha \phi_n \}$, $\tau_n(k) := \max\{\eta_n(k),0\}$, and, for $i=1,2,\ldots,n$,
        $$d_{n,k,i} := \begin{cases}
            \frac{k - \alpha \phi_n - \Gamma_n(\eta_n(k))}{ \gamma_n(\eta_n(k))}, & \text{ if } \xi_{n,i}=\eta_n(k) \text{ and } \gamma_n(\eta_n(k)) > 0, \\
            I\{\xi_{n,i}>\eta_n(k)\}, & \text{ otherwise.}
        \end{cases}$$
        \item Compute $k_n$, which is used to define an estimate of $\edecision_0$ for which the plug-in estimator is asymptotically linear.
        \begin{itemize}
            \item If $\tau_n (\kappa) > 0$ \textit{and} there is a solution in $k \in [0,\infty)$ to
            \begin{equation} \label{e RC equation: sample splitting update quantile}
                \frac{1}{n}\sum_{i=1}^n  d_{n,k,i} \left[ \Delta^C_n(W_i) + \frac{1}{T_i+\Q^T_n(W_i)-1} [C_i - \Q^C_n(T_i,W_i)] \right] + \alpha \phi_n = \kappa,
            \end{equation}
            then take $k_n$ to be this solution.
            \item otherwise, set $k_n=\kappa$.
        \end{itemize}
        \item For each $i=1,2,\ldots,n$, estimate $\edecision_0(V_i)$ with
        $$\edecision_{n,i} := \begin{cases}
            \frac{k_n - \alpha \phi_n - \Gamma_n(\tau_n(k_n))}{\gamma_n(\tau_n(k_n))}, & \text{ if } \xi_{n,i}=\tau_n(k_n),\; \text{ and } \gamma_n(\tau_n(k_n))>0, \\
            I\{\xi_{n,i}>\tau_n(k_n)\}, & \text{ otherwise.}
        \end{cases}$$
    \end{enumerate}
    
    \item Obtain an estimate $\edecision^\mathcal{R}_n$ of the reference ITR $\edecision^\mathcal{R}_0$ as follows:
    \begin{itemize}
        \item For $\mathcal{R}=\FR$, take $\edecision^\mathcal{R}_n$ to be $\edecision^\FR$.
        \item For $\mathcal{R}=\RD$,
        \begin{enumerate}
            \item obtain a targeted estimate $\hat{\Q}^C_n$ of $\Q^C_0$: run an ordinary least-squared regression using observations $i=1,2,\ldots,n$ with outcome $C_i$, offset $\Q^C_n(T_i,W_i)$, no intercept and covariate $1/(T_i+\Q^T_n(W_i)-1)$. Take $\hat{\Q}^C_n$ to be the fitted mean model;
            \item take $\edecision^\mathcal{R}_n$ to be the constant function $o\mapsto \min \{1,(\kappa-\alpha\phi_n)/\hat{P}_{W,n} \hat{\Delta}^C_n\}$, where we define pointwise $\hat{\Delta}^C_n: w \mapsto \hat{\Q}^C_n(1,w) - \hat{\Q}^C_n(0,w)$.
        \end{enumerate}
        \item For $\mathcal{R}=\TP$, take $\edecision^\mathcal{R}_n$ to be $\Q^T_n$.
    \end{itemize}

    \item Estimate ATE of $\edecision_0$ relative to the reference ITR $\edecision^\mathcal{R}_0$ with a targeted minimum-loss based estimator (TMLE) $\psi_n$:
    \begin{enumerate}
        \item obtain a targeted estimate $\hat{\Q}^Y_n$ of $\Q^Y_0$: run an ordinary least-squares linear regression using observations $i=1,2,\ldots,n$ with outcome $Y_i$, offset $\Q^Y_n(T_i,W_i)$, no intercept and covariate $[\edecision_{n,i}-\edecision^\mathcal{R}_n(O_i)]/[T_i + \Q^T_n(W_i) - 1]$. Take $\hat{\Q}^Y_n$ to be the fitted mean function.
        
        \item with $\hat{P}_n$ being any distribution with components $\hat{\Q}^Y_n$ and $\hat{P}_{W,n}$, set $\psi_n:=\frac{1}{n} \sum_{i=1}^n \edecision_{n,i} \hat{\Delta}^Y_n(W_i) - \Psi_{\edecision^\mathcal{R}_n}(\hat{P}_n)$ where $\hat{\Delta}^Y_n: w \mapsto \hat{\Q}^Y_n(1,w)-\hat{\Q}^Y_n(0,w)$.
    \end{enumerate}
\end{enumerate}

\section{Proof of theorems} \label{section: proof}

\subsection{Identification results (Theorem~\ref{theorem: identify ATE} and \ref{theorem: true optimal IER})}

Theorem~\ref{theorem: identify ATE} is a simple corollary of the standard G-formula \citep{Robins1986}. We provide a complete proof below.
\begin{proof}[Proof of Theorem~\ref{theorem: identify ATE}]
    Note that
    \begin{align*}
        \cexpect[Y(1)  \mid  W] = \cexpect[Y(1) \mid T=1, W] = \expect_0[Y  \mid T=1,W] = \Q^Y_0(1,W).
    \end{align*}
    Similarly, $\cexpect[Y(0) \mid W] = \expect_0[Y \mid T=0,W] = \Q^Y_0(0,W)$. Hence, $\cexpect[Y(1)-Y(0) \mid W]=\Delta^Y_0(W)$. By the law of total expectation, this yields that $\cexpect[Y(1)-Y(0) \mid V]=\expect_0[\Delta^Y_0(W) \mid V] = \delta^Y_{0}(V)$. It then follows that
    \begin{align*}
        \cexpect[Y(\edecision) - Y(\edecision^\mathcal{R}_0)]&= \cexpect[\{\edecision(V)-\edecision^\mathcal{R}_0(W)\} \{Y(1) - Y(0)\}] \\
        &= \expect_0[\{\edecision(V)-\edecision^\mathcal{R}_0(W)\} \cexpect[Y(1) - Y(0) \mid W] ] \\
        &= \expect_0[\{\edecision(V)-\edecision^\mathcal{R}_0(W)\} \Delta^Y_0(W)].
    \end{align*}
    The results for the treatment cost can be proved similarly.
\end{proof}

We next prove Theorem~\ref{theorem: true optimal IER}.
\begin{proof}[Proof of Theorem~\ref{theorem: true optimal IER}]
    Let $\edecision$ be any ITR that satisfies the constraint that $\expect_0[\edecision(V)\delta^C_{0}(V)] + \alpha \phi_0 \le \kappa$. We will show that $\expect_0[\edecision_0(V) \delta^Y_{0}(V)] \ge \expect_0[\edecision(V) \delta^Y_{0}(V)]$, implying that $\edecision_0$ is a solution to \eqref{eq:ATEknapsack}.
    
    Observe that
    \begin{align*}
        & \expect_0[\edecision_0(V) \delta^Y_{0}(V)] - \expect_0[\edecision(V) \delta^Y_{0}(V)] \\
        &= \expect_0[\{\edecision_0(V)-\edecision(V)\} \delta^Y_{0}(V)] \\
        &= \expect_0[\{\edecision_0(V)-\edecision(V)\} \delta^Y_{0}(V) I(\xi_0(V) > \tau_0)] + \expect_0[\{\edecision_0(V)-\edecision(V)\} \delta^Y_{0}(V) I(\xi_0(V) < \tau_0)] \\
        &\hspace{.5in}+ \expect_0[\{\edecision_0(V)-\edecision(V)\} \delta^Y_{0}(V) I(\xi_0(V) = \tau_0)] \\
        &=\expect_0[\{\edecision_0(V)-\edecision(V)\} \xi_0(V) \delta^C_{0}(V) I(\xi_0(V) > \tau_0)] + \expect_0[\{\edecision_0(V)-\edecision(V)\} \xi_0(V) \delta^C_{0}(V) I(\xi_0(V) < \tau_0)] \\
        &\hspace{.5in}+ \expect_0[\{\edecision_0(V)-\edecision(V)\} \xi_0(V) \delta^C_{0}(V) I(\xi_0(V) = \tau_0)].
    \end{align*}
    Note that $\edecision_0(v) = 1 \geq \edecision(v)$ if $\xi_0(v) > \tau_0$ and $\edecision_0(v) = 0 \leq \edecision(v)$ if $\xi_0(v) < \tau_0$. Combining this observation with the fact that $\tau_0\ge 0$, the above shows that
    \begin{align*}
        & \expect_0[\edecision_0(V) \delta^Y_{0}(V)] - \expect_0[\edecision(V) \delta^Y_{0}(V)] \\
        &\geq \tau_0 \expect_0[\{\edecision_0(V)-\edecision(V)\} \delta^C_{0}(V) I(\xi_0(V) > \tau_0)] + \tau_0 \expect_0[\{\edecision_0(V)-\edecision(V)\} \delta^C_{0}(V) I(\xi_0(V) < \tau_0)] \\
        &\hspace{.5in}+ \tau_0 \expect_0[\{\edecision_0(V)-\edecision(V)\} \delta^C_{0}(V) I(\xi_0(V) = \tau_0)] \\
        &= \tau_0 \expect_0[\{\edecision_0(V)-\edecision(V)\} \delta^C_{0}(V)].
    \end{align*}
    If $\tau_0 = 0$, then $\expect_0[\edecision_0(V) \delta^Y_{0}(V)] - \expect_0[\edecision(V) \delta^Y_{0}(V)] \geq 0$, as desired; otherwise, $\tau_0 > 0$ and $\expect_0[\edecision(V) \delta^C_{0}(V)]\leq \kappa - \alpha \phi_0 = \expect_0[\edecision_0(V) \delta^C_{0}(V)]$, and so it follows that $\expect_0[\edecision_0(V) \delta^Y_{0}(V)] \geq \expect_0[\edecision(V) \delta^Y_{0}(V)]$. Therefore, we conclude that $\edecision_0$ is a solution to \eqref{eq:ATEknapsack}.
\end{proof}

\subsection{Pathwise differentiability of ATE parameter (Theorem~\ref{e RC theorem: differentiability})} \label{section: proof differentiability}

We follow existing literature on semiparametric efficiency theory closely to prove pathwise differentiability of our estimands and asymptotic efficiency of our estimators under nonparametric models. We refer readers to, for example, \citet{Pfanzagl1982,Pfanzagl1990,Bolthausen2002}, for a more thorough introduction to semiparametric efficiency.

To derive the canonical gradient of the ATE parameters, let $\mathcal{H} \subseteq L^2_0(P_0)$ be the set of score functions with range contained in $[-1,1]$ and we study the behavior of the parameters under perturbations in an arbitrary direction $H \in \mathcal{H}$. We note that the $L^2_0(P_0)$-closure of $\mathcal{H}$ is indeed $L^2_0(P_0)$.

We define $H_W: w \mapsto \expect_0[H(O) \mid W=w]$, $H_{T}: (t \mid w) \mapsto \expect_0[H(O) \mid T=t,W=w]$ and $P_{H,\epsilon}$ via its Radon-Nikodym derivative with respect to $P_0$:
\begin{align}
    \frac{dP_{H,\epsilon}}{dP_0} : o\mapsto \left[1+\epsilon H(o)-\epsilon H_{T}(t \mid w)-\epsilon H_{W}(w)\right] \left[1+\epsilon H_{T}(t \mid w)\right] \left[1+\epsilon H_{W}(w)\right] \label{eq:submodel}
\end{align}
for any $\epsilon$ in a sufficiently small neighborhood of 0 such that the right-hand side is positive for all $o \in \mathcal{W} \times \{0,1\} \times \{0,1\} \times \real$. It is straightforward to verify that the score function for $\epsilon$ at $\epsilon=0$ is indeed $H$. For the rest of this section, we may drop $H$ from the notation and use $P_\epsilon$ as a shorthand notation for $P_{H,\epsilon}$ when no confusion should arise.

We will see that each parameter evaluated at $P_\epsilon$ depends on the following marginal or conditional distributions in a clean way: the marginal distribution $P_{W,\epsilon}$ of $W$, the marginal distribution $P_{T,W,\epsilon}$ of $(T,W)$, the conditional distribution $P_{T,\epsilon}$ of $T$ given $W$, the conditional distribution $P_{C,\epsilon}$ of $C$ given $(T,W)$, and the conditional distribution $P_{Y,\epsilon}$ of $Y$ given $(T,W)$. We now derive their closed-form expressions. Let $H_C : (c \mid t,w)\mapsto \expect_0[H(O) \mid C=c,T=t,W=w] - H_{T}(t \mid w) - H_W(w)$, and $H_Y : (y \mid t,w)\mapsto \expect_0[H(O) \mid Y=y,T=t,W=w] - H_{T}(t \mid w) - H_W(w)$. We can then show that
\begin{align}
    \begin{split}
        \frac{dP_{W,\epsilon}}{dP_{W,0}} &: w\mapsto 1 + \epsilon H_W(w), \\
        \frac{dP_{T,W,\epsilon}}{dP_{T,W,0}} &: (t,w)\mapsto 1 + \epsilon H_T(t \mid w) + \epsilon H_{W}(w), \\
        \frac{dP_{T,\epsilon}}{dP_{T,0}}(\, \cdot\,  \mid w) &: t\mapsto 1 + \epsilon H_T(t \mid w), \\
        \frac{dP_{C,\epsilon}}{dP_{C,0}}(\, \cdot\,  \mid t,w) &: c\mapsto 1 + \epsilon H_C(c \mid t,w), \\
        \frac{dP_{Y,\epsilon}}{dP_{Y,0}}(\, \cdot\,  \mid t,w) &: y\mapsto 1 + \epsilon H_Y(y \mid t,w).
        \label{eq:epsMarginalsConditionals}
    \end{split}
\end{align}
Moreover, $\expect_0[H_W(W)]=0$, $\expect_0[H_{T}(T \mid W) \mid W]=0$ $P_0$-a.s., $\expect_0[H_C(C \mid T,W) \mid T,W]=0$ $P_0$-a.s., and $\expect_0[H_Y(Y \mid T,W) \mid T,W]=0$ $P_0$-a.s.

We finally introduce some additional notations that are used for the rest of the section. We use $\const$ to denote a generic positive constant that may vary line by line. Let $S_0$ be the survival function of the distribution of $\xi_0(V)$ when $V\sim P_0$. We also use the notation $\lesssim$ defined in Section~\ref{section: discussion of IER remainder}. For a generic function $f : \mathbb{R}\rightarrow\mathbb{R}$, we will use the big- and little-oh notations, namely $\bigO(f(\epsilon))$ and $\smallo(f(\epsilon))$, respectively, to denote the behavior of $f(\epsilon)$ as $\epsilon\rightarrow 0$. Finally, for a general function or quantity $f_P$ that depends on a distribution $P$, we use $f_\epsilon$ to denote $f_{P_\epsilon}$. For example, we may write $\Q^Y_\epsilon$ as a shorthand for $\Q^Y_{P_\epsilon}$. We will also write expectations under $P_\epsilon$ as $\expect_\epsilon$.

The derivation of the canonical gradients of $P \mapsto \Psi_{\edecision^\FR}(P)$ can be found in the Supplement of \citeauthor{Qiu2021} \citep{Qiu2021}. We now derive the canonical gradients of $P \mapsto \Psi_{\edecision^\TP_P}(P)$, $P \mapsto \Psi_{\edecision^\RD_P}(P)$ and $P \mapsto \Psi_{\edecision_P}(P)$, which are different from the parameters in \citet{Qiu2021}.

\subsubsection{Canonical gradient of \texorpdfstring{$P \mapsto \Psi_{\edecision^\TP_P}(P)$}{true propensity reference rule mean outcome} (Theorem~\ref{e RC theorem: differentiability})}

Fix a score $H\in\mathscr{H}$.  Note that, for all $P\in\mathscr{M}$,  $\Psi_{\edecision^\TP_P}(P) = \int \Q^T_P(w) \Delta^Y_P(w) P_{W}(dw)$. Combining this, \eqref{eq:epsMarginalsConditionals} and the chain rule yields that
\begin{align*}
    & \left. \frac{d}{d \epsilon} \Psi_{\edecision^\TP_{\epsilon}}(P_\epsilon) \right|_{\epsilon=0} \\
    &= \int \left. \frac{d}{d \epsilon} \left[\Q^T_{\epsilon}(w) \Delta^Y_{\epsilon}(w) P_{W,\epsilon}(dw)\right] \right|_{\epsilon=0} \\
    &= \int \left( \left. \frac{d}{d \epsilon} \Q^T_{\epsilon}(w) \right|_{\epsilon=0} \right) \Delta^Y_0(w) P_{W,0}(dw) + \int \Q^T_0(w) \left( \left. \frac{d}{d \epsilon} \Delta^Y_{\epsilon}(w) \right|_{\epsilon=0} \right) P_{W,0}(dw) \\
    &\quad+ \int \Q^T_0(w) \Delta^Y_0(w) \left. \frac{d}{d \epsilon} P_{W,\epsilon}(dw)\right|_{\epsilon=0} \\
    &= \iint (t-\Q^T_0(w)) H_T(t \mid w) \Delta^Y_0(w) P_{T,0}(dt \mid w) P_{W,0}(dw) \\
    &\quad+ \iiint \Q^T_0(w) \left( \frac{I(t=1)}{\Q^T_P(w)} - \frac{I(t=0)}{1-\Q^T_P(w)} \right) (y-\Q^Y_0(t,w)) H_Y(y \mid t,w) P_{Y,0}(dy \mid t,w) P_{T,0}(dt \mid w) P_{W,0}(dw) \\
    &\quad+ \int (\Q^T_0(w) \Delta^Y_0(w) - \Psi_{\edecision^\TP_0}(P_0)) H_{W}(w) P_{W,0}(dw) \\
    &= \int G_\TP(P_0)(o) H(o) P_0(do),
\end{align*}
where we have used the fact that $\expect_0[H_Y(Y \mid T,W) \mid T,W]=0$ $P_0$-a.s., $\expect_0[H_T(T \mid W) \mid W]=0$ $P_0$-a.s., and $\expect_0[H_{W}(W)]=0$. Therefore, the canonical gradient of $P\mapsto \Psi_{\edecision^\TP_P}(P)$ at $P_0$ is $G_\TP(P_0)$.

\subsubsection{Canonical gradient of \texorpdfstring{$P \mapsto \Psi_{\edecision^\RD_P}(P)$}{randomly distributed reference ITR mean outcome}}

Let $H$ be a score function in $\mathcal{H}$. We aim to show that
\begin{align}
    \left.\frac{d}{d\epsilon}\Psi_{\edecision^\RD_{\epsilon}}(P_\epsilon)\right|_{\epsilon=0} = \int G_\RD(P_0)(o) H(o) P_0(do), \label{eq:RDlim}
\end{align}
which shows that $P \mapsto \Psi_{\edecision^\RD_P}(P)$ is pathwise differentiable with canonical gradient $G_\RD(P_0)$ at $P_0$.

By similar arguments to those in Section~3.4 of \citeauthor{Kennedy2016} \citep{Kennedy2016}, we can show that
\begin{align}
    \left. \frac{d}{d\epsilon} P_\epsilon \Q^C_{\epsilon}(0,\cdot) \right|_{\epsilon=0}&= \int \left\{ \frac{1-t}{1-\Q^T_{0}(w)} [c-\Q^C_{0}(0,w)] + \Q^C_{0}(0,w) - P_0 \Q^C_{0}(0,\cdot) \right\} H(o) P_0(do), \label{eq:resourcegrad1} \\
    \left. \frac{d}{d\epsilon} P_\epsilon \Delta^C_{\epsilon} \right|_{\epsilon=0}&= \int \left\{ \frac{1}{t+\Q^T_{0}(w)-1} [c-\Q^C_{0}(t,w)] + \Delta^C_{0}(w) - P_0 \Delta^C_{0} \right\} H(o) P_0(do). \label{eq:resourcegrad2}
\end{align}
Consequently, $P_\epsilon \Q^C_{\epsilon}(0,\cdot)= P_0 \Q^C_{0}(0,\cdot) + \bigO(\epsilon)$ and $P_\epsilon \Delta^C_{\epsilon}= P_0 \Delta^C_{0} + \bigO(\epsilon)$. It follows that, for all $\epsilon$ in a sufficiently small neighborhood of zero, Condition~\ref{e RC assumption2: active constraint} implies that $ (\kappa-\alpha P_\epsilon \Q^C_{\epsilon}(0,\cdot))/P_\epsilon \Delta^C_{\epsilon} < 1$. Consequently, for each $\epsilon$ in this neighborhood, $\Psi_{\edecision^\RD_{\epsilon}}(P_\epsilon) = \frac{\kappa - \alpha P_\epsilon \Q^C_{\epsilon}(0,\cdot)}{P_\epsilon \Delta^C_{\epsilon}} \Psi_{v\mapsto 1}(P_\epsilon)$, where we have used that $P_\epsilon \Delta^Y_{\epsilon}=\Psi_{v\mapsto 1}(P_\epsilon)$. It follows that the derivative $\left. \frac{d}{d\epsilon} P_\epsilon \Q^C_{\epsilon}(0,\cdot) \right|_{\epsilon=0}$ is the same as the derivative of $f : \epsilon \mapsto \frac{\kappa - P_\epsilon \Q^C_{\epsilon}(0,\cdot)}{P_\epsilon \Delta^C_{\epsilon}} \Psi_{v\mapsto 1}(P_\epsilon)$ at $\epsilon=0$, provided this derivative exists. Noting that $v\mapsto 1$ is a particular instance of a fixed treatment rule, we may take $\edecision^\FR$ to be $v \mapsto 1$ in the results on pathwise differentiability of $P \mapsto \Psi_{\edecision^\FR}(P)$ and show that
\begin{align}
    \left.\frac{d}{d\epsilon} \Psi_{v\mapsto 1}(P_\epsilon) \right|_{\epsilon=0} = \int D(P_0,v\mapsto 1,0,\Q^C_0)(o) H(o) P_0(do). \label{eq:fixedref1grad}
\end{align}
As both the above derivative and the derivatives in \eqref{eq:resourcegrad1} and \eqref{eq:resourcegrad2} exist, by the chain rule, it follows that
\begin{align*}
    \left.\frac{d}{d\epsilon} f(\epsilon) \right|_{\epsilon=0}&= \frac{\kappa - \alpha P_0 \Q^C_0(0,\cdot)}{P_0 \Delta^C_0}\left.\frac{d}{d\epsilon} \Psi_{v\mapsto 1}(P_\epsilon) \right|_{\epsilon=0} - \frac{(\kappa - \alpha P_0 \Q^C_0(0,\cdot))\Psi_{v\mapsto 1}(P_0)}{(P_0 \Delta^C_0)^2} \left. \frac{d}{d\epsilon} P_\epsilon \Delta^C_{\epsilon} \right|_{\epsilon=0} \\
    &\hspace{.5in}- \alpha \frac{\Psi_{v\mapsto 1}(P_0)}{P_0 \Delta^C_0}\left.\frac{d}{d\epsilon} P_\epsilon \Q^C_{\epsilon}(0,\cdot) \right|_{\epsilon=0}.
\end{align*}
Note that $\phi_P=P \Q^C_P(0,\cdot)$. Plugging \eqref{eq:resourcegrad1}, \eqref{eq:resourcegrad2} and \eqref{eq:fixedref1grad} into the above and we can show that the right-hand side of the above is equal to the right-hand side of \eqref{eq:RDlim}. As $\left.\frac{d}{d\epsilon} f(\epsilon) \right|_{\epsilon=0}=\left.\frac{d}{d\epsilon}\Psi_{\edecision^\RD_{\epsilon}}(P_\epsilon)\right|_{\epsilon=0}$, we have shown that \eqref{eq:RDlim} holds, and the desired result follows.

\subsubsection{Canonical gradient of \texorpdfstring{$P \mapsto \Psi_{\edecision_P}(P)$}{optimal ITR mean outcome}}

Let $H$ be a score function in $\mathcal{H}$. The argument that we use parallels that of \citeauthor{Luedtke2016} \citep{Luedtke2016} and \citeauthor{Qiu2021} \citep{Qiu2021}, except that it is slightly modified to account for the fact that the resource constraint takes a different form in this paper.

We first note that all of following hold for all $\epsilon$ sufficiently close to zero:
\begin{align}
    \sup_w |\Delta^C_{\epsilon}(w) - \Delta^C_{0}(w)|&\lesssim |\epsilon|, \label{eq:DeltaAbd} \\
    \sup_v |\delta^Y_{\epsilon}(v) - \delta^Y_{0}(v)|&\lesssim |\epsilon|, \label{eq:DeltaYbd} \\
    \sup_v |\delta^C_{\epsilon}(v) - \delta^C_{0}(v)|&\lesssim |\epsilon|. \label{eq:DeltaAbbd}
\end{align}
The derivations of these inequalities are straightforward and hence omitted. Under Condition~\ref{IV assumption: strong encouragement}, the above inequalities imply that
\begin{align}
    \sup_v |\xi_{\epsilon}(v) - \xi_{0}(v)| = \left| \frac{\delta^Y_{\epsilon}(v)}{\delta^C_{\epsilon}(v)} - \frac{\delta^Y_{0}(v)}{\delta^C_{0}(v)} \right| \lesssim |\epsilon|. \label{eq:xibd}
\end{align}

For $\epsilon$ sufficiently close to zero, it will be useful to define
\begin{align*}
    \Gamma_\epsilon : \eta\mapsto \expect_\epsilon [I\{\xi_{\epsilon}(V) > \eta\} \delta^C_{\epsilon}(V)]
\end{align*}
for $\eta \in [-\infty,\infty)$. We also define $\Gamma_\epsilon' : \eta\mapsto \frac{d}{ds} \Gamma_\epsilon(s)|_{s=\eta}$ when the derivative exists.

We first show two lemmas. These two lemmas show that, under a perturbed distribution $P_\epsilon$ with magnitude $\epsilon$, the fluctuation in the threshold $\tau_\epsilon-\epsilon_0$ is of order $\epsilon$. This result is crucial in quantifying the convergence rate of two terms in the expansion of $\Psi_{\rho_\epsilon}(P_\epsilon) - \Psi_{\rho_0}(P_0)$, namely terms~1 and 3 in \eqref{eq:IEROptMainPDdisplay} below. In particular, term~1 is the main challenge in the analysis as it comes from the perturbation in the threshold and is unique in estimation problems involving the evaluation of optimal ITRs. The first studies the convergence of $\eta_\epsilon$ to $\eta_0$. Because it may be the case that $\eta_0=-\infty$, the convergence stated in this result is convergence in the extended real line.
\begin{lemma}\label{lem:etacontinuity}
    Under the conditions of Theorem~\ref{e RC theorem: differentiability}, $\eta_{\epsilon} \rightarrow \eta_{0}$ as $\epsilon \rightarrow 0$.
\end{lemma}
\begin{proof}[Proof of Lemma~\ref{lem:etacontinuity}]
    We separately consider the cases where $\eta_{0} > -\infty$ and $\eta_{0} = -\infty$.
    
    Suppose that $\eta_{0} > -\infty$. For all sufficiently small $\delta>0$ and sufficiently small $|\epsilon|$, by \eqref{eq:DeltaAbbd}, \eqref{eq:xibd} and the fact that the range of $H$ is contained in $[-1,1]$, we can show that
    \begin{align*}
        \Gamma_\epsilon(\eta_0 + \delta) + \alpha \phi_\epsilon &\leq (1+\const |\epsilon|) \expect_0[ I\{\xi_{\epsilon}(V) > \eta_0+\delta\} \delta^C_{0}(V)] + \alpha \phi_\epsilon \leq (1+\const |\epsilon|) \Gamma_0\left(\eta_0 + \delta - \const |\epsilon|\right) + \alpha \phi_\epsilon.
    \end{align*}
    Under Condition~\ref{e RC assumption: continuous weight}, as long as $\delta$ is small enough, the right-hand side converges to $\Gamma_0(\eta_0 + \delta) + \alpha \phi_0$ as $\epsilon\rightarrow 0$. Moreover, Conditions~\ref{e RC assumption: continuous weight}~and~\ref{IV assumption: strong encouragement} can be combined to show that the derivative of $\Gamma_0$ is strictly negative for all $x\in[\eta_0,\eta_0+\delta]$ for sufficiently small $\delta$, and so $\Gamma_0(\eta_0)>\Gamma_0(\eta_0 + \delta)$. Because $\Gamma_0(\eta_0)+\alpha\phi_0=\kappa$ by the definition of $\edecision_0$ under Condition~\ref{e RC assumption: continuous density}, it follows that, for all $\epsilon$ sufficiently close to zero, $\Gamma_\epsilon(\eta_0 + \delta)+\alpha\phi_\epsilon<\kappa$. By the definition $\eta_{\epsilon}:=\inf \{ \eta: \Gamma_\epsilon(\eta) \leq \kappa-\alpha\phi_\epsilon \}$,  it follows that, for all $\epsilon$ sufficiently close to zero, $\eta_0+\delta \ge \eta_{\epsilon}$, that is, $\eta_{\epsilon}-\eta_0\le \delta$.
    
    By similar arguments, we can show that, for all $\epsilon$ sufficiently close to zero, $\eta_{\epsilon}-\eta_0\ge -\delta$. Indeed,
    \begin{align*}
        \Gamma_\epsilon(\eta_0 - \delta) + \alpha\phi_\epsilon \ge (1-\const |\epsilon|)\Gamma_0(\eta_0 - \delta + \const |\epsilon|) + \alpha\phi_\epsilon.
    \end{align*}
    The right-hand side converges to $\Gamma_0(\eta_0-\delta) + \alpha\phi_0$ as $\epsilon\rightarrow 0$ provided $\delta$ is sufficiently small. The derivative of $\Gamma_0$ is strictly negative on $[\eta_0-\delta,\eta_0]$ provided $\delta$ is small enough, and therefore, $\Gamma_0(\eta_0-\delta) + \alpha\phi_0 > \Gamma_0(\eta_0) + \alpha\phi_0 = \kappa$. Hence, $\Gamma_\epsilon(\eta_0 - \delta) + \alpha\phi_\epsilon > \kappa$. By the definition of $\eta_{\epsilon}$, it follows that $\eta_{\epsilon} - \eta_0\ge -\delta$.
    
    Combining these two results, we see that, for all $\epsilon$ sufficiently close to zero, $|\eta_{\epsilon}-\eta_0|\le \delta$. Hence, $\limsup_{\epsilon\rightarrow 0} |\eta_{\epsilon}-\eta_0|\le \delta$. As $\delta>0$ is an arbitrary in a neighborhood of zero, it follows that $\limsup_{\epsilon\rightarrow 0} |\eta_{\epsilon}-\eta_0|=0$. That is, $\eta_{\epsilon}\rightarrow \eta_0$ as $\epsilon\rightarrow 0$ in the case that $\eta_{0} > -\infty$.
    
    We now study the case where $\eta_{0} = -\infty$. If $\kappa=\infty$, then it is trivial that $\eta_\epsilon=-\infty=\eta_0$ for all $\epsilon$, and so the desired result holds. Suppose now that $\kappa<\infty$. Fix a small enough $\delta>0$ so that the bound in \eqref{eq:xibd} is valid for all $\epsilon\in[-\delta,\delta]$. Also fix $\epsilon\in[-\delta,\delta]$ and $\eta\in\real$. By \eqref{eq:xibd} and the bound on the range of $H$,
    \begin{align*}
        \Gamma_\epsilon(\eta) + \alpha \phi_\epsilon \leq (1+\const |\epsilon|) \expect_0[I\{\xi_\epsilon(V)>\eta\} \Delta^C_{0,b}(V)] + \alpha \phi_\epsilon \leq (1+\const |\epsilon|) \Gamma_0(\eta-\const |\epsilon|) + \alpha \phi_\epsilon.
    \end{align*}
    Because $\Gamma_0$ is a nonnegative decreasing function, the right-hand side is no greater than $(1+\const |\epsilon|) \Gamma_0(\eta - \const \delta) + \alpha \phi_\epsilon$. This upper bound tends to $\Gamma_0(\eta - \const \delta) + \alpha \phi_0$ as $\epsilon\rightarrow 0$. Hence, $\limsup_{\epsilon\rightarrow 0} \Gamma_\epsilon(\eta) + \alpha \phi_\epsilon \le \Gamma_0(\eta - \const \delta) + \alpha \phi_0$. By Condition~\ref{e RC assumption: continuous weight} and the monotonicity of $\Gamma_0$, $\Gamma_0(\eta - \const \delta) + \alpha \phi_0 < \kappa$, and so $\Gamma_\epsilon(\eta) + \alpha \phi_\epsilon < \kappa$ for all $\epsilon$ sufficiently close to zero. By the definition of $\eta_\epsilon$, it follows that $\eta_\epsilon \leq \eta$ for all $\epsilon$ sufficiently close to zero. Since $\eta \in \real$ is arbitrary, the desired result follows.
\end{proof}

The next lemma establishes a rate of convergence of $\tau_\epsilon$ to $\tau_0$ as $\epsilon\rightarrow 0$.
\begin{lemma}\label{lem:taurate}
    Under conditions of Theorem~\ref{e RC theorem: differentiability}, $\tau_\epsilon=\tau_0 + \bigO(\epsilon)$.
\end{lemma}
\begin{proof}[Proof of Lemma~\ref{lem:taurate}]
    We separately consider the cases where $\eta_{0} <0$ and $\eta_{0}\ge 0$.
    
    We start with the easier case where $\eta_{0} <0$. In this case, Lemma~\ref{lem:etacontinuity} shows that $\tau_{\epsilon}:=\max\{\eta_{\epsilon},0\}$ is equal to $\tau_{0} = 0$ for all $\epsilon$ sufficiently close to zero. Thus, $\tau_{\epsilon} - \tau_{0}=\bigO(|\epsilon|)$.
    
    Now consider the more difficult case where $\eta_{0} \ge 0$. By the Lipschitz property of the function $x\mapsto \max\{x,0\}$, we can show that $|\max\{\eta_\epsilon,0\}-\max\{\eta_0,0\}|\le |\eta_\epsilon - \eta_0|$. As a consequence, to show that $\tau_{\epsilon} - \tau_{0}=\bigO(\epsilon)$, it suffices to show that $\eta_\epsilon - \eta_0 = \bigO(\epsilon)$. We next establish this statement.
    
    Fix $\epsilon$ in a sufficiently small neighborhood of zero. By the definition $\eta_{\epsilon}:=\inf \{ \eta: \Gamma_\epsilon(\eta)  \leq \kappa - \phi_\epsilon \}$, the bound on the range of $H$, and \eqref{eq:xibd}, it holds that $\kappa<\Gamma_\epsilon(\eta_\epsilon - |\epsilon|) + \alpha \phi_\epsilon \le [1+\const |\epsilon|]\Gamma_0(\eta_\epsilon - [1+\const]|\epsilon|) + \alpha \phi_\epsilon$.  We use a Taylor expansion of $\Gamma_0$ about $\eta_0$, which is justified by Condition~\ref{e RC assumption: continuous weight} provided $|\epsilon|$ is small enough, and it follows that
    \begin{align*}
        \kappa< [1 + \const |\epsilon|]\left[\Gamma_0(\eta_0) + \{\eta_\epsilon-\eta_0 - (1-\const)|\epsilon|\}\{\Gamma_0'(\eta_0) + \smallo(1)\}\right] + \alpha \phi_0 + \bigO(\epsilon).
    \end{align*}
    By Condition~\ref{e RC assumption: continuous weight}, $\Gamma_0(\eta_0) + \alpha \phi_0=\kappa$. Plugging this into the above shows that
    \begin{align*}
        0&< \const \Gamma_0(\eta_0) |\epsilon| +  [1 + \const |\epsilon|]\left[\eta_\epsilon-\eta_0 - (1-\const)|\epsilon|\right]\left[\Gamma_0'(\eta_0) + \smallo(1)\right]  + \bigO(\epsilon).
    \end{align*}
    Note that Condition~\ref{e RC assumption: continuous weight} implies that $\Gamma_0'(\eta_0)\in(-\infty,0)$. Therefore, the above shows that, for all $\epsilon$ sufficiently close to zero, $0< [\eta_\epsilon-\eta_0]\Gamma_0'(\eta_0) + \const |\epsilon| + \smallo\left(\eta_\epsilon-\eta_0\right)$, which implies that there exists an $\bigO(\epsilon)$ sequence for which $\eta_\epsilon-\eta_0<\bigO(\epsilon)$.
    
    A similar argument, which is based on the observation that $\Gamma_\epsilon(\eta_\epsilon+|\epsilon|) + \phi_\epsilon \le \kappa$, can be used to show that there exists an $\bigO(\epsilon)$ sequence such that $\eta_\epsilon-\eta_0>\bigO(\epsilon)$. Combining these two bounds shows that $\eta_\epsilon-\eta_0=\bigO(\epsilon)$, as desired. This concludes the proof.
\end{proof}

Our derivation of the canonical gradient is based on the following decomposition:
\begin{align}
    \Psi_{\edecision_{\epsilon}}&(P_\epsilon) - \Psi_{\edecision_{0}}(P_0) \nonumber \\
    =& \Psi_{\edecision_{\epsilon}}(P_\epsilon) - \Psi_{\edecision_{0}}(P_\epsilon) + \Psi_{\edecision_{0}}(P_\epsilon) - \Psi_{\edecision_{0}}(P_0) \nonumber \\
    =& P_\epsilon \{ [\edecision_{\epsilon} - \edecision_{0}] \delta^Y_{\epsilon} \} + \Psi_{\edecision_{0}}(P_\epsilon) - \Psi_{\edecision_{0}}(P_0) \nonumber \\
    =& P_\epsilon \{ [\edecision_{\epsilon} - \edecision_{0}] (\delta^Y_{\epsilon} - \tau_{0} \delta^C_{\epsilon}) \} + \tau_{0} P_\epsilon \{ (\edecision_{\epsilon} - \edecision_{0}) \delta^C_{\epsilon} \} + \Psi_{\edecision_{0}}(P_\epsilon) - \Psi_{\edecision_{0}}(P_0) \nonumber \\
    =& P_\epsilon \{ [\edecision_{\epsilon} - \edecision_{0}] (\delta^Y_{\epsilon} - \tau_{0} \delta^C_{0}) \} + [\Psi_{\edecision_{0}}(P_\epsilon) - \Psi_{\edecision_{0}}(P_0)] + \tau_{0} \{ P_\epsilon [\delta^C_{\epsilon} \edecision_{\epsilon}] + \alpha \phi_\epsilon - P_0 [\delta^C_{0} \edecision_{0}] - \alpha \phi_0 \} \nonumber \\
    &- \tau_{0} P_\epsilon \{ (\delta^C_{\epsilon} - \Delta^C_{0,b}) \edecision_0 \} - \tau_{0} (P_\epsilon-P_0) \{ \delta^C_{0} \edecision_{0} \} - \alpha \tau_0 \{ \phi_\epsilon - \phi_0 \}. \label{eq:IEROptMainPDdisplay}
\end{align}
We separately study each of the six terms on the right-hand side, which we refer to as term~1 up to term~6.

\vspace{0.5em}\noindent\textbf{Study of term 1 in \eqref{eq:IEROptMainPDdisplay}:} We will show that this term is $\smallo(\epsilon)$. By Lemma~\ref{lem:taurate} and \eqref{eq:DeltaYbd},
$$\sup_v|\delta^Y_{\epsilon}(v) - \tau_{\epsilon} \delta^C_{\epsilon} - \delta^Y_{0}(v) + \tau_{0} \delta^C_{0}| \leq \sup_v|\delta^Y_{\epsilon}(v) - \delta^Y_{0}(v)| + \sup_v|\delta^C_{\epsilon}(v) - \delta^C_{0}(v)| + |\tau_{\epsilon} - \tau_{0}| \lesssim |\epsilon|.$$
Under Condition~\ref{e RC assumption: continuous density}, $P_0\{\xi_{0}(V)=\tau_{0}\}=0$. We apply a similar argument as that used to prove Lemma~2 in \citeauthor{VanderLaan2014} \cite{VanderLaan2014}:
\begin{align*}
    \big|P_\epsilon &\{ [\edecision_{\epsilon} - \edecision_{0}] (\delta^Y_{\epsilon} - \tau_{0} \delta^C_{\epsilon}) \}  \big| \\
    &= \left| \int [\edecision_{\epsilon}(v) - \edecision_{0}(v)] [\delta^Y_{\epsilon}(v) - \tau_{0} \delta^C_{\epsilon}(v)] P_{W,\epsilon}(dw) \right| \\
    &\leq \int \left|  \edecision_{\epsilon}(v) - \edecision_{0}(v) \right| \left| \delta^Y_{\epsilon}(v) - \tau_{0} \delta^C_{\epsilon}(v) \right| P_{W,\epsilon}(dw). \\
    \intertext{Because $\edecision_{\epsilon}(v) \neq \edecision_{0}(v)$ implies that either (i) $\xi_{\epsilon}(v) - \tau_{\epsilon}$ and $\xi_{0}(v) - \tau_{0}$ have different signs or (ii) only one of these quantities is zero, the display continues as}
    &\leq \int I\{|\xi_{0}(v) - \tau_{0}| \leq |\xi_{\epsilon}(v) - \tau_{\epsilon} - \xi_{0}(v) + \tau_{0}|\} \left| \delta^Y_{\epsilon}(v) - \tau_{0} \delta^C_{\epsilon}(v) \right| P_{W,\epsilon}(dw) \\
    &\leq \int I\{|\xi_{0}(v) - \tau_{0}| \leq \const |\epsilon|\} \left( \left| \delta^Y_{0}(v) - \tau_{0} \delta^C_{0}(v) \right| + \const |\epsilon| \right) P_{W,\epsilon}(dw). \\
    \intertext{Using the facts that $\inf_v \delta^C_{0}(v)>0$ by Condition~\ref{IV assumption: strong encouragement}, that $\sup_v \delta^C_{0}(v)\le 1$ since probabilities are no more than one, and that $\xi_0(v):=\delta^Y_{0}(v)/\delta^C_{0}(v)$, the display continues as}
    &\leq \int I\{|\xi_{0}(v) - \tau_{0}| \leq \const |\epsilon|\} \left( \left| \xi_0(v) - \tau_{0} \right| + \const |\epsilon| \right) P_{W,\epsilon}(dw) \\
    \intertext{Leveraging the bound on $|\xi_{0}(v) - \tau_{0}|$ that appears in the indicator function, we see that}
    &\leq \int I\{|\xi_{0}(v) - \tau_{0}| \leq \const |\epsilon|\} (\const |\epsilon|  + \const |\epsilon|)  P_{W,\epsilon}(dw) \\
    &\lesssim |\epsilon| \int I\{|\xi_{0}(v) - \tau_{0}| \leq \const |\epsilon|\} P_{W,0}(dw) \\
    &= |\epsilon| \int I\{0<|\xi_{0}(v) - \tau_{0}| \leq \const |\epsilon|\} P_{W,0}(dw),
\end{align*}
where the final equality holds by Condition~\ref{e RC assumption: non-exceptional law}. The integral in the final expression is $\smallo(1)$, and so this expression is $\smallo(\epsilon)$.

\vspace{0.5em}\noindent\textbf{Study of term 2 in \eqref{eq:IEROptMainPDdisplay}:} By the result on the pathwise differentiability of $P \mapsto \Psi_{\edecision^\FR}(P)$, setting $\edecision^\FR$ to be $\edecision_0$, we see that the second term satisfies $\Psi_{\edecision_{0}}(P_\epsilon) - \Psi_{\edecision_{0}}(P_0) = \epsilon\int G_2(o) H(o) P(do) + \smallo(\epsilon)$, 
where $G_2\in L^2_0(P_0)$ is equal to $D(P_0,\edecision_0,0,\Q^C_0)$.

\vspace{0.5em}\noindent\textbf{Study of term 3 in \eqref{eq:IEROptMainPDdisplay}:} We will show that the third term is identical to zero for any $\epsilon$ that is sufficiently close to zero. If $\tau_{0}=0$, then this term is trivially zero. Otherwise, $\tau_{0}=\eta_{0} > 0$. Lemma~\ref{lem:etacontinuity} shows that, in this case, $\eta_{\epsilon} > 0$ for $\epsilon$ sufficiently close to zero. Hence, $\expect_{\epsilon} [ \delta^C_{\epsilon}(V) \edecision_{\epsilon}(V) ] + \alpha \phi_\epsilon = \kappa = \expect_0 [ \delta^C_{0}(V) \edecision_{0}(V) ] + \alpha \phi_0$. Consequently, term~3 equals zero for all $\epsilon$ sufficiently close to zero.

\vspace{0.5em}\noindent\textbf{Study of term 4 in \eqref{eq:IEROptMainPDdisplay}:} We will show that this term can be writes as $\epsilon\int G_4(o) H(o) P_0(do) + \smallo(\epsilon)$ for an appropriately defined $G_4\in L^2_0(P_0)$ that does not depend on $H$. Note that there exists a function $H_W : (w \mid v)\mapsto H_W(w \mid v)$ for which $\int H_W(w \mid v) P_{W,0}(dw \mid v) = 0$, $\sup_{w,v}|H_W(w \mid v)| < \infty$, and, for all $v$,
\begin{align*}
    & P_{W,\epsilon}(dw \mid v)=(1+\epsilon H_{W}(w \mid v)+\smallo(\epsilon)) P_{W,0}(dw \mid v).
\end{align*}
The function $H_W$ can be chosen so that the above $\smallo(\epsilon)$ term indicates little-oh behavior uniformly over $w$ and $v$.
By the definition of $H_C$ from \eqref{eq:epsMarginalsConditionals}, we see that
\begin{align*}
    \delta^C_{\epsilon}(v) - \delta^C_{0}(v) 
    &= \iint c \Big\{[1+\epsilon H_C(c \mid 1,w)]  [1+\epsilon H_W(w \mid v)+\smallo(\epsilon)]-1\Big\}P_0(dc \mid 1,w) P_0(dw \mid v) \\
    &\quad- \iint c \Big\{[1+\epsilon H_C(c \mid 0,w)]  [1+\epsilon H_W(w \mid v)+\smallo(\epsilon)]-1\Big\}P_0(dc \mid 0,w) P_0(dw \mid v) \\
    &= \epsilon \Big\{ \iint c (H_C(c \mid 1,w) + H_W(w \mid v)+\smallo(1)) P_0(dc \mid 1,w) P_0(dw \mid v) \\
    &\quad- \iint c (H_C(c \mid 0,w) + H_W(w \mid v)+\smallo(1)) P_0(dc \mid 0,w) P_0(dw \mid v) \Big\} + \smallo(\epsilon),
\end{align*}
where the little-oh terms are uniform over $w$ and $v$. Hence,
\begin{align*}
    &\left.\frac{d}{d\epsilon} P_\epsilon \{ [\delta^C_{\epsilon} - \delta^C_{0}] \edecision_{0} \}\right|_{\epsilon=0} \\
    &= \iint \edecision_{0}(v) c \{H_C(c \mid 1,w) + H_W(w \mid v)\} P_0(dc \mid 1,w) P_0(dw) \\
    &\quad- \iint \edecision_{0}(v) c \{H_C(c \mid 0,w) + H_W(w \mid v)\} P_0(dc \mid 0,w) P_0(dw) \\
    &= \expect_0 \left[ \edecision_{0}(V)\left( \frac{1}{T+\mu^T_0(W)-1} \{C - \Q^C_{0}(T,W)\} H_C(C \mid 1,W) + \{\Delta^C_{0}(W) - \delta^C_{0}(V)\} H_W(W \mid V) \right)\right]. \\
    \intertext{Since $\expect_0[H_C(C \mid T,W) \mid T,W]=\expect_0[H_W(W \mid V) \mid V]=0$ $P_0$-a.s., the display continues as}
    &= \expect_0 \left[ \edecision_{0}(V) \left( \frac{1}{T+\mu^T_0(W)-1} \{C - \Q^C_{0}(T,W)\} + \Delta^C_{0}(W) - \delta^C_{0}(V)\right) H(O) \right].
\end{align*}
As a consequence, term~4 satisfies
\begin{align*}
    -\tau_0 P_\epsilon \{ [\delta^C_{\epsilon} - \delta^C_{0}] \edecision_{0} \}&= \epsilon \int G_4(o) H(o) P_0(do) + \smallo(\epsilon),
\end{align*}
where
\begin{align*}
    G_4 : o \mapsto&-\tau_{0} \edecision_{0}(v) \left\{ \frac{1}{t+\mu^T_0(w)-1} [ c - \Q^C_{0}(t,w) ] + \Delta^C_{0}(w) - \delta^C_{0}(v) \right\}.
\end{align*}

\vspace{0.5em}\noindent\textbf{Study of term 5 in \eqref{eq:IEROptMainPDdisplay}:} By \eqref{eq:epsMarginalsConditionals} and the fact that $P_0 \{ \delta^C_{0} \edecision_{0} \}=\kappa-\alpha \phi_0$ whenever $\tau_{0}>0$, we see that $-\tau_0 (P_\epsilon - P_0)\{\delta^C_{0} \edecision_0\}= \epsilon \int G_5(o) H_V(v) P_0(do)$, where $G_5\in L^2_0(P_0)$ is defined as $o\mapsto -\tau_{0} [\delta^C_{0}(v) \edecision_{0}(v) - \kappa + \alpha \phi_0]$. Since $H_V$ is defined as $v\mapsto \expect_0[H(O) \mid V=v]$, we see that it also holds that $-\tau_0 (P_\epsilon - P_0)\{\delta^C_{0} \edecision_0\}= \epsilon \int G_5(o) H(o) P_0(do)$.

\vspace{0.5em}\noindent\textbf{Study of term 6 in \eqref{eq:IEROptMainPDdisplay}:} We have shown that
$$\phi_\epsilon - \phi_0 = \epsilon \int \left\{ \frac{1-t}{1-\Q^T_0(w)} [c-\Q^C_0(0,w)] + \Q^C_0(0,w) - \phi_0 \right\} H(o) P_0(do) + \smallo(\epsilon).$$
Therefore, $-\tau_0 \alpha (\phi_\epsilon - \phi_0) = \epsilon \int G_6(o) H(o) P_0(do) + \smallo(\epsilon)$ where
$$G_6: o \mapsto -\tau_0 \alpha \left\{ \frac{1-t}{1-\Q^T_0(w)} [c-\Q^C_0(0,w)] + \Q^C_0(0,w) - \phi_0 \right\}.$$

\vspace{0.5em}\noindent\textbf{Conclusion of the derivation of the canonical gradient of $P\mapsto \Psi_{\edecision_P}(P)$:} Combining our results regarding the six terms in \eqref{eq:IEROptMainPDdisplay}, we see that
\begin{align*}
    \Psi_{\edecision_{\epsilon}}(P_\epsilon) - \Psi_{\edecision_{0}}(P_0)&= \epsilon \int \left[G_2(o) + G_4(o) + G_5(o) + G_6(o) \right] H(o) P_0(do) + \smallo(\epsilon).
\end{align*}
Dividing both sides by $\epsilon\not=0$ and taking the limit as $\epsilon\rightarrow 0$, we see that $G_2 + G_4 + G_5 + G_6=G(P_0)$ is the canonical gradient of $P\mapsto \Psi_{\edecision_P}(P)$ at $P_0$.

\subsection{Expansions based on gradients or pseudo-gradients} \label{section: pseudo-gradient expansion}

In this section, we present (approximate) first-order expansions of ATE parameters based on which we construct our proposed targeted minimum-loss based estimators (TMLE) and prove their asymptotic linearity. We refer the readers to Supplement~S5 of \citeauthor{Qiu2021} \citep{Qiu2021} for an overview of TMLE based on gradients and pseudo-gradients. The overall idea behind TMLE based on gradients is the following: the empirical mean of the gradient at the estimated distribution can be viewed as the first-order bias of the plug-in estimator; this bias can be removed by solving the estimating equation that equates the first-order bias to zero. The idea behind pseudo-gradients is similar, except that the gradient is replaced by an approximation that we term \textit{pseudo-gradient} so that the corresponding estimating equation is easy to solve with a single regression step.

For any ITR $\edecision: \mathcal{W} \rightarrow [0,1]$ that utilizes all covariates, we define
\begin{align*}
    R_{\edecision}(P,P_0) :=& \Psi_{\edecision}(P) - \Psi_{\edecision}(P_0) + P_0 D(P,\edecision,0,\Q^C) \\
    =& \expect_0 \Bigg[ \edecision(W) \Bigg\{ \frac{\Q^T_P(W)-\Q^T_0(W)}{\Q^T_P(W)} (\Q^Y_P(1,W) - \Q^Y_0(1,W)) \\
    &\quad\quad\quad\quad\quad+ \frac{\Q^T_P(W)-\Q^T_0(W)}{1-\Q^T_P(W)} (\Q^Y_P(0,W) - \Q^Y_0(0,W)) \Bigg\} \Bigg].
\end{align*}
For any ITR $\edecision: \mathcal{V} \rightarrow [0,1]$ that only utilizes $V$, for convenience, we define $R_{\edecision}(P,P_0) := R_{w \mapsto \edecision(V(w))}(P,P_0)$.

For $\Psi_{\edecision^\FR}$ and $\Psi_{\edecision^\TP_P}$, it is straightforward to show that the following expansions hold:
\begin{align*}
    \Psi_{\edecision^\FR}(P) - \Psi_{\edecision^\FR}(P_0) &= -P_0 G_\FR(P) + R_{\edecision^\FR}(P,P_0), \\
    \Psi_{\edecision^\TP_P}(P) - \Psi_{\edecision^\TP_{P_0}}(P_0) &= -P_0 G_\TP(P) + P_0 \left\{ \frac{\Q^T_P(\cdot)-\Q^T_0(\cdot)}{1-\Q^T_P(\cdot)} (\Q^Y_P(0,\cdot)-\Q^Y_0(0,\cdot)) \right\}.
\end{align*}

For $P \mapsto \Psi_{\edecision^\RD_P}(P)$, we expand this parameter sequentially as follows:
\begin{align*}
    &\Psi_{\edecision^\RD_P}(P) - \Psi_{\edecision^\RD_{P_0}}(P_0) = P_0 D(P,\edecision^\RD_P,0,\Q^C_P) + R_{\edecision^\RD_P}(P,P_0) + (\edecision^\RD_P - \edecision^\RD_0) P_0 \Delta^Y_0, \\
    & \edecision^\RD_P - \edecision^\RD_0 = \frac{\kappa - \phi_P}{P \Delta^C_P} - \frac{\kappa - \phi_0}{P_0 \Delta^C_0}, \\
    &(\kappa-\alpha \phi_P) - (\kappa - \alpha \phi_0) = \alpha \left\{ P_0 D_1(P,\Q^C) + P_0 \left\{ (\Q^C(0,\cdot) - \Q^C(1,\cdot)) \frac{\Q^T_P-\Q^T_0}{1-\Q^T_P} \right\} \right\}, \\
    &P \Delta^C_P - P_0 \Delta^C_0 = - P_0 D_2(P,\Q^C_P) + P_0 \left\{ (\Q^C_P(1,\cdot) - \Q^C_0(1,\cdot)) \frac{\Q^T_P-\Q^T_0}{\Q^T_P} + (\Q^C_P(0,\cdot) - \Q^C_0(0,\cdot)) \frac{\Q^T_P-\Q^T_0}{1-\Q^T_P} \right\}.
\end{align*}

For $P \mapsto \Psi_{\edecision_P}(P)$, straightforward but tedious calculation shows that the following expansion holds:
\begin{align*}
    \Psi_{\edecision_P}(P) - \Psi_{\edecision_{0}}(P_0) &= -P_0 D(P,\edecision,\tau_0,\Q^C) \\
    &\quad+ R_{\edecision}(P,P_0) + P_0 \{ (\edecision-\edecision_{0}) (\delta^Y_{0} - \tau_{0} \delta^C_{0}) \} \\
    &\quad- \tau_{0} \expect_0 \left[ \edecision(V) \frac{\Q^T_P(W) - \Q^T_0(W)}{\Q^T_P(W)} \{\Q^C(1,W) - \Q^C_0(1,W)\} \right] \\
    &\quad+ \tau_{0} \expect_0 \left[ (1-\edecision(V)) \frac{\Q^T_P(W) - \Q^T_0(W)}{1-\Q^T_P(W)} \{\Q^C(0,W) - \Q^C_0(0,W)\} \right].
\end{align*}

\subsection{Asymptotic linearity of proposed estimator (Theorem~\ref{e RC theorem: asymptotic linearity})} \label{proof: asymptotic linearity}

For convenience, we set $\hat{P}_n$ to have component $\Q^T_n$ and $\hat{\Q}^C_n$, even though the plug-in estimator does not explicitly involve these functions. We start with some lemmas that facilitate the proof of the main theorem. In this section, we define $\eta_n:=\eta_n(k_n)$ and $\tau_n:=\tau_n(k_n)$ to simplify notations.

Our proof is centered around the expansions in Supplement~\ref{section: pseudo-gradient expansion}. We first prove a few lemmas. Lemma~\ref{e RC lemma: AL of resource0 and additional resource} is a standard asymptotic linearity result on estimators $\phi_n$ and $P_n \hat{\Delta}^C_n$ about treatment resource being used for constant ITRs $v \mapsto 0$ and $v \mapsto 1$, respectively; Lemma~\ref{lemma: equivalence of L2 distance} is a technical convenient tool to convert conditions on norms in Condition~\ref{e RC assumption: remainder} between functions; Lemmas~\ref{lemma: bound for consistency of survival function}--\ref{e RC lemma: updated tau} are analysis results for estimators that are similar to Lemmas~\ref{lem:etacontinuity}--\ref{lem:taurate} for deterministic perturbations of $P_0$, and they lead to the crucial Lemma~\ref{e RC lemma: remainder} on the negligibility of the remainder $R_\edecision(\hat{P}_n,P_0)$ for an arbitrary ITR $\edecision$.

\begin{lemma}[Asymptotic linearity of $\phi_n$ and $P_n \hat{\Delta}^C_n$] \label{e RC lemma: AL of resource0 and additional resource}
    Under the conditions of Theorem~\ref{e RC theorem: asymptotic linearity},
    \begin{align*}
        \phi_n - \phi_0 &= (P_n-P_0) D_1(P_0,\Q^C_0) + \smallo_p(n^{-1/2}) = \bigO_p(n^{-1/2}), \\
        P_n \hat{\Delta}^C_n - P_0 \Delta^C_0 &= (P_n-P_0) D_2(P_0,\Q^C_0) + \smallo_p(n^{-1/2}) = \bigO_p(n^{-1/2}).
    \end{align*}
\end{lemma}
This result follows from the facts that (i) $\phi_n$ is a one-step correction estimator of $\phi_0$ \citep{Pfanzagl1982}, and (ii) $P_n \hat{\Delta}^C_n$ is a TMLE for $P_0 \Delta^C_0$ \citep{VanderLaan2017,VanderLaan2018}. Therefore the proof is omitted.

\begin{lemma}[Lemma~S8 in \citeauthor{Qiu2021} \citep{Qiu2021}] \label{lemma: equivalence of L2 distance}
    Fix functions $\Q^C : \{0,1\} \times \mathcal{W}\rightarrow [0,1]$ and $\Q^Y : \{0,1\} \times \mathcal{W}\rightarrow \real$, and suppose that $P_0 \Q^Y(0,\cdot)^2<\infty$ and $P_0 \Q^Y(1,\cdot)^2<\infty$. If Condition~\ref{IV assumption: strong IV positivity} holds, then
    \begin{align*}
        & \| \Q^Y(1,\cdot) - \Q^Y_0(1,\cdot) \|_{2,P_0} + \| \Q^Y(0,\cdot) - \Q^Y_0(0,\cdot) \|_{2,P_0} \simeq \| \Q^Y - \Q^Y_0\|_{2,P_0}, \\
        & \| \Q^C(1,\cdot) - \Q^C_0(1,\cdot) \|_{2,P_0} + \| \Q^C(0,\cdot) - \Q^C_0(0,\cdot) \|_{2,P_0} \simeq \| \Q^C - \Q^C_0\|_{2,P_0},
    \end{align*}
    where $a \simeq b$ is defined as $a \lesssim b$ and $b \lesssim a$.
\end{lemma}

The following Lemmas~\ref{lemma: bound for consistency of survival function}--\ref{e RC lemma: updated tau} prove consistency of the estimated thresholds used to define the estimated optimal ITR $\edecision_n$.

\begin{lemma}[Lemma~S5 in \citeauthor{Qiu2021} \citep{Qiu2021}] \label{lemma: bound for consistency of survival function}
    Let $\epsilon>0$, $\eta \in \real$, $g: \mathscr{O} \rightarrow \real$ be bounded and functions $f_0 : \mathscr{O} \rightarrow \real$ and $f: \mathscr{O} \rightarrow \real$. Then
    \begin{align*}
        |P_0 ([I(f>\eta)-I(f_0>\eta)] g)|&\leq P_0 |[I(f>\eta)-I(f_0>\eta)] g| \\ &\lesssim P_0\{|f(O)-f_0(O)|>\epsilon\} + P_0\{|f_0(O)-\eta| \leq \epsilon\}.
    \end{align*}
    If $g$ takes values in $[-1,1]$, then $\lesssim$ can be replaced by $\leq$.
\end{lemma}

\begin{lemma}[Consistency of $\eta_n(\kappa)$] \label{e RC lemma: consistency of eta}
    Under Conditions~\ref{e RC assumption: continuous density},~\ref{e RC assumption: continuous weight}~and~\ref{e RC assumption: Glivenko-Cantelli}, $\eta_n(\kappa) \overset{p}{\rightarrow} \eta_0$.
\end{lemma}
This lemma is a stochastic variant of the deterministic result in Lemma~\ref{lem:etacontinuity} and has a similar proof. Therefore, the arguments are slightly abbreviated here.
\begin{proof}[Proof of Lemma~\ref{e RC lemma: consistency of eta}]
    We separately consider the cases where $\eta_0 > -\infty$ and $\eta_0=-\infty$.
    
    First consider the case where $\eta_0 > -\infty$. We start by showing that, for any $\eta$ sufficiently close to $\eta_0$, it holds that $\Gamma_n(\eta) - \Gamma_0(\eta)=\smallo_p(1)$. Fix an $\eta$ in a neighborhood of $\eta_0$. By the triangle inequality,
    \begin{align}
        |\Gamma_n(\eta) - \Gamma_0 (\eta)|&\leq |P_0 [I\{\xi_n(V(\cdot)) > \eta\} - I(\xi_0(V(\cdot)) > \eta)] \Delta^C_0| \nonumber \\
        &\quad+ |P_0 I\{\xi_n(V(\cdot)) > \eta\} [\Delta^C_n - \Delta^C_0]| \nonumber \\
        &\quad+ |(P_n-P_0) I\{\xi_n(V(\cdot)) > \eta\} \Delta^C_n|. \label{eq:gammancons}
    \end{align}
    We will show that the right-hand side is $\smallo_p(1)$. By Condition~\ref{e RC assumption: Glivenko-Cantelli}, the third term on the right is $\smallo_p(1)$ for $\eta$ sufficiently close to $\eta_0$. Moreover, because the second term is no greater than $\| \Delta^C_n - \Delta^C_0 \|_{1,P_0}$, Condition~\ref{e RC assumption: Glivenko-Cantelli} also implies that this second term  is also $\smallo_p(1)$. We will now argue that the first term is $\smallo_p(1)$. By Lemma~\ref{lemma: bound for consistency of survival function} and Condition~\ref{e RC assumption: bounded encouragement effect}, for any $\epsilon'>0$,
    \begin{align*}
        |P_0 [I(\xi_n(V(\cdot)) > \eta) - I(\xi_0(V(\cdot)) > \eta)] \Delta^C_0| 
        &\lesssim P_0 |I(\xi_n > \eta) - I(\xi_0 > \eta)| \\
        &\leq P_0 I(|\xi_n-\xi_0|>\epsilon') + P_0 I(|\xi_0-\eta| \leq \epsilon') \\
        &\leq \frac{\|\xi_n-\xi_0\|_{1,P_0}}{\epsilon'} + P_0 I(|\xi_0-\eta| \leq \epsilon'),
    \end{align*}
    where the final relation follows from Markov's inequality. We next show that the last line is $\smallo_p(1)$. Fix $\epsilon>0$. For $\eta$ that is sufficiently close to $\eta_0$ and $\epsilon'$ that is sufficiently small, by Condition~\ref{e RC assumption: continuous density}, we see that $S_0$ is continuous in $[\eta-\epsilon',\eta+\epsilon']$ and hence, for all sufficiently small $\epsilon'>0$, it holds that $P_0 I(|\xi_0-\eta| \leq \epsilon') \leq \epsilon/2$. Therefore,
    \begin{align*}
        P_0\left\{|P_0 [I(\xi_n>\eta)-I(\xi_0>\eta)]|>\epsilon\right\} &\leq  P_0\left\{\frac{\| \xi_n-\xi_0 \|_{1,P_0}}{\epsilon'} + P_0 I(|\xi_0-\eta| \leq \epsilon')>\epsilon\right\} \\
        &\leq  P_0\left\{\frac{\| \xi_n-\xi_0 \|_{1,P_0}}{\epsilon'}>\epsilon/2\right\}.
    \end{align*}
    Since $\| \xi_n-\xi_0\|_{1,P_0}=\smallo_p(1)$ by Condition~\ref{e RC assumption: Glivenko-Cantelli}, the right-hand side of the above display converges to zero as $n \rightarrow \infty$. Therefore, $|P_0 [I(\xi_n(V(\cdot)) > \eta) - I(\xi_0(V(\cdot)) > \eta)] \Delta^C_0| = \smallo_p(1)$. Recalling \eqref{eq:gammancons}, the above results imply that $\Gamma_n(\eta) - \Gamma_0(\eta)=\smallo_p(1)$ for any $\eta$ that is sufficiently close to $\eta_0$.
    
    Fix $\epsilon>0$. For any $\epsilon$ sufficiently small, the above result and Lemma~\ref{e RC lemma: AL of resource0 and additional resource} imply that $\Gamma_n(\eta_0 - \epsilon) + \alpha \phi_n = \Gamma_0(\eta_0 - \epsilon) + \alpha \phi_0 + \smallo_p(1)$ and $\Gamma_n(\eta_0 + \epsilon) + \alpha \phi_n = \Gamma_0(\eta_0 + \epsilon) + \alpha \phi_0 + \smallo_p(1)$. By Condition~\ref{e RC assumption: continuous weight}, $\Gamma_0(\eta_0 - \epsilon) + \alpha \phi_0 > \kappa > \Gamma_0(\eta_0 + \epsilon) + \alpha \phi_0$ provided $\epsilon$ is sufficiently small. It follows that, with probability tending to one, $\Gamma_n(\eta_0 - \epsilon) + \alpha \phi_n > \kappa > \Gamma_n(\eta_0 + \epsilon) + \alpha \phi_n$, and hence $\eta_0 - \epsilon \leq \eta_n(\kappa) \leq \eta_0 + \epsilon$ by the definition of $\eta_n(\kappa)$. Because $\epsilon$ is arbitrary, it follows that $\eta_n(\kappa) \overset{p}{\rightarrow} \eta_0$.
    
    The case where $\eta_0 = -\infty$ can be proved similarly. If $\kappa=\infty$, then it trivially holds that $\eta_n(\kappa) = -\infty = \eta_0$ for all $n$ and the desired result holds. Otherwise, for any $\eta<0$ for which $|\eta|$ is sufficiently large, a nearly identical argument to that used above shows that $[\Gamma_n(\eta) + \alpha \phi_n] - [\Gamma_0(\eta) + \alpha \phi_0]=\smallo_p(1)$. By Condition~\ref{e RC assumption: continuous weight} and monotonicity of $\Gamma_0$, it follows that $\Gamma_0(\eta) + \alpha \phi_0 < \kappa$, and so, with probability tending to one, $\Gamma_n(\eta) + \alpha \phi_n < \kappa$ and hence $\eta_n(\kappa) \leq \eta$ by the definition of $\eta_n(\kappa)$. Because $\eta$ is arbitrary, we have shown that $\eta_n(\kappa) \overset{p}{\rightarrow} -\infty = \eta_0$.
\end{proof}

\begin{lemma}[Consistency of $\tau_n$ and existence of solution to Eq.~\ref{e RC equation: update quantile} when $\eta_0>-\infty$] \label{e RC lemma: updated tau}
    Assume that the conditions of Theorem~\ref{e RC theorem: asymptotic linearity} hold. The following statements hold:
    \begin{enumerate}[i)]
        \item\label{it:etaknEventually} if $\eta_0 > -\infty$, then, with probability tending to one, a solution $k_n' \in [0,\infty)$ to \eqref{e RC equation: update quantile} exists. Note that we let $k_n=k_n'$ when $\eta_n(\kappa)>0$ and $k_n'$ exists. Hence, if $\eta_0 > 0$, $\eta_n=\eta_n(k_n)=\eta_n(k_n')$ with probability tending to one;
        \item\label{it:etaSaturatesEventually} if a solution $k_n'$ to \eqref{e RC equation: update quantile} exists, then with probability tending to one, $P_0 \{d_{n,k_n'} \Delta^C_n\} + \alpha \phi_0 = \kappa + \bigO_p(n^{-1/2})$;
        \item\label{it:tauEcons} $\tau_n - \tau_0 = \smallo_p(1)$.
    \end{enumerate}
\end{lemma}
We separately prove \ref{it:etaknEventually}, \ref{it:etaSaturatesEventually}, and \ref{it:tauEcons} in the case that $\eta_0>-\infty$, and then we separately prove \ref{it:tauEcons} in the cases where $\eta_0>0$, $\eta_0=0$ and $\eta_0<0$.

\begin{proof}[Proof of \ref{it:etaknEventually} from Lemma~\ref{e RC lemma: updated tau}.]
    Our strategy for showing the existence of a solution to \eqref{e RC equation: update quantile} is as follows. First, we show that the left-hand side of \eqref{e RC equation: update quantile} consistently estimates the treatment resource being used uniformly over rules $\{d_{n,k}: k \in [0,\infty]\}$. Next, we show that the left-hand side of \eqref{e RC equation: update quantile} is a continuous function in $k$ that takes different signs at $k=0$ and $k=\infty$ with probability tending to one.
    
    Define $f_{n,k} : o\mapsto d_{n,k}(v) \left[ \Delta^C_n(w) + \frac{1}{t+\Q^T_n(w)-1} [c - \Q^C_n(t,w)] \right]$. 
    We first show that
    \begin{align}
        \sup_{k\in[0,\infty]}|P_n f_{n,k} - P_0 \{ d_{n,k} \Delta^C_0 \}| =  \bigO_p(n^{-1/2}). \label{eq:fnkallk}
    \end{align}
    We rely on the fact that, for fixed $d_{n,k}$, $P_n f_{n,k}$ is a one-step estimator of $P_0 \{ d_{n,k} \Delta^C_0 \}$.
    Note that
    \begin{align*}
        \sup_{k\in[0,\infty]}\left|P_n f_{n,k} - P_0 \{ d_{n,k} \Delta^C_0 \}\right| &\le \sup_{k\in[0,\infty]}\left|(P_n-P_0) d_{n,k} D_2(\hat{P}_n,\Q^C_n)\right| \\
        &\quad+ \sup_{k\in[0,\infty]}\left|P_0 \left\{ d_{n,k}(V(\cdot)) \frac{\Q^T_n(\cdot) - \Q^T_0(\cdot)}{\Q^T_n(\cdot)} [\Q^C_n(1,\cdot) - \Q^C_0(1,\cdot)] \right\}\right| \\
        &\quad+ \sup_{k\in[0,\infty]}\left|P_0 \left\{ d_{n,k}(V(\cdot)) \frac{\Q^T_n(\cdot) - \Q^T_0(\cdot)}{1-\Q^T_n(\cdot)} [\Q^C_n(0,\cdot) - \Q^C_0(0,\cdot)] \right\}\right|.
    \end{align*}
    Conditions~\ref{e RC assumption: consistency of estimated influence function}~and~\ref{e RC assumption: Donsker condition} along with Lemma~\ref{lemma: equivalence of L2 distance} imply that the first term on the right-hand side is $\bigO_p(n^{-1/2})$. For the second term, we note that the fact that $d_{n,k}(V(w)) \in [0,1]$ for all $w \in \mathcal{W}$ and all $k \in [0,\infty]$ and the Cauchy-Schwarz inequality imply that
    \begin{align*}
        & \sup_{k \in [0,\infty]} \left| P_0 \left\{ d_{n,k}(V(\cdot)) \frac{\Q^T_n(\cdot) - \Q^T_0(\cdot)}{\Q^T_n(\cdot)} [\Q^C_n(1,\cdot) - \Q^C_0(1,\cdot)] \right\} \right| \\
        &\leq P_0 \left| \frac{\Q^T_n(\cdot) - \Q^T_0(\cdot)}{\Q^T_n(\cdot)} [\Q^C_n(1,\cdot) - \Q^C_0(1,\cdot)] \right| \lesssim \| \Q^T_n - \Q^T_0 \|_{2,P_0} \| \Q^C_n(1,\cdot) - \Q^C_0(1,\cdot) \|_{2,P_0}.
    \end{align*}
    Hence, the second term is $\smallo_p(n^{-1/2})$ by Condition~\ref{e RC assumption: remainder}. The third term is also $\smallo_p(n^{-1/2})$ by an almost identical argument. Combining the previous two displays shows that \eqref{eq:fnkallk} holds.
    
    Applying \eqref{eq:fnkallk} at $k=0$ shows that $P_n f_{n,0} + \alpha \phi_n = P_0 \{ d_{n,0} \Delta^C_0 \} + \alpha \phi_0 + \bigO_p(n^{-1/2}) = \alpha \phi_0 + \bigO_p(n^{-1/2})$. Therefore, $P_n f_{n,0} + \alpha \phi_n <\kappa$ with probability tending to one. Applying this result at $k=\infty$ shows that $P_n f_{n,1} + \alpha \phi_n = P_0 \{ d_{n,\infty} \Delta^C_0 \} + \alpha \phi_0 + \bigO_p(n^{-1/2}) = P_0 \Delta^C_0 + \alpha \phi_0 + \bigO_p(n^{-1/2})$. Combining this fact with the fact that $P_0 \Delta^C_0 + \alpha \phi_0 > \kappa$ whenever $\eta_0 > -\infty$ shows that $P_n f_{n,1} + \alpha \phi_n > \kappa$ with probability tending to one. Combining these results at $k=0$ and $k=\infty$ with the fact that $k\mapsto P_n f_{n,k}$ is a continuous function shows that, with probability tending to one, there exists a $k_n' \in [0,\infty)$ such that $P_n f_{n,k_n'}=\kappa - \alpha \phi_n$. Lemma~\ref{e RC lemma: consistency of eta} then implies $\eta_n=\eta_n(k_n)$ with probability tending to 1.
\end{proof}

\begin{proof}[Proof of \ref{it:etaSaturatesEventually} from Lemma~\ref{e RC lemma: updated tau}.].
    By Lemma~\ref{e RC lemma: AL of resource0 and additional resource}, Eq. \ref{eq:fnkallk} and part \ref{it:etaknEventually} of this lemma, we see that $P_0 \{ d_{n,k_n} \Delta^C_0 \} + \alpha \phi_0 = P_n f_{n,k_n} + \alpha \phi_n + \bigO_p(n^{-1/2}) = \kappa + \bigO_p(n^{-1/2})$, as desired.
\end{proof}
\begin{proof}[Proof of \ref{it:tauEcons} from Lemma~\ref{e RC lemma: updated tau} when $\eta_0>0$.]
    In this proof, we use $P_0^n$ to denote a probability statement over the draws of $O_1,\ldots,O_n$. Fix $\epsilon>0$. We will argue by contradiction to show that $P_0^n\left\{\eta_n\ge \eta_0 + \epsilon\right\}\rightarrow 0$ and $P_0^n\left\{\eta_n \leq \eta_0 - \epsilon\right\}\rightarrow 0$ as $n\rightarrow\infty$, implying the consistency of $\eta_n$. The consistency of $\tau_n$ then follows. We study these two events separately. First, we suppose that
    \begin{align}
        \limsup_n P_0^n\left\{\eta_n \geq \eta_0 + \epsilon\right\} > 0. \label{eq:limsupetan}
    \end{align}
    Then there exists $\delta>0$ such that, for all $n$ in an infinite sequence $N\subseteq\mathbb{N}$, the probability $P_0^n\left\{\eta_n \geq \eta_0 + \epsilon\right\}$ is at least $\delta$. Consequently, for any $n\in N$, the following holds with probability at least $\delta$:
    \begin{align}
        P_0 \{ d_{n,k_n} \Delta^C_0 \} + \alpha \phi_0 - \kappa
        &\leq P_0 \{ I(\xi_n > \eta_0 + \epsilon/2) \Delta^C_0 \} + \alpha \phi_0 - \kappa \nonumber \\
        &= P_0 \{ [I(\xi_n > \eta_0 + \epsilon/2) - I(\xi_0 > \eta_0 + \epsilon/2)] \Delta^C_0 \} + \Gamma_0(\eta_0 + \epsilon/2) + \alpha \phi_0 - \kappa. \label{eq:dnkGap}
    \end{align}
    We now show that the first term is $\smallo_p(1)$. For any $x>0$ and $n\in\mathbb{N}$, by Lemma~\ref{lemma: bound for consistency of survival function} and Condition~\ref{e RC assumption: bounded encouragement effect},
    \begin{align*}
        |P_0 \{ [I(\xi_n > \eta_0 + \epsilon/2) - I(\xi_0 > \eta_0 + \epsilon/2)] \Delta^C_0 \}| 
        &\lesssim P_0 I(|\xi_n-\xi_0| > x) + P_0 I(|\xi_0-\eta_n+\epsilon/2| \leq x) \\
        &\leq \frac{\| \xi_n-\xi_0 \|_{1,P_0}}{x} + P_0 I(|\xi_0-\eta_n+\epsilon/2| \leq x).
    \end{align*}
    Similarly to the proof of Lemma~\ref{e RC lemma: consistency of eta}, the fact that $\| \xi_n - \xi_0 \|_{1,P_0} = \smallo_p(1)$ (Condition~\ref{e RC assumption: Glivenko-Cantelli}) ensures that $P_0 \{ [I(\xi_n > \eta_0 + \epsilon/2) - I(\xi_0 > \eta_0 + \epsilon/2)] \Delta^C_0=\smallo_p(1)$. By Condition~\ref{e RC assumption: continuous weight}, $\Gamma_0(\eta_0 + \epsilon/2) - \Gamma_0(\eta_0)$ is a negative constant. Because \eqref{eq:dnkGap} holds with probability at least $\delta>0$ for infinitely many $n$, this shows that $P_0 \{ d_{n,k_n} \Delta^C_0 \} + \alpha \phi_0 - \kappa$ \textit{is not} $\smallo_p(1)$. This contradicts our result from part~\ref{it:etaSaturatesEventually} of this lemma. Therefore, \eqref{eq:limsupetan} is false, that is, $\limsup_n P_0^n\left\{\eta_n \geq \eta_0 + \epsilon\right\} = 0$.
    
    Now we assume that, for some $\epsilon>0$, $\limsup_n P_0^n\left\{\eta_n \leq \eta_0 - \epsilon\right\} > 0$. Then there exists $\delta>0$ such that, for all $n$ in an infinite sequence $N\subseteq\mathbb{N}$, $P_0^n\left\{\eta_n \leq \eta_0 - \epsilon\right\}\ge \delta$. Now, for any $n\in N$, the following holds with probability at least $\delta$:
    \begin{align*}
        P_0 \{ d_{n,k_n} \Delta^C_0 \} + \alpha \phi_0 - \kappa
        &\geq P_0 \{ I(\xi_n > \eta_0 - \epsilon) \Delta^C_0 \} + \alpha \phi_0 - \kappa \nonumber \\
        &= P_0 \{ [I(\xi_n > \eta_0 - \epsilon) - I(\xi_0 > \eta_0 - \epsilon)] \nu_0 \} + \Gamma_0(\eta_0 - \epsilon) + \alpha \phi_0 - \kappa.
    \end{align*}
    The rest of the argument is almost identical to the contradiction argument for the previous event, and is therefore omitted.
    
    Since $\epsilon$ is arbitrary, combining the results of these two contradiction arguments shows that $|\tau_n-\tau_0| \leq |\eta_n-\eta_0|=\smallo_p(1)$, as desired.
\end{proof}

\begin{proof}[Proof of \ref{it:tauEcons} from Lemma~\ref{e RC lemma: updated tau} when $\eta_0=0$.]
    If $\eta_0 =0$, then the construction of $\eta_n$ implies that $\eta_n$ takes values from two sequences: $\eta_n(\kappa)$ and $\eta_n(k_n)$ where $k_n$ is a solution to \eqref{e RC equation: update quantile}. By Lemma~\ref{e RC lemma: consistency of eta}, $\eta_n(\kappa)$ is consistent for $\eta_0$. When a solution to \eqref{e RC equation: update quantile} exists and equals $k_n$, the proof of part~\ref{it:tauEcons} from Lemma~\ref{e RC lemma: updated tau} when $\eta_0>0$ shows that $\eta_n(k_n)$ is consistent for $\eta_0$ and the desired result follows.
\end{proof}

\begin{proof}[Proof of \ref{it:tauEcons} from Lemma~\ref{e RC lemma: updated tau} when $\eta_0<0$.]
    If $\eta_0 < 0$, then by Lemma~\ref{e RC lemma: consistency of eta}, $\eta_n(\kappa) \leq 0$ with probability tending to one. Hence, with probability tending to one, $\tau_n=0=\tau_0$. Therefore, part~\ref{it:tauEcons} holds.
\end{proof}

The following Lemma~\ref{e RC lemma: remainder} show that certain remainders in the expansions in Section~\ref{section: pseudo-gradient expansion} are $\smallo_p(n^{-1/2})$.

\begin{lemma} \label{e RC lemma: remainder}
    Under Conditions~\ref{IV assumption: strong IV positivity},~\ref{e assumption: consistency of strong IV positivity}~and~\ref{e RC assumption: remainder},
    $$\sup_{\edecision: \mathcal{W} \rightarrow [0,1]} \left| R_\edecision(\hat{P}_n,P_0) \right|=\smallo_p(n^{-1/2}).$$
\end{lemma}
\begin{proof}[Proof of Lemma~\ref{e RC lemma: remainder}]
    By the boundedness of the range of $\edecision$, we see that
    \begin{align*}
        &\sup_{\edecision: \mathcal{W} \rightarrow [0,1]} \left| R_\edecision(\hat{P}_n,P_0) \right| \\
        &= \sup_{\edecision: \mathcal{W} \rightarrow [0,1]} P_0 \left| \edecision(\cdot) \left[ \frac{\Q^T_n(\cdot)-\Q^T_0(\cdot)}{\Q^T_n(\cdot)} \{\hat{\Q}^Y_n(1,\cdot) - \Q^Y_0(1,\cdot)\} + \frac{\Q^T_n(\cdot)-\Q^T_0(\cdot)}{1-\Q^T_n} \{\hat{\Q}^Y_n(0,\cdot) - \Q^Y_0(0,\cdot)\} \right] \right| \\
        &\leq P_0 \left| \left[ \frac{\Q^T_n(\cdot)-\Q^T_0(\cdot)}{\Q^T_n(\cdot)} \{\hat{\Q}^Y_n(1,\cdot) - \Q^Y_0(1,\cdot)\} + \frac{\Q^T_n(\cdot)-\Q^T_0(\cdot)}{1-\Q^T_n(\cdot)} \{\hat{\Q}^Y_n(0,\cdot) - \Q^Y_0(0,\cdot)\} \right] \right|. \\
        \intertext{Using Condition~\ref{e assumption: consistency of strong IV positivity} and Lemma~\ref{lemma: equivalence of L2 distance}, the display continues as}
        \lesssim& P_0 \left| (\Q^T_n(\cdot)-\Q^T_0(\cdot)) [\hat{\Q}^Y_n(1,\cdot) - \Q^Y_0(1,\cdot)] \right| + P_0 \left| (\Q^T_n(\cdot)-\Q^T_0(\cdot)) [\hat{\Q}^Y_n(0,\cdot) - \Q^Y_0(0,\cdot)] \right| \\
        \leq& \|\Q^T_n-\Q^T_0\|_{2,P_0} \|\hat{\Q}^Y_n(1,\cdot) - \Q^Y_0(1,\cdot)\|_{2,P_0} + \|\Q^T_n-\Q^T_0\|_{2,P_0} \|\hat{\Q}^Y_n(0,\cdot) - \Q^Y_0(0,\cdot)\|_{2,P_0} \\
        \lesssim& \|\Q^T_n-\Q^T_0\|_{2,P_0} \|\hat{\Q}^Y_n-\Q^Y_0\|_{2,P_0}.
    \end{align*}
    The right-hand side is $\smallo_p(n^{-1/2})$ by Condition~\ref{e RC assumption: remainder}.
\end{proof}

We next prove Theorem~\ref{e RC theorem: asymptotic linearity}.
\begin{proof}[Proof of Theorem~\ref{e RC theorem: asymptotic linearity}]
    By the expansion of $P \mapsto \Psi_{\edecision_P}(P)$ presented in Section~\ref{section: pseudo-gradient expansion},
    \begin{align*}
        \Psi_{\edecision_{n}}&(\hat{P}_n) - \Psi_{\edecision_{0}}(P_0) \\
        &= P_0 D(\hat{P}_n,\edecision_{n},\tau_0,\Q^C_n) + R_{\edecision_{n}}(\hat{P}_n,P_0) + P_0 \{ (\edecision_{n} - \edecision_{0}) (\delta^Y_{0} - \tau_0 \delta^C_{0}) \} \\
        &\quad- \tau_0 P_0 \left\{ \edecision_{n}(\cdot) \frac{\Q^T_n(\cdot) - \Q^T_0(\cdot)}{\Q^T_n(\cdot)} [\Q^C_n(1,\cdot) - \Q^C_0(1,\cdot)] \right\} \\
        &\quad+ \tau_0 P_0 \left\{ (1-\edecision_{n}(\cdot)) \frac{\Q^T_n(\cdot) - \Q^T_0(\cdot)}{1-\Q^T_n(\cdot)} [\Q^C_n(0,\cdot) - \Q^C_0(0,\cdot)] \right\} \\
        &= (P_n - P_0) D(P_0,\edecision_{0},\tau_0,\Q^C_0) - P_n D(\hat{P}_n,\edecision_{n},\tau_0,\Q^C_n) \\
        &\quad+ (P_n - P_0) \left[ D(\hat{P}_n,\edecision_{n},\tau_0,\Q^C_n) - D(P_0,\edecision_{0},\tau_0,\Q^C_0) \right] \\
        &\quad+ R_{\edecision_{n}}(\hat{P}_n,P_0) + P_0 \{ (\edecision_{n} - \edecision_{0}) (\delta^Y_{0} - \tau_0 \delta^C_{0}) \} \\
        &\quad- \tau_0 P_0 \left\{ \edecision_{n}(\cdot) \frac{\Q^T_n(\cdot) - \Q^T_0(\cdot)}{\Q^T_n(\cdot)} [\Q^C_n(1,\cdot) - \Q^C_0(1,\cdot)] \right\} \\
        &\quad+ \tau_0 P_0 \left\{ (1-\edecision_{n}(\cdot)) \frac{\Q^T_n(\cdot) - \Q^T_0(\cdot)}{1-\Q^T_n(\cdot)} [\Q^C_n(0,\cdot) - \Q^C_0(0,\cdot)] \right\}.
    \end{align*}
    Similarly,
    \begin{align*}
        \Psi_{\edecision^\FR}(\hat{P}_n) - \Psi_{\edecision^\FR}(P_0) 
        &= (P_n-P_0) D(P_0,\edecision^\FR,0,\Q^C_0) - P_n D(\hat{P}_n,\edecision^\FR,0,\Q^C_0) \\
        &\quad+ (P_n-P_0) \left[ D(\hat{P}_n,\edecision^\FR,0,\Q^C_0) - D(P_0,\edecision^\FR,0,\Q^C_0) \right] + R_{\edecision^\FR}(\hat{P}_n,P_0); \\
        \Psi_{\edecision^\RD_n}(\hat{P}_n) - \Psi_{\edecision^\RD_0}(P_0) &= (P_n-P_0) D(P_0,\edecision^\RD_0,0,\Q^C_0) - P_n D(\hat{P}_n,\edecision^\RD_n,0,\Q^C_0) \\
        &\quad+ (P_n-P_0) \left[ D(\hat{P}_n,\edecision^\RD_n,0,\Q^C_0) - D(P_0,\edecision^\RD_0,0,\Q^C_0) \right] \\
        &\quad+ R_{\edecision^\RD_n}(\hat{P}_n,P_0) + (\edecision^\RD_n - \edecision^\RD_0) P_0 \Delta^Y_0; \\
        \Psi_{\edecision^\TP_n}(\hat{P}_n) - \Psi_{\edecision^\TP_0}(P_0) &= (P_n-P_0) G_\TP(P_0) - P_n G_\TP(\hat{P}_n) + (P_n-P_0) \left[ G_\TP(\hat{P}_n) - G_\TP(P_0) \right] \\
        &\quad+ R_{\edecision^\TP_0}(\hat{P}_n,P_0).
    \end{align*}
    
    First, we note the following facts, which will be sufficient to ensure that the remainders and empirical process terms in all of the first-order expansions given above are $\smallo_p(n^{-1/2})$. By Condition~\ref{e RC assumption: remainder}, Lemmas~\ref{lemma: equivalence of L2 distance} and \ref{e RC lemma: remainder}, the Cauchy-Schwarz inequality, and boundedness of the range of an ITR, the following terms are all $\smallo_p(n^{-1/2})$:
    \begin{align*}
        & R_{\edecision}(\hat{P}_n,P_0) \ \textnormal{for $\edecision=\edecision_n,\edecision^\FR,\edecision^\RD_n$}, \\
        & P_0 \left\{ \frac{\Q^T_n(\cdot)-\Q^T_0(\cdot)}{1-\Q^T_n(\cdot)} (\hat{\Q}^Y_P(0,\cdot)-\hat{\Q}^Y_0(0,\cdot)) \right\}, \\
        &\tau_0 P_0 \left\{ \edecision_{n}(\cdot) \frac{\Q^T_n(\cdot) - \Q^T_0(\cdot)}{\Q^T_n(\cdot)} [\Q^C_n(1,\cdot) - \Q^C_0(1,\cdot)] \right\},  \\
        &\tau_0 P_0 \left\{ (1-\edecision_{n}(\cdot)) \frac{\Q^T_n(\cdot) - \Q^T_0(\cdot)}{1-\Q^T_n(\cdot)} [\Q^C_n(0,\cdot) - \Q^C_0(0,\cdot)] \right\}. 
    \end{align*}
    Moreover, by Condition~\ref{e RC assumption: IER remainder}, $P_0 \{ (\edecision_{n} - \edecision_{0}) (\delta^Y_{0} - \tau_0 \delta^C_{0}) \} = \smallo_p(n^{-1/2})$; by Conditions~\ref{e RC assumption: consistency of estimated influence function}~and~\ref{e RC assumption: Donsker condition}, $(P_n-P_0) \left[ D_{n,\mathcal{R}} - D_\mathcal{R}(P_0) \right]= \smallo_p(n^{-1/2})$ for all $\mathcal{R} \in \{\FR,\RD,\TP\}$ and
    $$(P_n-P_0) \left\{ [D(\hat{P}_n,\edecision_n,\tau_0,\Q^C_n) - D(\hat{P}_n,\edecision^\RD_n,0,\Q^C_0)] - [D(P_0,\edecision_0,\tau_0,\Q^C_0) - D(P_0,\edecision^\RD_0,0,\Q^C_0)] \right\} = \smallo_p(n^{-1/2}).$$
    Therefore, all relevant remainders and empirical process terms are $\smallo_p(n^{-1/2})$.
    
    We separately study the three cases where $\mathcal{R}=\FR$, $\mathcal{R}=\RD$, and $\mathcal{R}=\TP$. 
    
    \vspace{0.5em}\noindent\textit{Case I: $\mathcal{R}=\FR$.} 
    It holds that
    \begin{align*}
        \psi_n - \psi_0
        &= (P_n-P_0) D_\FR(P_0) - P_n D_{n,\FR} + \smallo_p(n^{-1/2}) \\
        &= (P_n-P_0) D_\FR(P_0) \\
        &\quad+ \tau_0 \Bigg\{ \frac{1}{n}\sum_{i=1}^n \Bigg\{ \edecision_{n}(V_i) \left[ \Delta^C_n(W_i) + \frac{1}{T_i+\Q^T_n(W_i)-1} [C_i - \Q^C_n(T_i,W_i)] \right] \\
        &\hspace{.5in}+ \alpha \left[ \Q^C_n(0,W_i) + \frac{1-T_i}{1-\Q^T_n(W_i)} [C_i - \Q^C_n(0,W_i)] \right] \Bigg\} - \kappa \Bigg\} + \smallo_p(n^{-1/2}),
    \end{align*}
    where the last step follows from the TMLE construction of $\hat{P}_n$ (Step~\ref{ie:targQY} of our estimator), which implies that
    $$\frac{1}{n} \sum_{i=1}^n \left\{ \frac{\edecision_{n}(V_i) - \edecision^\FR(V)}{T_i+\Q^T_n(W_i)-1} [Y_i - \hat{\Q}^Y_n(T_i,W_i)] \right\} = 0.$$
    We now show that the second term on the right-hand side is zero with probability tending to one. If $\tau_0=0$, then this term is zero. Otherwise, $\tau_0 = \eta_0 > 0$. By Lemma~\ref{e RC lemma: updated tau}, the following holds with probability tending to one:
    $$\frac{1}{n}\sum_{i=1}^n \left\{ \edecision_{n}(V_i) \left[ \Q^C_n(1,W_i) + \frac{T_i}{\Q^T_n(W_i)} [C_i - \Q^C_n(1,W_i)] \right] + \alpha \left[ \Q^C_n(0,W_i) + \frac{1-T_i}{1-\Q^T_n(W_i)} [C_i - \Q^C_n(0,W_i)] \right] \right\} = \kappa,$$
    and hence the second term is zero with probability tending to one, as desired. Therefore, $\psi_n - \psi_0 = (P_n-P_0) D_\FR(P_0) + \smallo_p(n^{-1/2})$.
    
    \vspace{0.5em}\noindent\textit{Case II: $\mathcal{R}=\RD$.}
    It holds that
    \begin{align*}
        \psi_n - \psi_0 &= (P_n-P_0) \{D(P_0,\edecision_0,\tau_0,\Q^C_0) - D(P_0,\edecision^\RD_0,0,\Q^C_0)\} \\
        &\quad- P_n \{ D(\hat{P}_n,\edecision_n,\tau_0,\Q^C_n) - D(\hat{P}_n,\edecision^\RD_n,0,\Q^C_0) \} \\
        &\quad- (\edecision^\RD_n-\edecision^\RD_0) P_0 \Delta^Y_0 + \smallo_p(n^{-1/2}),
    \end{align*}
    where we have used $\edecision^\RD_n$ and $\edecision^\RD_0$ to denote the values that the two functions take, respectively. The TMLE construction of $\hat{P}_n$ (Step~\ref{ie:targQY} of our estimator) implies that
    $$\frac{1}{n}\sum_{i=1}^n \frac{\edecision_{n}(V_i) - \edecision^\RD_n(V_i)}{T_i+\Q^T_n(W_i)-1} [Y_i - \hat{\Q}^Y_n(T_i,W_i)] = 0,$$
    and hence
    \begin{align*}
        & P_n \{ D(\hat{P}_n,\edecision_n,\tau_0,\Q^C_n) - D(\hat{P}_n,\edecision^\RD_n,0,\Q^C_0) \} \\
        &= -\tau_0 \Bigg\{ \frac{1}{n}\sum_{i=1}^n \Bigg\{ \edecision_{n}(V_i) \left[ \Delta^C_n(W_i) + \frac{1}{T_i+\Q^T_n(W_i)-1} [C_i - \Q^C_n(T_i,W_i)] \right] \\
        &\hspace{.5in}+ \alpha \left[ \Q^C_n(0,W_i) + \frac{1-T_i}{1-\Q^T_n(W_i)} [C_i - \Q^C_n(0,W_i)] \right] \Bigg\} - \kappa \Bigg\},
    \end{align*}
    which is zero with probability tending to one as proved above. By Condition~\ref{e RC assumption2: active constraint}, Lemma~\ref{e RC lemma: AL of resource0 and additional resource} and the delta method for influence functions, the value that $\edecision^\RD_n$ takes is an asymptotic linear estimator of the value that $\edecision^\RD_0$ takes. Straightforward application of the delta method for influence functions implies that
    $$\psi_n - \psi_0 = (P_n-P_0) D_\RD(P_0) + \smallo_p(n^{-1/2}).$$
    
    \vspace{0.5em}\noindent\textit{Case 3: $\mathcal{R}=\TP$.}
    It holds that
    $$\psi_n - \psi_0 = (P_n-P_0)D_{\TP}(P_0) - P_n D_{n,\TP} + \smallo_p(n^{-1/2}).$$
    The TMLE construction of $\hat{P}_n$ (Step~\ref{ie:targQY} of our estimator) implies that
    $$\frac{1}{n} \sum_{i=1}^n \left\{ \frac{\edecision_{n}(V_i) - \Q^T_n(W_i)}{T_i+\Q^T_n(W_i)-1} [Y_i - \hat{\Q}^Y_n(T_i,W_i)] \right\} = 0,$$
    so
    \begin{align*}
        P_n D_{n,\TP} &= -\tau_0 \Bigg\{ \frac{1}{n}\sum_{i=1}^n \Bigg\{ \edecision_{n}(V_i) \left[ \Delta^C_n(W_i) + \frac{1}{T_i+\Q^T_n(W_i)-1} [C_i - \Q^C_n(T_i,W_i)] \right] \\
        &\hspace{.5in}+ \alpha \left[ \Q^C_n(0,W_i) + \frac{1-T_i}{1-\Q^T_n(W_i)} [C_i - \Q^C_n(0,W_i)] \Bigg\} \right] - \kappa \Bigg\},
    \end{align*}
    which is zero with probability tending to one as proved above. Therefore,
    $$\psi_n - \psi_0 = (P_n-P_0) D_\TP(P_0) + \smallo_p(n^{-1/2}).$$
    
    \vspace{0.5em}\noindent\textit{Conclusion:} The asymptotic linearity result on $\psi_n$ follows from the above results. Consequently, the asymptotic normality result on $\psi_n$ holds by the central limit theorem and Slutsky's theorem.
\end{proof}

\subsection{Proof of Theorem~\ref{theorem: remainder due to estimating optimal rule}} \label{section: IER/ITR remainder proof}

In this section, we prove Theorem~\ref{theorem: remainder due to estimating optimal rule}. The arguments are almost identical to those in Supplement~S9 \citeauthor{Qiu2021} \citep{Qiu2021} with adaptations to the different treatment resource constraint.

\begin{lemma}[Convergence rate of $\tau_n$ if $\eta_0 > -\infty$] \label{e RC lemma: convergence rate of sample quantile}
    Assume that the conditions for Theorem~\ref{e RC theorem: asymptotic linearity} hold. Suppose that $\eta_0 > -\infty$, that the Lebesgue density of the distribution of $\xi_0(V)$ under $V\sim P_0$ is well-defined, nonzero and finite in a neighborhood of and that $P_0 I(\xi_n=\eta_n) = \bigO_p(n^{-1/2})$.
    Under these conditions, the following implications hold with probability tending to one:
    \begin{itemize}
        \item If $\| \xi_n - \xi_0 \|_{q,P_0}=\smallo_p(1)$ for some $0 < q < \infty$, then $|\tau_n - \tau_0| \lesssim \| \xi_n - \xi_0 \|_{q,P_0}^{q/{q+1}} + \bigO_p(n^{-1/2})$.
        \item If $\| \xi_n - \xi_0 \|_{\infty,P_0} = \smallo_p(1)$, then $|\tau_n - \tau_0| \lesssim \| \xi_n - \xi_0 \|_{\infty,P_0} + \bigO_p(n^{-1/2})$.
    \end{itemize}
\end{lemma}
The condition that $P_0 I(\xi_n=\eta_n) = \bigO_p(n^{-1/2})$ is reasonable if $\xi_n(V)$ has a continuous distribution when $V\sim P_0$, in which case $P_0 I(\xi_n=\eta_n) = 0$.

\begin{proof}[Proof of Lemma~\ref{e RC lemma: convergence rate of sample quantile}]
    We study the three cases where $\eta_0>0$, $\eta_0 < 0$ and $\eta_0=0$ separately.
    
    We first study the case where $\eta_0>0$. By Lemma~\ref{e RC lemma: updated tau}, with probability tending to one, $\eta_n = \eta_n(k_n)$ where $k_n$ is a solution to \eqref{e RC equation: update quantile}, and
    \begin{align*}
        P_0 \{ [I(\xi_n > \eta_n) - I(\xi_0 > \eta_0)] \Delta^C_0 \} = P_0 \{ d_{n,k_n} \Delta^C_0 \} - (\kappa - \alpha \phi_0) + \bigO_p(n^{-1/2}) = \bigO_p(n^{-1/2}). 
    \end{align*}
    We argue conditionally on the event that $k_n$ is a solution to \eqref{e RC equation: update quantile}. Adding $\Gamma_0(\eta_n)-P_0\{I(\xi_n > \eta_n)\Delta^C_0\}$ to both sides shows that $\Gamma_0(\eta_n) - \Gamma_0(\eta_0)= -P_0 \{ [I(\xi_n > \eta_n) - I(\xi_0 > \eta_n)] \Delta^C_0 \}  + \bigO_p(n^{-1/2})$. 
    By a Taylor expansion of $\Gamma_0$ under Conditions~\ref{e RC assumption: continuous density},~\ref{e RC assumption: continuous weight}~and~\ref{IV assumption: strong encouragement}, the left-hand side is equal to $-C (\eta_n-\eta_0) + \smallo_p(\eta_n - \eta_0)$ for some $C>0$, yielding that
    \begin{align*}
        [C + \smallo_p(1)][\eta_n-\eta_0]&= P_0 \{ [I(\xi_n > \eta_n) - I(\xi_0 > \eta_n)] \Delta^C_0 \}  + \bigO_p(n^{-1/2}), \nonumber
    \end{align*}
    which immediately implies that
    \begin{align}
        \eta_n-\eta_0&= \bigO_p\Big(P_0 \{ [I(\xi_n > \eta_n) - I(\xi_0 > \eta_n)] \Delta^C_0 \}\Big)  + \bigO_p(n^{-1/2}). \label{eq:etaEmainTaylor}
    \end{align}
    The rest of the proof for this case and the proof for the other two cases are identical to the proof of Lemma~S14 in \citeauthor{Qiu2021} \citep{Qiu2021}.
   We present the argument below for completeness. By Lemma~\ref{lemma: bound for consistency of survival function} and Condition~\ref{e RC assumption: bounded encouragement effect}, for any $\epsilon>0$ it holds that
   \begin{align*}
       &|P_0 \{ [I(\xi_n > \eta_n) - I(\xi_0 > \eta_n)] \Delta^C_0 \}| \\
       &\lesssim |P_0 \{ [I(\xi_n > \eta_n) - I(\xi_0 > \eta_n)] \}| \\
       &\leq P_0 I(|\xi_n-\xi_0|>\epsilon) + P_0 I(|\xi_0 - \eta_n| \leq \epsilon).
   \end{align*}
   Fix a positive sequence $\{\epsilon_n\}_{n=1}^\infty$, where each $\epsilon_n$ may be random through observations $O_1,\ldots,O_n$, such that $\epsilon_n\overset{p}{\rightarrow} 0$ as $n\rightarrow\infty$. By a Taylor expansion of $S_0$, the survival function of the distribution of $\xi_0(V)$ when $V\sim P_0$, around $\eta_0$, which is valid under Condition~\ref{e RC assumption: continuous density} provided $\epsilon_n$ is sufficiently small, it follows that
   \begin{align*}
       |P_0 \{ [I(\xi_n > \eta_n) - I(\xi_0 > \eta_n)] \Delta^C_0 \}|&\lesssim P_0 I(|\xi_n-\xi_0|>\epsilon_n) - 2(S_0)'(\eta_0) \epsilon_n + \smallo_p(\epsilon_n).
   \end{align*}
   Here we recall that $(S_0)'(\eta_0)$ is finite by Condition~\ref{e RC assumption: continuous density}. Returning to \eqref{eq:etaEmainTaylor},
   \begin{align*}
       \eta_n-\eta_0&= \bigO_p\Big(P_0 I(|\xi_n-\xi_0|>\epsilon_n)\Big)  - [2(S_0)'(\eta_0)+\smallo_p(1)] \epsilon_n + \bigO_p(n^{-1/2}).
   \end{align*}
   If $\|\xi_n-\xi_0\|_{q,P_0}=\smallo_p(1)$ for some $0<q<\infty$, by Markov's inequality, $P_0 I(|\xi_n-\xi_0|>\epsilon_n) \leq \|\xi_n-\xi_0\|_{q,P_0}^q/\epsilon_n^q$. In this case, taking $\epsilon_n = \|\xi_n-\xi_0\|_{q,P_0}^{q/(q+1)}$ yields that $|\eta_n-\eta_0| \lesssim \|\xi_n-\xi_0\|_{q,P_0}^{q/(q+1)}+\bigO_p(n^{-1/2})$ with probability tending to one. If $\|\xi_n-\xi_0\|_{\infty,P_0}=\smallo_p(1)$, then taking $\epsilon_n = \|\xi_n-\xi_0\|_{\infty,P_0}$ yields that $P_0 I(|\xi_n-\xi_0|>\epsilon_n)=0$, and hence that $|\eta_n-\eta_0| \lesssim \|\xi_n-\xi_0\|_{\infty,P_0}^2+\bigO_p(n^{-1/2})$ with probability tending to one. The desired result follows by noting that $\tau_0=\eta_0$ and in both cases, $\tau_n=\eta_n(k_n)$ with probability tending to one.
   
   We now study the case where $\eta < 0$. By Lemma~\ref{e RC lemma: consistency of eta}, with probability tending to one, $\eta_n < 0$ and hence $\tau_n=0=\tau_0$, as desired.
   
   We finally study the case where $\eta_0=0$. We argue conditional on the event that a solution $k_n'$ to \eqref{e RC equation: update quantile} exists, which happens with probability tending to one by Lemma~\ref{e RC lemma: updated tau}. Recall that for convenience we let $k_n=k_n'$ when $\eta_n(\kappa)>0$. Then, exactly one of the following two events happen: (i)  $\eta_n(\kappa) \leq 0$ or $\eta_n(k_n') \leq 0$, in which case $\tau_n=0=\tau_0$; (2) $\eta_n(\kappa) > 0$ and $\eta_n(k_n') > 0$, in which case a similar argument as the above proof for the case where $\eta_0>0$ shows that the distance between $\tau_n=\eta_n(k_n')$ and $\tau_0$ has the desired bound. The desired result holds conditional on either event, so it holds unconditional on either event.
\end{proof}

We finally prove Theorem~\ref{theorem: remainder due to estimating optimal rule}.

\begin{proof}[Proof of Theorem~\ref{theorem: remainder due to estimating optimal rule}]
    Observe that
    \begin{align} \label{eq:IER remainder expansion}
        \begin{split}
            | P_0 \{ (\edecision_{n} - \edecision_{0}) (\delta^Y_{0} - \tau_0 \Delta^C_0) \} |
            &\leq  P_0  |\{ I(\xi_n > \tau_n) - I(\xi_0 > \tau_0) \} (\xi_0 - \tau_0) \Delta^C_0|  \\
            &\lesssim  P_0  |\{ I(\xi_n > \tau_n) - I(\xi_0 > \tau_0) \} (\xi_0 - \tau_0) | \\
            &\leq  P_0  |\{ I(\xi_n > \tau_n) - I(\xi_0 > \tau_n) \} (\xi_0 - \tau_n) | \\
            &\quad+  P_0 |\{ I(\xi_0 > \tau_n) - I(\xi_0 > \tau_0) \} (\xi_0 - \tau_0) | \\
            &\quad+  |\tau_n - \tau_0| P_0 | I(\xi_n > \tau_n) - I(\xi_0 > \tau_n) |.
        \end{split}
    \end{align}
    Starting from this inequality, the rest of the proof is identical to that of Theorem~5 in \citeauthor{Qiu2021} \citep{Qiu2021}. We present the argument below for completeness. Let $\{\epsilon_n\}_{n=1}^\infty$ be a positive sequence, where each $\epsilon_n$ is random through the observations $O_1,\ldots,O_n$, such that $\epsilon_n\overset{p}{\rightarrow} 0$ as $n\rightarrow\infty$.
   
   We denote the three terms on the right-hand side by terms 1, 2, and 3, and study these terms separately. It is useful to note that $\tau_n - \tau_0 = \smallo_p(1)$, so the Lebesgue density of the distribution of $\xi_0(V)$, $V\sim P_0$, is finite in a neighborhood of $\tau_n$ with probability tending to one.
   
   \vspace{0.5em}\noindent\textbf{Study of term~1 in \eqref{eq:IER remainder expansion}:}
   Observe that
   \begin{align*}
       P_0  &|\{ I(\xi_n > \tau_n) - I(\xi_0 > \tau_n) \} (\xi_0 - \tau_n) | 
       \\
       &= P_0 | \{ I(\xi_n > \tau_n) - I(\xi_0 > \tau_n) \} (\xi_0 - \tau_n) | I(0 < |\xi_0 - \tau_n|).
   \end{align*}
   First consider the bound with the $L^q(P_0)$-distance. Because $I(\xi_n(v)>\tau_n) \neq I(\xi_0(v)>\tau_n)$ if and only if (i) $\xi_n(v)-\tau_n$ and $\xi_0(v)-\tau_n$ take different signs or (ii) only one of them is zero, this event implies $|\xi_0(v) - \tau_n| \leq |\xi_n(v)-\xi_0(v)|$, and so this term is upper bounded by
   \begin{align*}
       & P_0 | \{ I(\xi_n > \tau_n) - I(\xi_0 > \tau_n) \} (\xi_0 - \tau_n) | I(0 < |\xi_0 - \tau_n| \leq \epsilon_n) \\
       &\quad+ P_0 | \{ I(\xi_n > \tau_n) - I(\xi_0 > \tau_n) \} (\xi_0 - \tau_n) | I( |\xi_0 - \tau_n| > \epsilon_n) \\
       &\leq P_0 | \xi_n - \xi_0 | I(0 < |\xi_0 - \tau_n| \leq \epsilon_n) + P_0 | \xi_n - \xi_0 | I(|\xi_n - \xi_0| > \epsilon_n) \\
       &\leq \| \xi_n - \xi_0 \|_{q,P_0} \left\{ P_0(0 < |\xi_0(V) - \tau_n| \leq \epsilon_n) \right\}^{(q-1)/q} + \frac{P_0 |\xi_n - \xi_0|^q}{\epsilon_n^{q-1}} \\
       &\lesssim \| \xi_n - \xi_0 \|_{q,P_0} \cdot \epsilon_n^{(q-1)/q} + \frac{\| \xi_n - \xi_0 \|_{q,P_0}^q}{\epsilon_n^{q-1}},
   \end{align*}
   where second to last relation holds by H\"{o}lder's inequality and Markov's inequality, and the last relation holds with probability tending to one by the assumption that the distribution of $\xi_0(V)$, $V \sim P_0$, has a continuous finite Lebesgue density in a neighborhood of $\tau_0$ and Lemma~\ref{e RC lemma: updated tau}. Taking $\epsilon_n = \| \xi_n - \xi_0 \|_{q,P_0}^{q/(q+1)}$ yields that $| P_0 \{ I(\xi_n > \tau_n) - I(\xi_0 > \tau_n) \} (\xi_0 - \tau_n) | \lesssim \| \xi_n - \xi_0 \|_{q,P_0}^{2q/(q+1)}$.
   
   Next consider the bound with the $L^\infty(P_0)$-distance. We have that
   \begin{align*}
       P_0 | \{ I(\xi_n > \tau_n) - I(\xi_0 > \tau_n) \} (\xi_0 - \tau_n) | 
       &\leq P_0 I(|\xi_0 - \tau_n| \leq |\xi_n - \xi_0|) | \xi_0 - \tau_n | \\
       &= P_0 I(0 < |\xi_0 - \tau_n| \leq |\xi_n - \xi_0|) | \xi_0 - \tau_n | \\
       &\leq P_0 I(0 < |\xi_0 - \tau_n| \leq \| \xi_n - \xi_0 \|_{\infty,P_0}) | \xi_0 - \tau_n | \\
       &\leq \| \xi_n - \xi_0 \|_{\infty,P_0} P_0(0 < |\xi_0(V) - \tau_n| \leq \| \xi_n - \xi_0 \|_{\infty,P_0}) \\
       &\lesssim \| \xi_n - \xi_0 \|_{\infty,P_0}^2.
   \end{align*}
   Therefore, the first term is upper bounded by both $\| \xi_n - \xi_0 \|_{q,P_0}^{2q/(q+1)}$ and $\| \xi_n - \xi_0 \|_{\infty,P_0}^2$, up to an absolute constant.
   
   \vspace{0.5em}\noindent\textbf{Study of term~2 in \eqref{eq:IER remainder expansion}:}
   Because $I(\xi_0(v) > \tau_n) \neq I(\xi_0(v) > \tau_0)$ if and only if the two indicators take different signs or only one of them is zero, these indicators only take different values if $|\xi_0(v) - \tau_0| \leq |\tau_n - \tau_0|$. Therefore, term~2 bounds as
   \begin{align*}
       P_0 |\{ I(\xi_0 > \tau_n) - I(\xi_0 > \tau_0) \} (\xi_0 - \tau_0) |&\leq P_0 I(|\xi_0 - \tau_0| \leq |\tau_n - \tau_0|) |\xi_0 - \tau_0| \\
       &\leq |\tau_n - \tau_0| P_0 I(|\xi_0 - \tau_0| \leq |\tau_n - \tau_0|) \\
       &\lesssim |\tau_n - \tau_0|^2,
   \end{align*}
   where the last step holds for with probability tending to one by the assumption that the distribution of $\xi_0(V)$, $V \sim P_0$, has a continuous finite Lebesgue density in a neighborhood of $\tau_0$ and Lemma~\ref{e RC lemma: updated tau}. If $\eta_0 > -\infty$, by Lemma~\ref{e RC lemma: convergence rate of sample quantile}, with probability tending to one,
   $$ P_0 | I(\xi_0 > \tau_n) - I(\xi_0 > \tau_0) | |\xi_0 - \tau_0 | \lesssim \begin{cases}
       \| \xi_n - \xi_0 \|_{q,P_0}^{2q/(q+1)} + \bigO_p(n^{-1}), & \text{if } \| \xi_n - \xi_0 \|_{q,P_0} = \smallo_p(1) \\
       \| \xi_n - \xi_0 \|_{\infty,P_0}^2 + \bigO_p(n^{-1}), & \text{if } \| \xi_n - \xi_0 \|_{\infty,P_0} = \smallo_p(1)
   \end{cases}.$$
   Otherwise, by Lemma~\ref{e RC lemma: consistency of eta}, with probability tending to one, $\tau_n=0=\tau_0$ and the above result still holds.
   
   \vspace{0.5em}\noindent\textbf{Study of term~3 in \eqref{eq:IER remainder expansion}:} By Lemma~\ref{lemma: bound for consistency of survival function},
   $$ P_0 | I(\xi_n > \tau_n) - I(\xi_0 > \tau_n)  | \leq P_0 I(|\xi_n - \xi_0| > \epsilon_n) + P_0 I(|\xi_0 - \tau_n| \leq \epsilon_n).$$
   By a Taylor expansion of $S_0$ around $\tau_0$, similarly to the proof of Lemma~\ref{e RC lemma: convergence rate of sample quantile}, with probability tending to one,
   $$P_0 I(|\xi_0 - \tau_n| \leq \epsilon_n) = -2 (S_0)'(\tau_0) \epsilon_n + \smallo_p(|\tau_n - \tau_0 | + \epsilon_n),$$
   where $|(S_0)'(\tau_0)| < \infty$. If $\| \xi_n - \xi_0 \|_{q,P_0} = \smallo_p(1)$ for some $1 < q < \infty$, then $P_0 I(|\xi_n - \xi_0| > \epsilon_n) \leq \| \xi_n - \xi_0 \|_{q,P_0}^q/\epsilon_n^q$. Taking $\epsilon_n = \| \xi_n - \xi_0 \|_{q,P_0}^{q/(q+1)}$ yields that $| P_0 \{ I(\xi_n > \tau_n) - I(\xi_0 > \tau_n) \} | \lesssim \| \xi_n - \xi_0 \|_{q,P_0}^{q/(q+1)}$. If $\| \xi_n - \xi_0 \|_{\infty,P_0} = \smallo_p(1)$, then taking $\epsilon_n = \| \xi_n - \xi_0 \|_{\infty,P_0}$ yields that $| P_0 \{ I(\xi_n > \tau_n) - I(\xi_0 > \tau_n) \} | \lesssim \| \xi_n - \xi_0 \|_{\infty,P_0}$ with probability tending to one. Also note that, by Lemma~\ref{e RC lemma: convergence rate of sample quantile}, if $\eta_0 > -\infty$, then, with probability tending to one,
   $$|\tau_n - \tau_0| \leq |\eta_n - \eta_0| \lesssim \begin{cases}
       \| \xi_n - \xi_0 \|_{q,P_0}^{q/(q+1)} + \bigO_p(n^{-1/2}), & \text{if } \| \xi_n - \xi_0 \|_{q,P_0} = \smallo_p(1), \\
       \| \xi_n - \xi_0 \|_{\infty,P_0} + \bigO_p(n^{-1/2}), & \text{if } \| \xi_n - \xi_0 \|_{\infty,P_0} = \smallo_p(1).
   \end{cases}$$
   The same holds when $\eta_0 = -\infty$ since then $|\tau_n - \tau_0|=0$ with probability tending to one.
   
   Therefore, with probability tending to one,
   \begin{align*}
       | &\tau_n - \tau_0| P_0 | I(\xi_n > \tau_n) - I(\xi_0 > \tau_n) | \\
       &\lesssim \begin{cases}
           \| \xi_n - \xi_0 \|_{q,P_0}^{2q/(q+1)} + \| \xi_n - \xi_0 \|_{q,P_0}^{q/(q+1)} \bigO_p(n^{-1/2}) & \text{ if } \| \xi_n - \xi_0 \|_{q,P_0} = \smallo_p(1), \\
           \| \xi_n - \xi_0 \|_{\infty,P_0}^2 + \| \xi_n - \xi_0 \|_{\infty,P_0} \bigO_p(n^{-1/2}) & \text{ if } \| \xi_n - \xi_0 \|_{\infty,P_0} = \smallo_p(1).
       \end{cases}
   \end{align*}
   
   \vspace{0.5em}\noindent\textbf{Conclusion of the bound in \eqref{eq:IER remainder expansion}:} We finally combine the bounds for all three terms. Note that $a_n \bigO_p(b_n) \lesssim a_n^2 + \bigO_p(b_n^2)$ for any sequence of non-negative random variables $a_n$ and sequence of constants $b_n$. It follows that, with probability tending to one,
   $$| P_0 \{ (\edecision_{n} - \edecision_{0}) (\delta^Y_{0} - \tau_0 \nu_0) \} | \lesssim \begin{cases}
       \| \xi_n - \xi_0 \|_{q,P_0}^{2q/(q+1)} + \bigO_p(n^{-1}), & \text{if } \| \xi_n - \xi_0 \|_{q,P_0} = \smallo_p(1), \\
       \| \xi_n - \xi_0 \|_{\infty,P_0}^2 + \bigO_p(n^{-1}), & \text{if } \| \xi_n - \xi_0 \|_{\infty,P_0} = \smallo_p(1).
   \end{cases}$$
\end{proof}

\section{Additional simulations} \label{section: simulation2}

\subsection{Results of simulation with nuisance functions being truth} \label{section: simulation2 truth}

In this section, we present the results of the simulation with an identical setting as that in Section~\ref{section: simulation} in the main text except that the nuisance functions are taken to be the truth rather than estimated via machine learning. The purpose of this simulation is to show that the performance of our proposed estimator may be significantly improved by using machine learning estimators of nuisance functions that outperform those used in the simulation study reported in the main text.

Table~\ref{table: simulation2 truth} presents the performance of our proposed estimator in this simulation. The Wald CI coverage is close to 95\% for sample sizes of 1000 or more. The coverage of the confidence lower bounds is also close to the nominal coverage of 97.5\%. Therefore, our proposed procedure appears to have the potential to be significantly improved when using improved estimators of nuisance functions. Figure~\ref{figure: CI width2 truth} presents the width of our 95\% Wald CI scaled by the square root of sample size $n$.
For each estimand, the scaled width appears to stabilize as $n$ grows and to be similar to the scaled width observed in the simulation reported in Section~\ref{section: simulation}, where nuisance functions are estimated from data.

\begin{table}[bt]
    \caption{Performance of estimators of average causal effects in the simulation with nuisance functions being the truth.}
    \label{table: simulation2 truth}
    \begin{center}
        \begin{tabular}{lr|r|r|r}
            Performance measure & Sample size & $\FR$ & $\RD$ & $\TP$ \\ \hline \hline
            95\% Wald CI coverage & 500 & $93\%$ & $90\%$ & $90\%$ \\
            & 1000 & $94\%$ & $94\%$ & $93\%$ \\
            & 4000 & $96\%$ & $95\%$ & $95\%$ \\
            & 16000 & $94\%$ & $96\%$ & $95\%$ \\ \hline
            97.5\% confidence lower & 500 & $96\%$ & $96\%$ & $95\%$ \\
            bound coverage & 1000 & $97\%$ & $97\%$ & $96\%$ \\
            & 4000 & $98\%$ & $97\%$ & $97\%$ \\
            & 16000 & $97\%$ & $97\%$ & $97\%$ \\ \hline
            bias & 500 & $0.0023$ & $0.0004$ & $0.0010$ \\
            & 1000 & $0.0013$ & $0.0008$ & $0.0007$ \\
            & 4000 & $0.0003$ & $0.0003$ & $0.0003$ \\
            & 16000 & $0.0002$ & $0.003$ & $0.0002$ \\ \hline
            RMSE & 500 & $0.048$ & $0.023$ & $0.029$ \\
            & 1000 & $0.033$ & $0.015$ & $0.020$ \\
            & 4000 & $0.016$ & $0.008$ & $0.010$ \\
            & 16000 & $0.009$ & $0.004$ & $0.005$ \\
            \hline
            Ratio of mean standard error & 500 & $0.964$ & $0.911$ & $0.928$ \\
            to standard deviation & 1000 & $0.967$ & $0.998$ & $0.958$ \\
            & 4000 & $1.028$ & $0.983$ & $0.995$ \\
            & 16000 & $0.963$ & $1.007$ & $0.996$ \\
        \end{tabular}
    \end{center}
\end{table}

\begin{figure}[bt]
    \begin{center}
        \includegraphics{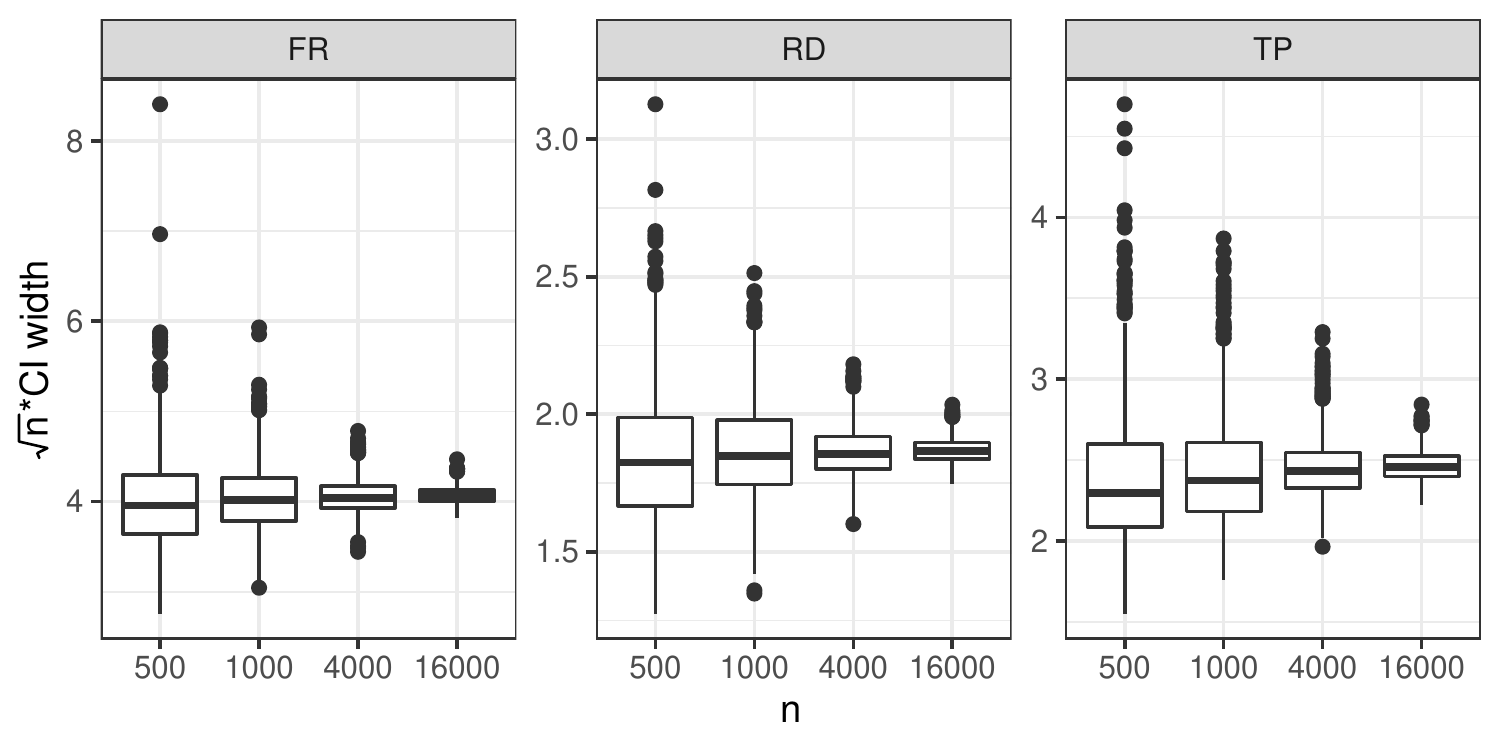}
    \end{center}
    \caption{Boxplot of $\sqrt{n} \times$ CI width for ATE relative to each reference ITR in the simulation with nuisance functions being the truth. \label{figure: CI width2 truth}}
\end{figure}

\subsection{Simulation under a low dimension and a parametric model} \label{section: simulation2 parametric}

In this section, we describe the additional simulation in a setting with a low dimension and a parametric model as well as the simulation results.

The data is generated as follows. We first generate a univariate covariate $W \sim \mathrm{Unif}(-1,1)$. We then generate $T$, $C$ and $Y$ as follows:
\begin{alignat*}{2}
    & T \mid W\  &&\sim\  \text{Bernoulli}\left(\expit(W)\right), \\
    & C \mid T,W\  &&\sim\  \text{Bernoulli}\left(\expit(2T-1+W)\right), \\
    & Y \mid T,W\   &&\sim\ \text{Bernoulli}\left(\expit(1.4T-0.7-0.3W)\right),
\end{alignat*}
where $C$ and $Y$ are independent conditional on $(W,T)$. We set $\edecision^\FR: v \mapsto 0$, $V=W$, and $\kappa=0.35$, which is an active constraint with $\tau_0>0$ and $\edecision^\RD_0 < 1$. We use logistic regression to estimate functions $\Q^T_0$, $\Q^C_0$ and $\Q^C_0$. All other simulation settings are identical to that in Section~\ref{section: simulation}.

The simulation results are presented in Table~\ref{table: simulation2 parametric} and Figure~\ref{figure: CI width2 parametric}. The performance is generally between the nonparametric setting in Section~\ref{section: simulation} and the oracle setting in Section~\ref{section: simulation2 truth}. The CI coverage is much better than the nonparametric case, thus suggesting that our method might perform better with improved estimators of nuisance functions $\Q^T_0$, $\Q^C_0$ and $\Q^C_0$.

\begin{table}[bt]
    \caption{Performance of estimators of average causal effects in the simulation with nuisance functions in a parametric model.}
    \label{table: simulation2 parametric}
    \begin{center}
        \begin{tabular}{lr|r|r|r}
            Performance measure & Sample size & $\FR$ & $\RD$ & $\TP$ \\ \hline \hline
            95\% Wald CI coverage & 500 & $95\%$ & $83\%$ & $96\%$ \\
            & 1000 & $91\%$ & $83\%$ & $93\%$ \\
            & 4000 & $94\%$ & $88\%$ & $93\%$ \\
            & 16000 & $94\%$ & $94\%$ & $95\%$ \\ \hline
            97.5\% confidence lower & 500 & $99\%$ & $99\%$ & $99\%$ \\
            bound coverage & 1000 & $99\%$ & $99\%$ & $99\%$ \\
            & 4000 & $99\%$ & $99\%$ & $98\%$ \\
            & 16000 & $98\%$ & $99\%$ & $99\%$ \\ \hline
            bias & 500 & $-0.0177$ & $-0.0161$ & $-0.0167$ \\
            & 1000 & $-0.0122$ & $-0.0113$ & $-0.0125$ \\
            & 4000 & $-0.0037$ & $-0.0036$ & $-0.0035$ \\
            & 16000 & $-0.0009$ & $-0.0010$ & $-0.0008$ \\ \hline
            RMSE & 500 & $0.035$ & $0.029$ & $0.040$ \\
            & 1000 & $0.026$ & $0.021$ & $0.029$ \\
            & 4000 & $0.012$ & $0.009$ & $0.013$ \\
            & 16000 & $0.006$ & $0.004$ & $0.006$ \\
            \hline
            Ratio of mean standard error & 500 & $1.122$ & $0.908$ & $1.046$ \\
            to standard deviation & 1000 & $0.988$ & $0.857$ & $0.984$ \\
            & 4000 & $0.978$ & $0.891$ & $0.942$ \\
            & 16000 & $0.986$ & $0.979$ & $0.979$ \\
        \end{tabular}
    \end{center}
\end{table}

\begin{figure}[bt]
    \begin{center}
        \includegraphics{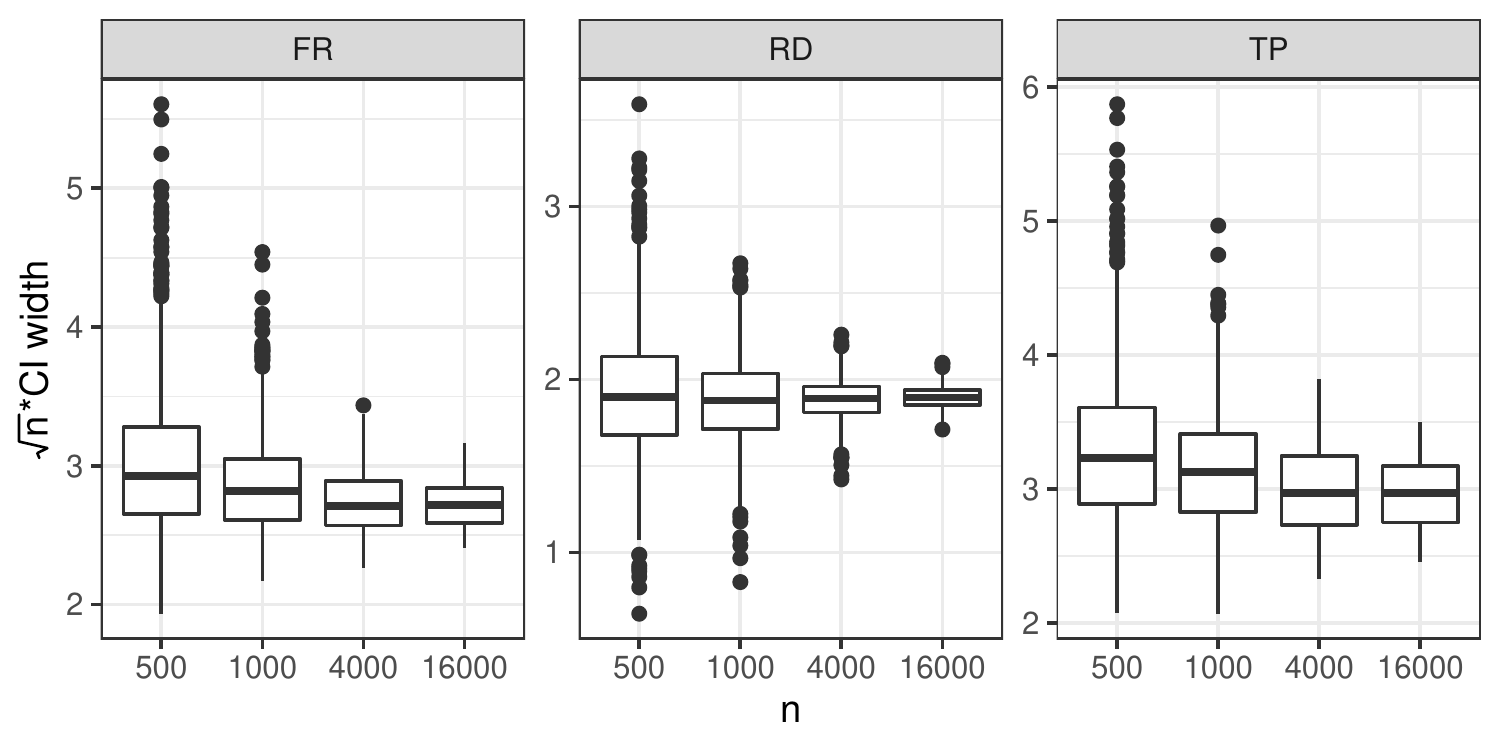}
    \end{center}
    \caption{Boxplot of $\sqrt{n} \times$ CI width for ATE relative to each reference ITR in the simulation with nuisance functions in a parametric model. \label{figure: CI width2 parametric}}
\end{figure}

\end{document}